\documentclass[format=acmsmall, screen]{acmart}
\usepackage{acm-ec-26}
\usepackage{booktabs} %
\usepackage[ruled]{algorithm2e} %

\SetAlFnt{\small}
\SetAlCapFnt{\small}
\SetAlCapNameFnt{\small}
\SetAlCapHSkip{0pt}
\IncMargin{-\parindent}
\allowdisplaybreaks 

\usepackage[textsize=small]{todonotes}
\usepackage{graphicx} %
\usepackage{mathrsfs}
\usepackage{enumitem}
\usepackage{amsthm}
\usepackage[super]{nth}
\usepackage[toc,page]{appendix}
\usepackage{listings}
\usepackage{mathtools}
\usepackage{color}
\usepackage{tabularx}
\usepackage{amsmath}
\usepackage{csquotes}
\usepackage{multicol}
\usepackage{fancyvrb}
\usepackage{makecell}
\usepackage{dsfont}
\usepackage{booktabs}
\usepackage{mdframed}
\usepackage[toc,page]{appendix}
\usepackage{mathdots}
\usepackage{xspace}
\usepackage{empheq}
\usepackage{arydshln}
\usepackage{nicefrac}
\usepackage{dsfont}
\newcommand{\clrstr}{20} 
\newcommand{\strstr}{40} 
\usepackage[nameinlink]{cleveref}

\usepackage{pifont} %
\newcommand{\cmark}{\textcolor{green!30!gray}{\ding{51}}} %
\newcommand{\xmark}{\textcolor{red!40!gray}{\ding{55}}} %
\newcommand{\ifsv}{robustness to fully satisfied voters\xspace}
\newcommand{\Ifsv}{Robustness to fully satisfied voters\xspace}
\newcommand{\IFSV}{Robustness to Fully Satisfied Voters\xspace}
\newcommand{\happygreen}[1]{\textcolor{green!30!gray}{#1}}
\newcommand{\sadred}[1]{\textcolor{red!40!gray}{#1}}
\newcommand{\ejrppar}{\ejrp \cap \mathrm{exPareto}}

\usepackage{array}
\newcolumntype{L}[1]{>{\raggedright\let\newline\\\arraybackslash\hspace{0pt}}m{#1}}
\newcolumntype{C}[1]{>{\centering\let\newline\\\arraybackslash\hspace{0pt}}m{#1}}
\newcolumntype{R}[1]{>{\raggedleft\let\newline\\\arraybackslash\hspace{0pt}}m{#1}}

\renewcommand{\arraystretch}{1.15}

\usepackage{tikz}
\usetikzlibrary{automata, positioning, calc, shapes, arrows, fit}

\usepackage{thm-restate}

\DeclareMathOperator{\pjr}{PJR}
\DeclareMathOperator{\ejr}{EJR}
\DeclareMathOperator{\pjrp}{PJR+}
\DeclareMathOperator{\ejrp}{EJR+}
\DeclareMathOperator{\prl}{NPR}
\DeclareMathOperator{\Core}{Core}

\DeclareMathOperator{\LQ}{LQ}
\theoremstyle{plain}
\newtheorem{theorem}{Theorem}

\newtheorem{corollary}[theorem]{Corollary}
\newtheorem{example}{Example}

\newtheorem{proposition}[theorem]{Proposition}
\newtheorem{definition}{Definition}

\newtheorem{observation}[theorem]{Observation}
\usepackage{wrapfig}
\usepackage{wrapstuff}
\usepackage{etoolbox}
\usepackage{epigraph}

\newcommand{\natWitnessEjrp}{w^{\ejrp}}
\newcommand{\natWitnessPjrp}{w^{\pjrp}}

\setcitestyle{authoryear}

\title{An Axiomatic Analysis of Proportionality Notions in Approval-Based Multiwinner Voting}

\author{Chris Dong}
\affiliation{ \institution{Hasso Plattner Institute, Universität Potsdam} \country{Germany}}
\email{chrisshuyu.dong@hpi.de}

\author{Jannik Peters}
\affiliation{ \institution{Shanghai University of Finance and Economics} \country{China}}
\email{jannikpeters2512@gmail.com}

\begin{abstract}

Even though proportional representation is a fundamental goal in multiwinner voting and a plethora of proportionality notions has been introduced, the normative justifications for choosing one notion over another remain poorly understood. We address this by introducing the axiomatic study of proportionality notions in the approval-based multiwinner voting setting.
That is, we define axioms (or desirable properties) that ``good'' proportionality notions should possess. Using these axioms, we then provide axiomatic characterizations of two prominent recently introduced notions: PJR+ and EJR+ \citep[ACM EC'23]{BrPe23a}.
Our characterization proceeds in two parts. Firstly, we provide a characterization of refinements of PJR+ and EJR+. That is, we define axioms such that any notion satisfying these axioms must imply PJR+ (or EJR+, respectively). In particular, the fundamental axiom distinguishing PJR+ and EJR+ from their predecessors PJR and EJR is the classical axiom of monotonicity. Secondly, we introduce our framework of witness-based proportionality notions, that is, proportionality notions that certify ``misrepresentation'' via a witness set of misrepresented voters. In this class, we provide characterizations of PJR+ and EJR+ as the strongest (assuming certain axioms). Thus, by putting both directions together we obtain exact characterizations of both notions.
Among our results, it may be worth highlighting that any notion satisfying mild conditions (monotonicity, independence of losers, robustness to fully satisfied voters, and lower quota) refines PJR+. In this sense, PJR+ turns out to be the canonical minimal requirement that one may impose on proportionality. 
\end{abstract}

\begin{document}
\maketitle
\setcounter{tocdepth}{1} 
\tableofcontents
\newpage
\section{Introduction}

The \emph{multiwinner voting problem} \citep{FSST17a} considers the selection of a fixed-size set of candidates given the preferences of voters.
An important objective for multiwinner voting is \emph{proportional representation}, i.e., the selected subset of candidates should reflect the voters' preferences in a proportional manner. Indeed, proportionality is a natural way of ensuring fairness in many variations of the multiwinner voting setting.
For instance, civic participation platforms are interested in curating a ``home page'' that represents user viewpoints proportionally on the platform \citep{RML25a}. In participatory budgeting elections, each possible group of voters should be able to spend money proportional to its size \citep{PPS21a}. In many real-world elections, the chosen candidates should proportionally represent the electorate. For instance, Scotland uses the proportional STV voting rule to select councilors in local government elections \citep{BBMP25a}. 
Across all these applications, a key challenge is to find appropriate formalizations of proportionality: when can we argue that the elected candidates sufficiently or insufficiently represent the voters?

We study these questions in a setting with \emph{dichotomous} voter preferences, commonly referred to as \emph{approval-based committee (ABC) voting} \citep{LaSk22a}. Here, each voter approves or disapproves each candidate, and hence their preferences are represented by the subset of candidates that they approve.
Apart from the \emph{apportionment setting}, in which the candidates are partitioned into parties and each voter approves of (only) one party \citep{BaYo01a}, the ABC voting setting is arguably the multiwinner voting setting with the simplest preference structure. Therefore, understanding proportionality in ABC voting is key to gaining an intuition for stronger or more general settings. Indeed, proportionality in ABC voting has recently served as the backbone for developing proportionality guarantees (or proportional voting rules) in various settings. For instance, \citet{BrPe23a} and \citet{KePe23a} showed how to lift proportional fairness from ABC voting to ordinal preferences and proportional clustering, respectively. \citet{PPS21a} generalized both fairness notions and the Method of Equal Shares (MES) voting rule to participatory budgeting, with MES even being used in real-world participatory budgeting elections.

A significant part of the research on proportionality in ABC voting has focused on the development of proportionality notions.\footnote{Also known as proportionality axioms. To avoid confusion, throughout the paper, we will only use the term \emph{proportionality notion} and reserve the term  \emph{axiom} for the (not necessarily proportionality-related) concepts used to characterize proportionality notions.} %
These notions are binary properties that, for a given election, partition the set of committees into two: those that are proportional according to the notion and those that are not. By evaluating voting rules using these notions, one can then judge the quality of the rules. If a rule always outputs committees that satisfy the notion, it is viewed as being proportional. Consequently, proportionality notions are both useful for distinguishing between existing rules and for inspiring the design of new rules.

In turn, many different proportionality notions have been proposed over the last few years \citep{ABC+16a,PeSk20a,KKL25a,CaEl25a,LaNa25a}; some of these notions build on top of each other and create a logical hierarchy, while others observe that previous notions did not capture all facets of proportionality and formalize underexplored aspects.
Most of the proposed notions have the goal of preventing \emph{underrepresentation}: each group of voters should obtain at least as many candidates on the committee as they deserve. %
When does a group of voters deserve any representation? A recurring approach is that of cohesiveness: a group of voters deserves representation proportional to its size and how many candidates it approves unanimously. \citet{ABC+16a} and \citet{SFF+17a} initiated this line of proportionality notions by introducing the extended justified representation (EJR) and proportional justified representation (PJR) notions.

{However, \citet{BrPe23a} criticized both notions for relying on a strong form of cohesiveness: 
a group is eligible for representation by $\ell$ candidates only if the group is sufficiently large and its voters share at least $\ell$ commonly approved candidates.} In both synthetic and real-world instances, the intersections of approval ballots usually turn out to be too small. As a consequence, on many profiles, there are only a few groups of voters eligible for representation and thus PJR and EJR are undemanding in many instances. To address this issue, \citeauthor{BrPe23a} propose variants of these notions, called PJR+ and EJR+. The difference between the notions being that a group of voters can already be eligible for representation by multiple candidates if they unanimously approve a \emph{single} candidate not contained in the committee. Indeed, by focusing on single candidates and thus reducing the amount of cohesiveness required, PJR+ and EJR+ are more discriminating in simulations than their predecessors, and further have the added benefit of being verifiable in polynomial time. Both of these benefits, however, are primarily pragmatic, and have not been sufficient to establish a clear, normative case for only using PJR+ and EJR+.
Notably, there currently is no consensus on whether to use PJR and EJR \citep{DFPS25a,DBW+25a, AKRS24a}, or their strengthened versions PJR+ and EJR+ \citep{DoPe25a,BBC+24a,KSS25a}; with different papers adopting different choices, often without an explicit justification.  
Considering the growing body of proportionality notions, we ask: is there a principled perspective that establishes PJR+ and EJR+ as the canonical way to strengthen PJR and EJR?

\subsection{Our Contribution}
To address this problem, we initiate the study of analyzing proportionality notions via the \emph{axiomatic method}. A central approach of social choice is to axiomatically characterize voting rules. Our main conceptual contribution is to model proportionality notions themselves as (set-valued) voting rules that map each election instance to the set of committees deemed proportional, and then analyze these correspondences using desirable axioms. By taking this axiomatic approach, we obtain normative reasons (from the viewpoint of axiomatic characterizations) for choosing some notions over others. Concretely, our main technical contributions are axiomatic characterizations of PJR+ and EJR+. Put differently, we provide a collection of axioms such that PJR+ or EJR+ is the unique proportionality notion satisfying these, respectively. In particular, we identify sufficient conditions for notions to \emph{refine} PJR+ or EJR+ (that is, every committee deemed proportional by this notion must satisfy PJR+/EJR+) and to \emph{coarsen} them (thus selecting all committees they would select). As the basis of our axiomatic characterization, we use the behavior of the notions on party-list instances. On these instances, in which candidates are partitioned into parties and voters approve of exactly the candidates of a single party, we have a widely accepted ``ground-truth'' for proportionality: every party should receive at least its lower quota, which is its ``fair share'' of seats rounded down (i.e., if a party receives $31\%$ of the voters and there are $10$ seats, it should receive at least $3$ of these seats, provided that it runs at least $3$ candidates).

We begin in \Cref{sec:refinements} by characterizing PJR+ as the coarsest proportionality notion satisfying lower quota for party-lists, together with three technical axioms. Specifically, we distinguish PJR+ from PJR using the classical axiom of \emph{monotonicity}: if a committee is proportional according to a notion, then it should remain proportional if %
candidates on the committee gain additional approvals (while everything else stays the same).
Besides this, our characterization uses two axioms satisfied by nearly all proportionality notions: independence of losers (after candidates outside a proportional committee are deleted, the committee should stay proportional) and \ifsv (if a committee is proportional and a voter restricts their approvals to only approve candidates on the committee, then this committee should stay proportional). The mildness of the involved axioms establishes PJR+ as a canonical, minimal requirement for proportionality. We further provide a similar characterization for EJR+ by additionally imposing an axiom that we call independence of approval swaps. %

In \Cref{sec:Witness_Rationalizability} we introduce the concept of \emph{witness-based proportionality notions}. 
In this framework, we certify underrepresentation via a ``witness'' set of misrepresented voters. This captures the previously mentioned justified representation notions, and, for instance, core stability. In this section, we characterize EJR+ and PJR+ as the \emph{finest} notions among the family of ``cohesiveness-based'' and witness-based notions encompassing both EJR+ and PJR+, but also the coarser EJR and PJR. By combining the results %
of \Cref{sec:refinements,sec:Witness_Rationalizability}, %
we obtain exact characterizations of both notions.

\subsection{Related Work}
As mentioned, proportionality has been extensively studied, and we refer to the book by \citet{LaSk22a} for an introduction to ABC voting and proportionality therein.
\paragraph{Proportionality in ABC Voting} \citet{ABC+16a} initiated the study of proportionality notions in ABC voting by defining justified representation (JR), extended justified representation (EJR), and core stability. \citet{SFF+17a} then proposed proportional justified representation (PJR) as a notion that lies between JR and EJR. 

Closest to our work is that of \citet{BrPe23a}, who strengthened PJR and EJR by proposing their variants PJR+ and EJR+. While PJR and EJR are only demanding when the voters form cohesive groups, \citeauthor{BrPe23a} define PJR+ and EJR+ in a more robust way, where violations can be witnessed by a single unelected candidate. They prove that the new variants are both always satisfiable and can be verified in polynomial time and using simulations, they illustrate that PJR+ and EJR+ are more discriminating than their counterparts. 

Besides this, the list of proportionality notions is continuously expanding. For example, \citet{PeSk20a} consider laminar proportionality and priceability, \citet{PPS21a} and \citet{KKL25a} propose the notions of full justified representation and full (proportional) justified representation, and \citet{CaEl25a} modify existing proportionality notions by changing the quota they are based on from the Hare to the Droop quota. 

\paragraph{Proportionality Beyond ABC Voting} Outside of ABC voting, proportionality has been studied in related multiwinner voting settings, including participatory budgeting \citep{PPS21a,BFL+23a}, clustering \citep{ALM23a,KePe23a}, and social choice with ordinal preferences \citep{AzLe20x,ALP+25a,BBMP25a}. In all three areas, the proportionality notions are explicitly influenced by the ABC voting literature. There are further recent works expanding the ABC voting setting to, for instance, include constraints \citep{MPS23a,MMS23a} or to apply it to temporal settings \citep{CGP24a}. 

\paragraph{Axiomatic Characterizations in Multiwinner Voting} In their book on ABC voting, \citet{LaSk22a} explicitly highlight the lack of axiomatic characterizations of voting rules as a major open question. The few examples of such characterizations include the ones for Thiele and scoring rules \citep{LaSk21a, DoLe24a} and sequential valuation rules \citep{DoLe23a}. Similar characterizations have also been obtained in the ordinal multiwinner voting setting \citep{FPSN19a, SFS19a}. Finally, also in the ordinal setting, we want to highlight the work of  \citet{AzLe22a} which characterizes the proportionality notion of \emph{proportionality for solid coalitions} as exactly the set of outcomes of a particular voting rule. This differs from our characterizations, as they do not postulate desirable axioms for the notion (or then use these axioms to characterize it). 

\paragraph{Axiomatic Characterizations of Solution Concepts} Finally, let us mention that axiomatic characterizations of solution concepts are common in economics (and computation). For instance, in a recent work \citet{BrBr24a} characterized Nash equilibria in normal-form games. Other examples include characterizations of the core in cooperative games \citep{Pele85a} and two-sided matching \citep{SoTo92a}, of stability in matching markets \citep{KaTo24a}, and \citet{Hars55a}'s classical characterization of weighted utilitarianism.

\section{Approval-Based Multiwinner Voting}
In this section, we introduce the approval-based multiwinner voting setting, including the proportionality notions we study in this work.
\subsection{Notation and Definitions}
For $t \in \mathbb{N}^+$, we let $[t] = \{1, \dots, t\}$.
Let $N = [n]$ be a finite set of $n$ voters, $C = \{c_1, \dots, c_m\}$ a finite set of $m$ candidates, and $A=(A_i)_{i\in N}$ the \textit{approval profile} where each voter $i \in N$ submits an \textit{approval ballot} $A_i\subseteq C$. Each approval ballot $A_i$ is the set of candidates \emph{approved} by voter $i$. The goal is to elect a \textit{committee} $W\subseteq C$ of size $k\in \mathbb N^+$. Together, these form an \emph{approval-based multiwinner voting instance} (for short, \emph{instance}) $\mathcal{I} = (N, A, C, k)$. Throughout, we assume that $m \ge k$ for every instance. For a given candidate $c \in C$, we let $N_c \coloneqq \{i \in N \colon c\in A_i\}$ be the \emph{support} of $c$, and each $i\in N_c$ is a \emph{supporter} of $c$. We write $\binom{C}{k}$ for the set of all $k$-element subsets of $C$. 
Let $\mathfrak{I}$ be the set of all (valid) instances. 
A \emph{proportionality notion} is a function $f$, mapping every instance $\mathcal{I} = (N, A, C, k) \in \mathfrak{I}$ to a set of committees $f(\mathcal{I}) \subseteq {\binom{C}{k}}$ that satisfy the notion; i.e., a committee $W$ satisfies $f$ in $\mathcal{I}$ if and only if $W\in f(\mathcal{I})$. If $f(\mathcal{I}) \neq \emptyset$ for every instance $\mathcal{I} \in \mathfrak{I}$, we say that $f$ is \emph{always satisfiable}.

For convenience, when we specify the argument of $f$ as a tuple, we write $f(N,A,C,k)$ instead of $f((N,A,C,k))$.
For two proportionality notions $f$ and $f'$, we say that $f$ is a \emph{refinement} of $f'$ and write $f \subseteq f'$ if $f(\mathcal{I}) \subseteq f'(\mathcal{I})$ for every instance $\mathcal{I} \in \mathfrak{I}$. Equivalently, $f'$ \emph{coarsens} $f$.  We further single out two ``trivial'' proportionality notions.
The \emph{universal notion} $\mathrm{univ}(N,A,C,k) =\binom{C}{k}$ selects all size-$k$ committees, and the \emph{empty notion} $\mathrm{empt}(N,A,C,k) = \emptyset$ selects none.

Given an instance $\mathcal{I}=(N,A,C,k)$ and a subset $C'\subseteq C$ with $|C'|\ge k$,
we write $A_{\mid C'} = (A_i\cap C')_{i\in N}$ for the restriction of $A$ to $C'$ and define $\mathcal{I}_{\mid C'} \coloneqq (N, A_{\mid C'}, C', k)$ to be the instance restricted to $C'$.

We now introduce a structured class of instances.
We say that $\mathcal I=(N,A,C,k)$ is a \emph{party-list instance} if the candidate set $C$ can be partitioned into non-empty \emph{parties} $C_1,\dots,C_p$ (i.e., $C=\biguplus_{x=1}^p C_x$) such that for every voter
$i\in N$ there exists a (unique) $x\in[p]$ with $A_i=C_x$. For each $x\in[p]$, let $N_x \coloneqq \{i\in N \colon A_i=C_x\}$ and $n_x \coloneqq |N_x|$.

\subsection{Proportionality Notions}
Over the years, many proportionality notions have been proposed. For instance, \citet*[Figure~1]{KKL25a} list \emph{nine} of them, while the book by \citet[Figure~4.1]{LaSk22a} lists \emph{ten}. In the main part of our paper, we focus on four of them: the classical notions of extended and proportional justified representation (EJR and PJR) \citep{ABC+16a, SFF+17a} as well as their recent extensions PJR+ and EJR+ \citep{BrPe23a}.
The first such notions in ABC voting were developed by \citet*{ABC+16a}, who introduced core stability, justified representation, and extended justified representation (EJR). These three proportionality notions are all based on a game-theoretic notion of stability.\footnote{For further research on game-theoretic aspects of proportionality notions, see the recent work of \citet{HKMS24a}.} Intuitively, letting $q \coloneqq n/k$ denote the (Hare) quota, a sufficiently cohesive group of at least $q$ voters should be able to secure one seat, and a group of at least $\ell q$ voters should be able to secure $\ell$ seats.
The strongest formalization of this intuition is \emph{core stability}:
no group of at least $\ell n/k$ voters should be able to propose a set $T$ of $\ell$ candidates such that every member of the group strictly prefers $T$ to the elected committee $W$.
Here, we interpret preferences via approval utilities, i.e., voter $i$ strictly prefers $T$ to $W$ if $|T \cap A_i| > |W \cap A_i|$. This notion, however, appears to be very demanding: no voting rule is currently known to guarantee a core-stable committee for every instance, and it remains open whether a core-stable committee always exists \citep{Pete25a, BTC+26a}. In contrast to core stability, the earliest proportionality notions that are achievable by known rules do not aim to represent every voter group, but only groups that are sufficiently large and sufficiently cohesive.
\begin{definition}[$\ell$-Cohesive Group]
    Given an instance $\mathcal{I} = (N, A, C, k)$ and $\ell \in [k]$, a group $N'\subseteq N$ of voters is \emph{$\ell$-large} if $\lvert N'\rvert \ge \ell \frac{n}{k}$.
    We additionally say that $N'$ is \emph{$\ell$-cohesive} if
    \(
    \left\lvert \bigcap_{i \in N'} A_i \right\rvert \ge \ell.
    \)
\end{definition}
Intuitively, $\ell$-cohesive groups can be viewed as ``sub-parties'' (though they may overlap). The justified-representation family of notions requires that these groups are adequately represented in the chosen committee, with the two main notions being proportional justified representation (PJR) \citep{SFF+17a} and extended justified representation (EJR) \citep{ABC+16a}.
\begin{definition}[PJR and EJR]
    Given an instance $\mathcal{I} = (N, A, C, k)$ we say that a committee $W$ satisfies: 
    \begin{itemize}
        \item \emph{proportional justified representation (PJR)} if for every $\ell \in [k]$ and $\ell$-large and $\ell$-cohesive group $N'$ it holds that $\left\lvert \bigcup_{i \in N'} (A_i \cap W) \right\rvert\ge \ell$ \citep{SFF+17a};
        \item \emph{extended justified representation (EJR)} if for every $\ell \in [k]$ and $\ell$-large and $\ell$-cohesive group $N'$ there exists an $i \in N'$ with $\left\lvert A_i \cap W\right\rvert \ge \ell$ \citep{ABC+16a}.
    \end{itemize}
    We refer to the set of all committees satisfying PJR (respectively EJR) in the given instance as $\pjr(\mathcal{I})$ (respectively $\ejr(\mathcal{I})$).
\end{definition}
PJR requires the group to \emph{collectively} obtain $\ell$ approved winners, whereas EJR requires that some \emph{individual} in the group obtains $\ell$ approved winners.
By definition, any committee satisfying EJR also satisfies PJR.  Neither PJR nor EJR has been without criticism. In particular, \citet{BrPe23a} criticized the reliance on cohesive groups, which makes both PJR and EJR less demanding to satisfy, yet computationally hard to verify. Further, they criticized the lack of robustness of both proportionality notions: their guarantees are significantly weaker if groups are almost cohesive, yet do not fully meet the strong threshold of all approving the same $\ell$ candidates.
Motivated by these concerns, \citeauthor{BrPe23a} proposed variants that remove the requirement that the witnessing groups be $\ell$-cohesive and instead only require agreement on a single unselected candidate.
\begin{definition}[PJR+ and EJR+]
    Given an instance $\mathcal{I} = (N, A, C, k)$ we say that a committee $W$ satisfies: 
    \begin{itemize}
        \item \emph{proportional justified representation+ (PJR+)} if for every $\ell \in [k]$, $c \notin W$, and $\ell$-large group $N' \subseteq N_c$ it holds that $\left\lvert \bigcup_{i \in N'} (A_i \cap W) \right\rvert\ge \ell$ \citep{BrPe23a};\footnote{As pointed out by \citeauthor{BrPe23a}, PJR+ is equivalent to the IPSC notion of \citet{AzLe21a} restricted to approval-based multiwinner voting.}
        \item \emph{extended justified representation+ (EJR+)} if for every $\ell \in [k]$, $c \notin W$, and $\ell$-large group $N' \subseteq N_c$ there exists an $i \in N'$ with $\left\lvert A_i \cap W \right\rvert\ge \ell$ \citep{BrPe23a}.
    \end{itemize}
    We refer to the set of all committees satisfying PJR+ (respectively EJR+) in the given instance as $\pjrp(\mathcal{I})$ (respectively $\ejrp(\mathcal{I})$).
\end{definition}
Indeed, both PJR+ and EJR+ can be verified in polynomial time. Further, \citeauthor{BrPe23a} argue that both notions are seemingly more robust than PJR and EJR and experimentally observe that in synthetic data both are rarer than their ``cohesive predecessors''. We note that EJR+ refines both EJR and PJR+, while EJR and PJR+ are incomparable (i.e., there exist committees satisfying EJR but not PJR+ and committees satisfying PJR+ but not EJR). Further, committees satisfying EJR+ (and hence also EJR, PJR, and PJR+) always exist, for instance, those chosen by the Method of Equal Shares \citep{PeSk20a, BrPe23a}.

\section{Characterizing Refinements of PJR+ and EJR+} 
\label{sec:refinements}
The goal of this section is twofold: (i) we introduce the general concept of studying proportionality notions from an axiomatic perspective, and (ii) we provide an axiomatic distinction between PJR/EJR and their ``$+$'' variants, thereby formalizing the observation of \citeauthor{BrPe23a} that PJR and EJR are non-robust. For this, we introduce desirable, abstract properties (axioms) that reasonable proportionality notions should satisfy. Through this lens, we analyze in what sense the previously defined proportionality notions are reasonable and which weaknesses they may have. Using these axioms, we characterize PJR+ and EJR+ as the \emph{coarsest} (with regard to $\subseteq$, i.e., with the largest set of admissible committees) proportionality notions satisfying certain desirable properties. Equivalently, every proportionality notion satisfying these properties is a refinement of PJR+ (respectively, EJR+). Our key distinguishing axiom is monotonicity: if a committee $W$ satisfies the notion (and hence is proportional according to it), then it should continue to satisfy the notion after some voters additionally approve members of $W$. As we show in this section, neither PJR nor EJR satisfies this axiom. Omitted proofs appear in \Cref{app:missing}.

\subsection{Axioms for Proportionality Notions}

As our starting point, we adopt the axiom of lower quota for party-list instances \citep{BLS18a}, arguably one of the weakest possible proportionality requirements. As mentioned in the preliminaries, so-called party-list instances only allow for very restricted profiles: the candidates are partitioned into parties and each voter approves all candidates of exactly one party. The class of party-list instances has two main advantages: (i) it allows us to represent real-world parliamentary apportionment elections and lets us leverage the well-established theory of proportional representation for apportionment \citep{BaYo01a}, and (ii) for party-list instances, it is significantly easier to agree on which committees are proportional than in the full approval-based multiwinner setting. The behavior of voting rules on party-list instances was also used by \citet{LaSk21a} and \citet{DoLe24a} in their characterizations of voting rules in ABC voting.

In particular, lower quota for party-list instances requires that from a party $C_x \subseteq C$ approved by the set of voters $N_x \subseteq N$ we select at least $\min\left(\lvert C_x\rvert, \lfloor k\frac{n_x}{n}\rfloor\right)$ candidates from $C_x$, i.e., the party's lower quota, capped by party size. For instance, in an election with $k = 50$, if a party fielding $20$ candidates receives a third of the votes, it should get at least $\min(\lfloor\frac{50}{3}\rfloor, 20) = 16$ seats. In particular, for a given party-list instance we denote by $\LQ(\mathcal{I})$ the set of all committees $W$ for which $\lvert W \cap C_x\rvert \ge \min\left(\lfloor k \frac {n_x}{n}\rfloor, \lvert C_x\rvert\right)$ for all parties $C_x$.
\begin{definition}
A proportionality notion $f$ satisfies \emph{lower quota (for party-lists)} if $f(\mathcal{I}) \subseteq \LQ(\mathcal{I})$ for every party-list instance $\mathcal I$.
\end{definition}
We take lower quota on party-list instances as our baseline proportionality requirement, following an approach suggested by \citet[p.~47]{LaSk22a}:
``One approach to reasoning about proportionality of voting rules is to first identify a class of well-structured preference profiles where the concept of proportionality can be intuitively captured.'' 
They then use party-list instances to compare voting rules with apportionment methods.
In contrast, we use the apportionment setting as a foundation for reasoning about proportionality notions on arbitrary instances.
On party-list instances, all previously defined notions (including core stability) coincide with lower quota.\footnote{For this proposition, we write $\Core(\mathcal I)$ for the set of core-stable committees in $\mathcal I$.}
\begin{restatable}{proposition}{lqpleq}
    Let $\mathcal{I} = (N,A,C, k)$ be a party-list instance. Then 
    \(
    \LQ(\mathcal{I}) = \pjr(\mathcal{I}) = \pjrp(\mathcal{I}) = \ejr(\mathcal{I}) = \ejrp(\mathcal{I}) = \Core(\mathcal{I}).\)
\label{prop:lq}
\end{restatable}
In light of \Cref{prop:lq}, we further note that the lower quota (for party-lists) axiom can also be strengthened: instead of requiring that on party-list instances all committees satisfying the proportionality notion adhere to the lower quota, one could also require that the proportionality notion chooses \emph{precisely} the committees satisfying lower quota. By \Cref{prop:lq}, all notions introduced above satisfy this. %
However, other proportionality notions studied in the literature, for instance, priceability \citep{PeSk20a}, are not lower quota extensions (see \Cref{sec:compliance}).
\begin{definition}[Lower Quota Extension]
    A proportionality notion $f$ is a \emph{lower quota extension} if for every party-list instance $\mathcal{I}$ it holds that $f(\mathcal{I}) = \LQ(\mathcal{I})$.
\end{definition}
\begin{wrapstuff}[r,type=figure,width=5cm]
    \centering
    \begin{tikzpicture}
    [yscale=0.6,xscale=0.8,
    voter/.style={anchor=south}]
    
        \foreach \i in {1,...,4}
    		\node[voter] at (\i-0.5, -1) {$\i$};
        
        \draw[fill=teal!\clrstr] (0, 0) rectangle (2, 1);
        \draw[fill=teal!\clrstr] (0, 1) rectangle (2, 2);
        \draw[fill=teal!\clrstr] (0, 2) rectangle (1, 3);
        \draw[fill=magenta!\strstr, ultra thick] (2, 0) rectangle (4, 1);
        \draw[fill=magenta!\strstr, ultra thick] (2, 1) rectangle (4, 2);
        \draw[fill=magenta!\strstr, ultra thick] (2, 2) rectangle (4, 3);
        \draw[fill=magenta!\strstr, ultra thick] (2, 3) rectangle (4, 4);
        
        \node at ( 1, 0.5) {$c_{1}$};
        \node at ( 1, 1.5) {$c_{2}$};
        \node at ( 0.5, 2.5) {$c_{3}$};
        \node at ( 3, 0.5) {$c_{4}$};
        \node at ( 3, 1.5) {$c_{5}$};
        \node at ( 3, 2.5) {$c_{6}$};
        \node at ( 3, 3.5) {$c_{7}$};
    \end{tikzpicture}
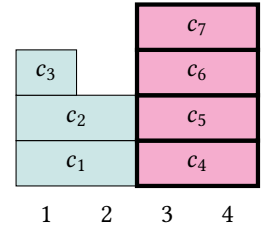
    \caption{Example for Lower Quota and Independence of Losers.}
    \label{fig:one_add}
\end{wrapstuff}
In general, however, lower quota (and lower quota extension) only constrains behavior on party-list instances. Hence, a proportionality notion may satisfy them and still output arbitrary (even unreasonable) committees on non-party-list instances. To illustrate this point, consider the instance depicted in \Cref{fig:one_add}. In this and subsequent examples, the voters are depicted as indices on the x-axis, and each candidate is represented by a bar placed above the voters approving it. Concretely, the instance has $n=4$ with ballots $A_1=\{c_1,c_2,c_3\}$, $A_2=\{c_1,c_2\}$, and $A_3=A_4=\{c_4,c_5,c_6,c_7\}$.
 The candidates outlined in bold indicate a chosen committee. We note that the instance is not a party-list instance because of candidate $c_3$. Hence, a notion could select the bolded (and severely non-proportional) committee $\{c_4, \dots, c_7\}$ for $k = 4$ without violating lower quota for party-lists.
\wrapstuffclear
We therefore introduce further desirable properties that allow us to connect our choices on party-list instances to the choices on the full domain. For convenience, we provide a list of introduced axioms in \Cref{tab:roadmap-axioms}, as well as which proportionality notions among PJR, EJR, PJR+, and EJR+ satisfy them.
First, we define the axiom of \emph{independence of losers}: if a committee satisfies a proportionality notion, it should still satisfy this notion, even if unselected candidates are removed from the instance, as intuitively removing these candidates should not make the chosen committee less proportional.
\wrapstuffclear
\begin{definition}
    A proportionality notion $f$ satisfies \emph{independence of losers} if for every instance $\mathcal{I} = (N, A, C, k)$, every committee $W \in f(\mathcal{I})$ and $c \in C \setminus W$ it holds that 
    \(
    W \in f(\mathcal{I}_{\mid C\setminus \{c\}}).
    \)
\end{definition}
We note that by iterating, the axiom also applies to removing any subset of unselected candidates.
Independence of losers is reminiscent of the contraction axiom $\alpha$ \citep{Cher54a}. Contraction states that if an outcome is chosen, reducing the feasible set should not change that as long as the outcome remains feasible.\footnote{We study a strengthening of independence of losers in \Cref{sec:strongIoL}. The strengthened version also entails a condition reminiscent of the expansion axiom $\gamma$ \citep{Sen70a}, which is the natural counterpart of $\alpha$. Interestingly, this strengthening characterizes proportionality notions that justify violations via single candidates.} A variant of independence of losers has also been studied by \citet{DoLe23a, DoLe24a} in their characterizations of (sequential) Thiele rules. Returning to the instance depicted in~\Cref{fig:one_add}, the committee $\{c_4, \dots, c_7\}$ would not be proportional for any proportionality notion satisfying lower quota for party-lists and independence of losers, as after removing $c_3$ the remaining instance is a party-list instance, in which both $c_1$ and $c_2$ need to be selected for the party $\{c_1, c_2\}$, as $\min(2, \lfloor 4\frac{2}{4}\rfloor) = 2$ candidates need to be selected from it.

Second, we consider how fully satisfied voters affect proportionality. In this paper, we focus on notions aimed at preventing \emph{underrepresentation}, i.e., every group of voters should obtain at least as many representatives as it deserves. Through this lens, a voter for whom every approved candidate is contained in the committee is as well-represented as possible. Therefore, such voters should never be eligible for claiming a proportionality violation. Taking this thought further, if  a committee $W$ is acceptable according to a proportionality notion (and thus does not cause underrepresentation), and some voters restrict their ballots so that they approve only candidates in $W$, then these voters become irrelevant for underrepresentation: they cannot be part of a group witnessing that the committee fails to represent them. Hence, after this change $W$ should still be acceptable according to the proportionality notion.

\begin{definition}
    A proportionality notion $f$ satisfies %
    \emph{robustness to fully satisfied voters} if for every instance $\mathcal{I} = (N, A, C, k)$, every committee $W \in f(\mathcal{I})$, and each approval profile $A'$ such that for every $i\in N$ we have either $A'_i = A_i$ or $A'_i \subseteq A_i \cap W$, it holds that
    \(
    W \in f(N, A', C, k).
    \)
\end{definition}
This axiom is satisfied by all the notions we have studied so far: if a voter changes their approval ballot to be a subset of the committee committee members they previously approved, then this voter can no longer be part of any ``violating group'' for the given notion.

Third, we turn our attention to a classical axiom that will be crucial in distinguishing PJR and EJR from PJR+ and EJR+: monotonicity. To motivate this axiom, consider a committee $W$ which does not cause underrepresentation. If some voters add approvals for candidates in $W$ (and only for candidates in $W$) to their ballots, then these voters are better represented than before. Intuitively, if a voter is already well-represented by a committee, then making this voter even ``more represented'' should not lead to this committee becoming non-proportional. Therefore, $W$ should still not cause underrepresentation after this change of the profile.

\begin{definition}
    A proportionality notion $f$ satisfies \emph{monotonicity} if for every instance $\mathcal{I} = (N, A, C, k)$, every committee $W \in f(\mathcal{I})$ and approval profile $A' \supseteq A$ (i.e., $A'_i \supseteq A_i$ for every $i \in N$) with $A'_i \setminus A_i \subseteq W$ for all $i \in N$ it holds that
    \(
    W \in f(N, A', C, k).
    \)
\end{definition}
Monotonicity is indeed satisfied by PJR+ and EJR+.\footnote{Our definition of monotonicity can be seen as a stronger version of the ``candidate monotonicity without additional voters'' axiom of \citet{SaFi17a}. This axiom would require that for a given proportionality notion $f$ and committee $W \in f(\mathcal{I})$ if an approval is added to a candidate $c\in W$ there must exist \emph{some} committee $W'$ which is proportional according to $f$ in the new instance and contains $c$. This axiom would indeed be satisfied by all our proportionality notions, as for all of them, and any candidate approved by at least one voter, there always exists a proportional committee containing this candidate \citep{BFNK19a, DFPS25a}.} For instance, for EJR+, after adding approvals to candidates in $W$, for any $c \notin W$, the set $N_c$ consisting of the supporters of $c$ remains unchanged, while $\lvert A_i \cap W\rvert$ can only increase for each $i\in N_c$. Hence, no new EJR+ violations can occur. The same reasoning also applies to PJR+. However, both PJR and EJR can fail monotonicity, as we demonstrate by the following example.
\begin{example}
\label{exp:pjr_mon}
    \begin{wrapstuff}[r,type=figure,width=5cm]
    \centering
    \begin{tikzpicture}
    [yscale=0.6,xscale=0.8,
    voter/.style={anchor=south}]
    
        \foreach \i in {1,...,4}
    		\node[voter] at (\i-0.5, -1) {$\i$};
        
        \draw[fill=teal!\strstr, ultra thick] (0, 0) rectangle (2, 1);
        \draw[fill=teal!\clrstr] (0, 1) rectangle (2, 2);
        \draw[fill=teal!\strstr, ultra thick] (0, 2) rectangle (1, 3);
        \draw[fill=magenta!\strstr, ultra thick] (2, 0) rectangle (4, 1);
        \draw[fill=magenta!\strstr, ultra thick] (2, 1) rectangle (4, 2);
        \draw[fill=magenta!\strstr, ultra thick] (2, 2) rectangle (4, 3);
        \draw[fill=magenta!\strstr, ultra thick] (2, 3) rectangle (4, 4);
        
        \node at ( 1, 0.5) {$c_{1}$};
        \node at ( 1, 1.5) {$c_{2}$};
        \node at ( 0.5, 2.5) {$c_{3}$};
        \node at ( 3, 0.5) {$c_{4}$};
        \node at ( 3, 1.5) {$c_{5}$};
        \node at ( 3, 2.5) {$c_{6}$};
        \node at ( 3, 3.5) {$c_{7}$};
    \end{tikzpicture}
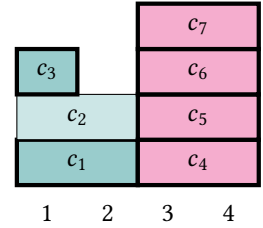
    \caption{Example for PJR, EJR, and monotonicity.}
    \label{fig:pjr_mon}
\end{wrapstuff}
To see that PJR and EJR do not satisfy monotonicity, consider the instance depicted in \Cref{fig:pjr_mon} with $k = 6$. The depicted committee (shown in bold) $\{c_1, c_3, c_4, c_5, c_6, c_7\}$ satisfies PJR and EJR for $k = 6$ as the group of voters $\{1,2\}$ is only $2$-cohesive, and voter $1$ approves two candidates in the outcome, while the group $\{3,4\}$ is fully satisfied. However, if voter $2$ additionally approves candidate $c_3$, the committee no longer satisfies PJR or EJR, as now $\{1,2\}$ are  $3$-large and $3$-cohesive, but only receive two candidates in the outcome.

We note that already in the original instance the bolded outcome $\{c_1, c_3, c_4, c_5, c_6, c_7\}$ neither satisfies PJR+ nor EJR+ as witnessed by candidate $c_2$ together with voters $1$ and $2$ and $\ell = 3$.
\end{example}

Intuitively, there are two opposing approaches to representation: one which views the group as a whole, and one which considers each voter within the group individually.  We will see soon that these approaches distinguish PJR+ from EJR+, as PJR+ focuses on the representation allotted to groups, while EJR+ focuses on single voters. For now, we focus on the latter approach. Since we only have dichotomous preferences at hand, without further information, the utility of a single voter depends only on the \emph{number} of winning candidates they approve, as opposed to their identities. We formalize this intuition via the following axiom.
\begin{definition}\label{def:independence_of_approval_swaps}
    A proportionality notion $f$ satisfies \emph{independence of approval swaps} if for every instance $\mathcal{I} = (N, A, C, k)$, every committee $W \in f(\mathcal{I})$ and approval profile $A'$ such that $A'_i\setminus W = A_i\setminus W$ and $\lvert A'_i \cap W\rvert = \lvert A_i \cap W\rvert$ for all $i \in N$, it holds that
    \(
    W \in f(N, A', C, k).
    \)
\end{definition}
This axiom requires that if a committee $W$ is proportional according to $f$, and some voters change their approval ballot by approving a different set of candidates in $W$ of the same cardinality as before, then $W$ remains proportional according to $f$. The fact that EJR+ satisfies the axiom is part of the reason why EJR+ can be verified in polynomial time: the identity of the candidates inside the committee approved by some voter does not matter, the only thing that matters about the committee is how many candidates this voter approves. 

\begin{table}[t]
  \small
  \centering
  \setlength{\tabcolsep}{4.5pt}
  \begin{tabular}{lllll}
    \toprule
      & \textbf{PJR} & \textbf{EJR} & \textbf{PJR+} & \textbf{EJR+} \\
    \midrule
    Lower quota for party-lists
      & \cmark & \cmark & \cmark$^\star$ & \cmark$^\star$ \\
    Independence of losers
      & \cmark & \cmark & \cmark$^\star$ & \cmark$^\star$ \\
    \Ifsv
      & \cmark & \cmark & \cmark$^\star$ & \cmark$^\star$ \\
    Monotonicity
      & \xmark & \xmark & \cmark$^\star$ & \cmark$^\star$ \\
    Independence of approval swaps
      & \xmark & \xmark & \xmark & \cmark$^\star$ \\
    \bottomrule
  \end{tabular}
  \caption{Overview of the axioms used in our characterizations.
  Monotonicity separates the classical notions (\textsc{PJR}, \textsc{EJR}) from their robust variants (\textsc{PJR+}, \textsc{EJR+});
  adding independence of approval swaps separates \textsc{EJR+} from \textsc{PJR+}. A $^\star$ indicates that this axiom is used in the characterization of the respective notion. }
  \label{tab:roadmap-axioms}
\end{table}

\subsection{Characterizing Refinements of PJR+}
Using the first four axioms, we can now characterize PJR+ as the \emph{coarsest} proportionality notion satisfying them. That is, we show that any proportionality notion satisfying independence of losers, monotonicity, \ifsv, and lower quota for party-lists selects only committees satisfying PJR+.
We further show that PJR+ indeed satisfies all four axioms. As a consequence, as long as one takes these four axioms to be granted, %
PJR+ is the least restrictive (or coarsest) acceptable proportionality notion consistent with these axioms.
\wrapstuffclear  
\begin{theorem}
   For any proportionality notion $f$ satisfying independence of losers, monotonicity, \ifsv, and lower quota for party-lists, it holds that $f \subseteq \pjrp$. Further, $\pjrp$ satisfies all four axioms.
   \label{thm:PJRP_Refinement}
\end{theorem}
\begin{proof}
    \textbf{Part 1:}
    Let $f$ be a proportionality notion satisfying the four axioms, and consider an instance $\mathcal{I} = (N, A, C, k)$ and a committee $W \in f(\mathcal{I})$.
    Our goal is to show that $W$ also satisfies PJR+ in the instance $\mathcal{I}$.
    Fix $\ell \in [k]$ and let $c \in C\setminus W$ be any candidate outside of $W$. Consider an $\ell$-large group $N'\subseteq N_c$  of voters. It remains to show that $\lvert \bigcup_{i \in N'} (A_i \cap W) \rvert \ge \ell$. To this end, we will now modify the approval profile $A$ using independence of losers, monotonicity, and robustness to fully satisfied voters in such a way that we can apply the lower quota for party-lists axiom.

    We abbreviate $S \coloneqq \bigcup_{i \in N'} (A_i \cap W)$. Our goal is to show that $S$ is of cardinality at least $\ell$.
    First, if $W = S$ we are already done, as $\ell$ can be at most $k$. Hence, we can assume that $S\subset W$. To bring more structure into the approval profile, we restrict our attention to the candidate set $C' \coloneqq W \cup \{c\}$. 
    That is, we consider the instance $\mathcal{I}_1= (N, A_{\mid C'},C', k )$. By starting from $\mathcal I$ and iteratively removing candidates $d\in C\setminus C'$ from consideration, we eventually reach $\mathcal I_1$. Since $f$ satisfies independence of losers, we can apply this iteratively and obtain $W \in f(\mathcal{I}_1)$. Recall that, by definition of $\mathcal{I}_1$, each $i \in N$ has the approval ballot $A_i\cap C'$.

    Now, we modify $\mathcal{I}_1$ further to create an instance $\mathcal{I}_2= (N, A',C', k )$ as follows. We make $N'$ agree on the same set of winners by adding $S$ to the ballots of all voters in $N'$. More precisely, we set $A'_i \coloneqq (A_i \cap C')\cup S$ for all $i\in N'$. For all $i \notin N'$, we do not change the approval ballots, i.e., $A'_i \coloneqq A_i \cap C'$. Since  $S\subseteq C'$, we have $A'_i\subseteq C'$ for all $i\in N$, and hence $\mathcal{I}_2$ is a well-defined instance.
    Note that as we obtained $\mathcal{I}_2$ from $\mathcal{I}_1$ by adding candidates from $S\subseteq W$ to the approval ballots of some voters, monotonicity of $f$ implies that $W \in f(\mathcal{I}_2)$. By definition of $\mathcal{I}_2$, all voters in $N'$ approve exactly the ``party'' $S\cup\{c\}$.

    Finally, we transform $\mathcal{I}_2$ into a party-list instance $\mathcal{I}_3= (N, A'', C', k)$ by applying robustness to fully satisfied voters and monotonicity. For all $i \in N\setminus N'$, we change the approval ballot to $A''_i \coloneqq W \setminus S$, while all $i \in N'$ keep their approval ballots $A''_i = A'_i$. This was a valid modification according to robustness to fully satisfied voters and monotonicity. Formally, for each $i \in N \setminus N'$ we first apply \ifsv to replace $A'_i$ by $(A'_i \cap W) \setminus S$, and then apply monotonicity to add all remaining candidates from $W \setminus S$. Therefore we get that $W \in f(\mathcal I_3)$. As $S \subset W$ the profile is structured as follows: every voter in $N'$ approves exactly $S \cup \{c\}$, while the voters in $N \setminus N'$ approve exactly $W \setminus S \neq \emptyset$. Hence, $\mathcal{I}_3$ is a party-list instance with parties $C_1 = S \cup \{c\}$ and voters $N_1 = N'$ as well as $C_2 = W \setminus S$ and voters $N_2 = N \setminus N'$. Lower quota for party-lists applied to $C_1$ now implies for $W \in f(\mathcal{I}_3)$, that 
    \(
    \lvert W \cap C_1\rvert = \lvert W \cap (S \cup \{c\})\rvert \ge \min\left(\left\lfloor k\frac{\lvert N'\rvert}{n}\right\rfloor, \lvert S \cup \{c\}\rvert\right) \ge \min(\ell, \lvert S \rvert + 1).
    \)
    Since $\lvert S\rvert < \lvert S\rvert +1$ and the above inequality is true, it holds that $\ell < \lvert S\rvert +1$ and thus $\lvert S \rvert \ge \ell$, as desired. This concludes the proof that $W$ satisfies PJR+ and hence $f\subseteq \pjrp$.

    \textbf{Part 2:}
    To prove that $\pjrp$ satisfies all axioms consider an instance $\mathcal{I} = (N, A, C,k)$ and let $W\in \pjrp(\mathcal I)$. By \Cref{prop:lq} PJR+ satisfies lower quota for party-lists.

    For all other three axioms, the proof that PJR+ satisfies them follows the same structure: we start with an instance $\mathcal{I} = (N, A, C, k)$ and $W \in \pjrp(\mathcal I)$. Then, we modify $\mathcal{I}$ to an instance $\mathcal{I}' = (N, A', C', k)$ and need to show that $W \in \pjrp(\mathcal{I}')$. To show this, we assume that we are given a $c \in C' \setminus W$, $\ell \in [k]$ and $\ell$-large group $N' \subseteq N_c$ in the instance $\mathcal{I}'$. We now need to prove that $\left\lvert\bigcup_{i \in N'} (A'_i \cap W) \right\rvert \ge \ell$. 
    
    For independence of losers, we notice that since $W$ satisfies PJR+ and $c$ is present in $\mathcal{I}$ it must hold that $\left\lvert\bigcup_{i \in N'} (A'_i \cap W) \right\rvert = \left\lvert\bigcup_{i \in N'} (A_i \cap W) \right\rvert \ge \ell$. Hence, $W$ also satisfies PJR+ in $\mathcal{I}'$ and therefore PJR+ satisfies independence of losers.

    For monotonicity, since the voters only added approvals of candidates from $W$, it also holds that $c\in \bigcap_{i\in N'} (A_i \setminus W)$. Thus, since $W\in \pjrp(\mathcal I)$, we get that $\lvert \bigcup_{i\in N'} (A_i \cap W)\rvert \ge \ell$, and since $ A'_i \cap W \supseteq   A_i \cap W$ for all voters $i\in N$, it especially holds that $\lvert \bigcup_{i\in N'} (A'_i \cap W)\rvert \ge \ell$. This proves that there are no witnesses of a $\pjrp$ violation in $\mathcal {I'}$. Therefore, $W \in \pjrp(\mathcal {I'}) $, proving that PJR+ satisfies monotonicity.  
    
    For robustness to fully satisfied voters, for all voters $i\in N'$ we have $A'_i \setminus W \neq \emptyset$, and hence $A'_i\subseteq A_i\cap W$ cannot hold. Therefore, $N'$ consists entirely of unchanged voters and we must have $A'_i = A_i$.
    Since $W$ satisfies PJR+ in $\mathcal I$, this implies that $\lvert \bigcup_{i\in N'} (A'_i \cap W)\rvert = \lvert \bigcup_{i \in N'} (A_i \cap W)\rvert \ge \ell $. This proves that PJR+ satisfies robustness to fully satisfied voters.\qedhere

\end{proof}
We illustrate the proof idea behind \Cref{thm:PJRP_Refinement} with an example. Rather than following the proof forward, we present the contrapositive intuition. For a proportionality notion $f$ which satisfies the axioms, instance $\mathcal{I}$, and committee $W \in f(\mathcal{I})$ not satisfying PJR+ we use independence of losers, monotonicity, and \ifsv\ to transform $\mathcal{I}$ into a party-list instance $\mathcal{I}_3$ on which $W$ violates lower quota. As $f$ satisfies all three axioms, it must also hold that $W \in f(\mathcal{I}_3)$, thereby showing that $f$ does not satisfy lower quota for party-lists.
\begin{example}[Proof Sketch for \Cref{thm:PJRP_Refinement}]
Consider the following instance $\mathcal I$ with $n=4$ and $k=4$, in which the committee $W=\{c_1,c_2,c_3,c_4\}$ (highlighted in magenta) fails PJR+, due to the witness $N'=\{1,2\}$, $\ell = 2$, and candidate $d\notin W$. Note that candidate $c_1$ is approved by voters $1$ and $4$.

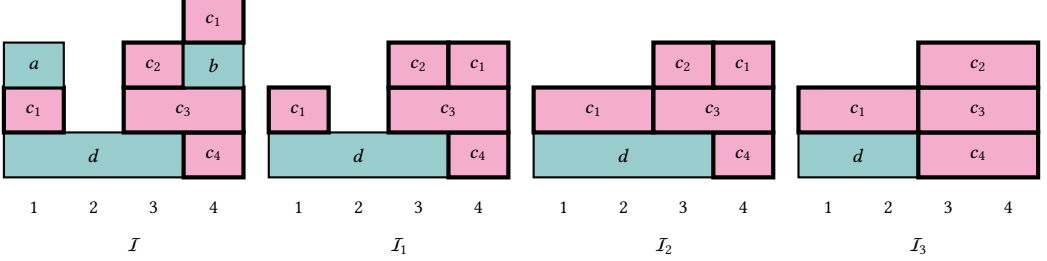
\begin{center}
\resizebox{\linewidth}{!}{
\begin{tabular}{@{}c@{\hspace{0.8em}}c@{\hspace{0.8em}}c@{\hspace{0.8em}}c@{}}

\begin{tikzpicture}
[yscale=0.6,xscale=0.8,
 voter/.style={anchor=south}]
    \foreach \i in {1,...,4}
        \node[voter] at (\i-0.5, -1) {\scriptsize $\i$};

    \draw[fill=teal!\strstr, thick]        (0, 0) rectangle (3, 1);
    \node at (1.5, 0.5) {\scriptsize $d$};

    \draw[fill=magenta!\strstr, ultra thick]  (3, 0) rectangle (4, 1);
    \node at (3.5, 0.5) {\scriptsize $c_{4}$};

    \draw[fill=magenta!\strstr, ultra thick]  (0, 1) rectangle (1, 2);
    \node at (0.5, 1.5) {\scriptsize $c_{1}$};

    \draw[fill=magenta!\strstr, ultra thick]  (2, 1) rectangle (4, 2);
    \node at (3, 1.5) {\scriptsize $c_{3}$};

    \draw[fill=teal!\strstr, thick]        (0, 2) rectangle (1, 3);
    \node at (0.5, 2.5) {\scriptsize $a$};

    \draw[fill=magenta!\strstr, ultra thick]  (2, 2) rectangle (3, 3);
    \node at (2.5, 2.5) {\scriptsize $c_{2}$};

    \draw[fill=teal!\strstr, thick]        (3, 2) rectangle (4, 3);
    \node at (3.5, 2.5) {\scriptsize $b$};

    \draw[fill=magenta!\strstr, ultra thick]  (3, 3) rectangle (4, 4);
    \node at (3.5, 3.5) {\scriptsize $c_{1}$};
\end{tikzpicture}

&

\begin{tikzpicture}
[yscale=0.6,xscale=0.8,
 voter/.style={anchor=south}]
    \foreach \i in {1,...,4}
        \node[voter] at (\i-0.5, -1) {\scriptsize $\i$};

    \draw[fill=teal!\strstr, thick]        (0, 0) rectangle (3, 1);
    \node at (1.5, 0.5) {\scriptsize $d$};

    \draw[fill=magenta!\strstr, ultra thick]  (3, 0) rectangle (4, 1);
    \node at (3.5, 0.5) {\scriptsize $c_{4}$};

    \draw[fill=magenta!\strstr, ultra thick]  (0, 1) rectangle (1, 2);
    \node at (0.5, 1.5) {\scriptsize $c_{1}$};

    \draw[fill=magenta!\strstr, ultra thick]  (2, 1) rectangle (4, 2);
    \node at (3, 1.5) {\scriptsize $c_{3}$};

    \draw[fill=magenta!\strstr, ultra thick]  (2, 2) rectangle (3, 3);
    \node at (2.5, 2.5) {\scriptsize $c_{2}$};

    \draw[fill=magenta!\strstr, ultra thick]  (3, 2) rectangle (4, 3);
    \node at (3.5, 2.5) {\scriptsize $c_{1}$};
\end{tikzpicture}

&

\begin{tikzpicture}
[yscale=0.6,xscale=0.8,
 voter/.style={anchor=south}]
    \foreach \i in {1,...,4}
        \node[voter] at (\i-0.5, -1) {\scriptsize $\i$};

    \draw[fill=teal!\strstr, thick]        (0, 0) rectangle (3, 1);
    \node at (1.5, 0.5) {\scriptsize $d$};

    \draw[fill=magenta!\strstr, ultra thick]  (3, 0) rectangle (4, 1);
    \node at (3.5, 0.5) {\scriptsize $c_{4}$};

    \draw[fill=magenta!\strstr, ultra thick]  (0, 1) rectangle (2, 2);
    \node at (1, 1.5) {\scriptsize $c_{1}$};

    \draw[fill=magenta!\strstr, ultra thick]  (2, 1) rectangle (4, 2);
    \node at (3, 1.5) {\scriptsize $c_{3}$};

    \draw[fill=magenta!\strstr, ultra thick]  (2, 2) rectangle (3, 3);
    \node at (2.5, 2.5) {\scriptsize $c_{2}$};

    \draw[fill=magenta!\strstr, ultra thick]  (3, 2) rectangle (4, 3);
    \node at (3.5, 2.5) {\scriptsize $c_{1}$};
\end{tikzpicture}

&

\begin{tikzpicture}
[yscale=0.6,xscale=0.8,
 voter/.style={anchor=south}]
    \foreach \i in {1,...,4}
        \node[voter] at (\i-0.5, -1) {\scriptsize $\i$};

    \draw[fill=teal!\strstr, thick]        (0, 0) rectangle (2, 1);
    \node at (1, 0.5) {\scriptsize $d$};

    \draw[fill=magenta!\strstr, ultra thick]  (0, 1) rectangle (2, 2);
    \node at (1, 1.5) {\scriptsize $c_{1}$};

    \draw[fill=magenta!\strstr, ultra thick]  (2, 0) rectangle (4, 1);
    \node at (3, 0.5) {\scriptsize $c_{4}$};

    \draw[fill=magenta!\strstr, ultra thick]  (2, 1) rectangle (4, 2);
    \node at (3, 1.5) {\scriptsize $c_{3}$};

    \draw[fill=magenta!\strstr, ultra thick]  (2, 2) rectangle (4, 3);
    \node at (3, 2.5) {\scriptsize $c_{2}$};
\end{tikzpicture}

\\[-0.25em]
\multicolumn{1}{c}{\scriptsize\textbf{$\mathcal I$}} &
\multicolumn{1}{c}{\scriptsize\textbf{$\mathcal I_1$}} &
\multicolumn{1}{c}{\scriptsize\textbf{$\mathcal I_2$}} &
\multicolumn{1}{c}{\scriptsize\textbf{$\mathcal I_3$}} \\

\end{tabular}
}
\end{center}

Since $\bigcup_{i\in N'} (A_i\cap W) = \{c_1\}$, our goal is to create a party-list profile in which all voters in $N'$ approve of $d$ and $c_1$, while all voters in $N\setminus N'$ approve of $W\setminus \{c_1\}$.

\begin{itemize}
    \item First, we remove all candidates not in $W\cup\{d\}$ (i.e., $a$ and $b$) from $\mathcal I$ using independence of losers. We arrive at the simpler $\mathcal I_1$.
    \item Next, we use the monotonicity of $f$ and add approvals to all voters in $N'$, until they all approve every candidate in $\bigcup_{i\in N'} (A_i\cap W)$. The resulting instance is $\mathcal I_2$. In this case, we only need to add the approval of $c_1$ for voter $2$.
    \item Finally, we modify the approvals of $N\setminus N' = \{3,4\}$. First, we apply {\ifsv} to remove any approvals of $d$ or $c_1$. In particular, we apply \ifsv to remove the approval of $d$ for voter $3$ and the approval of $c_1$ for voter $4$. Then, we use monotonicity to let all voters in $N\setminus N'$ approve of $W\setminus \{c_1\}$. We hence arrive at instance $\mathcal I_3$.
    \item The instance $\mathcal I_3$ is a party-list instance with two parties of equal size (two voters each), so lower quota requires each party to receive at least $\lfloor k\cdot 2/4\rfloor=2$ seats.
But $W$ selects only one candidate from $\{d,c_1\}$, hence $W\notin \LQ(\mathcal I_3)$.
\end{itemize}
\end{example}

Interestingly, one can achieve a similar result without independence of losers, robustness to fully satisfied voters, and lower quota for party-lists:
if we take PJR as granted but do not accept violations of monotonicity, we also naturally arrive at refinements of PJR+. Hence, monotonicity is the differentiating factor between PJR and PJR+.
\begin{restatable}{proposition}{pjrptopjr}\label{prop:pjr_monotonicity}
    Let $f$ be a proportionality notion such that $f \subseteq \pjr$ and $f$ satisfies monotonicity. Then $f \subseteq \pjrp$.
\end{restatable}

Note that \Cref{prop:pjr_monotonicity} does not follow from the fact that PJR satisfies all axioms but monotonicity in \Cref{thm:PJRP_Refinement}, because there are refinements of PJR which satisfy monotonicity but not all other axioms. For instance, priceability is a refinement of PJR which fails to satisfy \ifsv (see \Cref{sec:compliance}), and hence \Cref{thm:PJRP_Refinement} is not applicable. Nonetheless, we can apply \Cref{prop:pjr_monotonicity}: since priceability refines PJR and satisfies monotonicity, it also refines PJR+.

\subsection{Characterizing Refinements of EJR+}
We now build on the previous characterization and extend it to EJR+. Recall that PJR+ and EJR+ take two different perspectives on representation: while PJR+ considers the satisfaction of each group of voters as a whole, EJR+ considers the satisfaction of each voter within a group individually. Since for approval utilities the satisfaction of an individual voter is only measured in terms of the number of approved winning candidates, adding independence of approval swaps (\Cref{def:independence_of_approval_swaps}) to the axioms of \Cref{thm:PJRP_Refinement} excludes PJR+ and characterizes EJR+ as the coarsest notion satisfying this collection of axioms.
\begin{restatable}{theorem}{ejrprefc}
     Let $f$ be a proportionality notion satisfying independence of losers, monotonicity, \ifsv, independence of approval swaps, and lower quota for party-lists. Then $f \subseteq \ejrp$. Further, $\ejrp$ satisfies all five axioms.
     \label{ejrp-ref-char}
\end{restatable}
\begin{proof}
    \textbf{Part 1:} Consider any instance $\mathcal{I} = (N, A, C, k)$ and committee $W \in f(\mathcal{I})$. 
    
    Further, take any $c \notin W$, $\ell \in [k]$, and $\ell$-large set $N'\subseteq N_c$. We define $S \coloneqq \bigcup_{i \in N'}(A_i \cap W)$ and enumerate $S = \{c'_1, \dots, c'_{\lvert S\rvert}\}$. For voter $i \in N'$ we create the approval ballot $A'_i = (A_i \setminus W) \cup \{c'_1, \dots, c'_{\lvert A_i \cap W \rvert}\}$ while for any $i \notin N'$ we set $A'_i = A_i$ and consider the instance $\mathcal{I}' \coloneqq (N, A', C, k)$. Thus, for any $i \in N$ we have $A_i \setminus W = A'_i\setminus W$ and $\lvert A_i \cap W\rvert = \lvert A'_i \cap W\rvert$. Hence, by independence of approval swaps this implies that $W \in f(\mathcal{I}')$. Further, as $f$ satisfies independence of losers, monotonicity, robustness to fully satisfied voters, and lower quota for party-lists we know that it is a PJR+ refinement. Thus, $W$ satisfies PJR+ in $\mathcal{I}'$. However, now, we observe that $N' \subseteq N_c$ still holds in $\mathcal{I}'$ as the approval ballots outside of $W$ did not change.     
    Therefore, by applying PJR+ we get that $\left\lvert \bigcup_{i \in N'} (A'_i \cap W) \right\rvert \ge \ell$. Now, by construction, we know that for all $i \in N'$ the set $A'_i \cap W$ is a prefix of the same ordering of $S$. Thus, it must hold that $\ell \le \left\lvert \bigcup_{i \in N'} (A'_i \cap W) \right\rvert = \max_{i \in N'}\left\lvert A'_i \cap W \right\rvert$. Since this voter $i$ maximizing the number of approvals, approved the same number of candidates in $\mathcal{I}$ and $\mathcal{I}'$, it also must hold that $\left\lvert A_i \cap W \right\rvert \ge \ell$. Consequently, $W$ satisfies EJR+ in $\mathcal{I}$.

    \textbf{Part 2:} As it largely mirrors the proofs for PJR+ we delegate showing that EJR+ satisfies all five axioms to \Cref{app:missing}.
 \end{proof}
In \Cref{app:robust} we show that all axioms used in this section are necessary for obtaining the characterizations: dropping any one axiom, there exists a proportionality notion that satisfies the remaining axioms but is strictly coarser than PJR+ (respectively, EJR+).

As a corollary of our proof, we immediately obtain a second perspective on EJR+: any notion refining PJR+ and satisfying independence of approval swaps must also be a refinement of EJR+.
\begin{corollary}\label{cor:pjrp_ioas_ejrp}
    Let $f$ be a proportionality notion such that $f \subseteq \pjrp$ and $f$ satisfies independence of approval swaps. Then, $f \subseteq \ejrp$.
\end{corollary}
Note that this corollary even holds for refinements of PJR+ that do not satisfy monotonicity, \ifsv, or independence of losers.
Combining \Cref{cor:pjrp_ioas_ejrp} and \Cref{prop:pjr_monotonicity}, we obtain the following result.
\begin{restatable}{corollary}{ejrppjrpcon}\label{prop:ejr_mon}
    Let $f$ be a proportionality notion such that (i) $f \subseteq \pjr$, and (ii) $f$ satisfies independence of approval swaps and monotonicity. Then, $f \subseteq \ejrp$.
\end{restatable}
 Since every refinement of EJR is a refinement of PJR, too, \Cref{prop:ejr_mon} entails that EJR+ is a natural refinement of EJR.
Further, while EJR+ is the \emph{coarsest} notion satisfying the five axioms from \Cref{ejrp-ref-char}, some refinements of EJR+ do the same. 
We discuss a non-trivial such notion, called \emph{near perfect representation} (NPR), in \Cref{app:NPR}.

\section{Witness-Based Proportionality}\label{sec:Witness_Rationalizability}

Recall that we focus on notions capturing \emph{underrepresentation}, that is, proportionality notions preventing groups of voters from being less represented than they deserve to be.
For a chosen committee to be proportional, it should, intuitively, represent all groups within the electorate sufficiently. Thus, following this intuition, if a committee violates proportionality, then some group of voters is not properly represented. %
Indeed, the proportionality notions we have considered so far provide ``explanations'' in the form of underrepresented groups of voters for why a given committee is non-proportional. %
For instance, if a committee does not satisfy EJR+, we can identify a concrete subset of the voters who could challenge the selected committee and make a reasonable claim that their group is not as well represented as it deserves to be. We call such underrepresented groups of voters \emph{witnesses}, as they witness that the given committee is not proportional.

Our goal in this section is to formalize this concept of \textit{witness-based} proportionality and use it to obtain \emph{exact} characterizations of EJR+ and PJR+. %
First, a mandatory requirement for witnesses is that of locality: whether a group of voters is not properly represented should only depend on the group and its approvals, and not on the approval ballots of voters outside this group. 
If, instead, we claimed that $N'$ is a witness of a proportionality violation but $N'$ could become well-represented by (only) changing the approval ballots of other voters, then the violation would not truly be caused by $N'$. Instead, the culprit would be the interplay between $N'$ and the other voters, and therefore we should have provided a different witness. Hence, before we can define witness-basedness, we require a formal notion of locality.
\begin{definition}[Local Embeddings]
Let $\mathcal{I} = (N, A, C, k)$ and $\hat{\mathcal{I}} = (\hat N, \hat A, \hat C, k)$ be two instances and $W \in \binom{C}{k}$ as well as $\hat{W} \in \binom{\hat{C}}{k}$ be two committees. Further, let $N' \subseteq N$ and $\hat{N'} \subseteq \hat{N}$ be two subsets of voters. We say that $(N', W)$ can be \emph{locally embedded} into $(\hat{N'}, \hat{W})$ if there exists a bijective function $\Phi_{N'} \colon N' \to \hat{N'}$ and an injective function $\Phi_{C} \colon W\cup\bigcup_{i \in N'} A_i \to \hat{W}\cup\bigcup_{i \in \hat{N'}} \hat{A}_i $ such that $\Phi_{C}(W) = \hat W$ and for all $j \in N'$ and $c \in W\cup\bigcup_{i \in N'} A_i$ it holds that \(c \in A_j \text{ if and only if } \Phi_{C}(c) \in \hat A_{\Phi_{N'}(j)}.\)
    
\end{definition}
Equivalently $(N', W)$ locally embeds into $(\hat N', \hat W)$ if the bipartite incidence structure between voters in $N'$ and candidates in $W \cup \bigcup_{i \in N'} A_i$ can be injectively mapped into the corresponding structure for $\hat{N}'$ while mapping $W$ onto $\hat W$. In particular, approvals of voters in $N\setminus N'$ are not relevant for this embedding. Intuitively, a local embedding preserves how voters in $N'$ vote for the committee $W$. Voters outside of $N'$ and votes from $N'$ for candidates outside the image of $\Phi_{C}$ can behave arbitrarily.

To illustrate the concept of local embeddings, consider the following example.
\begin{example}
\begin{wrapstuff}[r,type=figure,width=5cm]
    \centering
    \begin{tikzpicture}
    [yscale=0.6,xscale=0.8,
    voter/.style={anchor=south}]
    
        \foreach \i in {1,...,6}
    		\node[voter] at (\i-0.5, -1) {$\i$};
        
        \draw[fill=magenta!\clrstr] (0, 0) rectangle (4, 1);
        \draw[fill=magenta!\clrstr] (0, 1) rectangle (1, 2);
        \draw[fill=magenta!\clrstr] (1, 1) rectangle (3, 2);
        \draw[fill=magenta!\clrstr] (4, 1) rectangle (6, 2);
        \draw[fill=magenta!\clrstr] (3,2) rectangle (6, 3);
        \draw[fill=magenta!\clrstr] (4,0) rectangle (6, 1);
        \node at ( 2, 0.5) {$c_{1}$};
        \node at ( 2, 1.5) {$c_{3}$};
        \node at ( 0.5, 1.5) {$c_{2}$};
        \node at ( 5, 1.5) {$c_{4}$};
        \node at ( 4.5, 2.5) {$c_{5}$};
        \node at ( 5, 0.5) {$c_{6}$};
    \end{tikzpicture}
    \caption{Example for local embeddings.}
    \label{fig:embed}
\end{wrapstuff}
As an example, consider the instance depicted in \Cref{fig:embed} with $k = 3$. First, we claim that in this instance itself we can embed $(\{1,2,3\}, \{c_1, c_2, c_3\})$ into  $(\{4,5,6\}, \{c_4, c_5, c_1\})$, namely by mapping voter $1$ to $4$, $2$ to $5$, and $3$ to $6$ as well as $c_1$ to $c_5$, $c_2$ to $c_1$, and $c_3$ to $c_4$. The approvals of candidate $c_1$ outside of this set of voters are irrelevant. Note that we cannot embed $(\{4,5,6\}, \{c_4, c_5, c_1\})$ into $(\{1,2,3\}, \{c_1, c_2, c_3\})$, though. This, for instance, can be observed by the fact that voter $6$ approves three candidates in total, while every voter in $\{1,2,3\}$ approves at most two. 
\wrapstuffclear
\end{example}

We are now ready to introduce our framework of witness-based proportionality notions.
Formalizing our previously discussed intuition, a witness for a proportionality violation should remain a witness if we can locally embed it into another instance with the same number of voters.  
\begin{definition}[Witness-Based Proportionality Notions]\label{def:witness_Based}
    Let $f$ be a proportionality notion. We call a function $w$ mapping every instance $\mathcal{I} = (N, A, C, k)$ and committee $W \in \binom{C}{k}$ to a \emph{set of witnesses} $w(\mathcal{I}, W) \subseteq 2^N \setminus \{\emptyset\}$ \emph{a witness function of $f$} if
    \begin{itemize}
        \item[i)] $W \in f(\mathcal{I})$ if and only if  $w(\mathcal{I}, W) = \emptyset$;
        \item[ii)] if $N' \in w(\mathcal{I}, W)$ and $(N', W)$ can be locally embedded into $(\hat{N}', \hat{W})$ in an instance $\hat{\mathcal{I}} = (\hat{N}, \hat{A}, \hat{C}, k)$ with $\lvert N\rvert = \lvert \hat{N}\rvert$, then $\hat{N}' \in w(\hat{\mathcal{I}},\hat{W})$.
    \end{itemize}
    If $f$ has a witness function, we also say that it is \emph{witness-based}.
\end{definition}
On its own, the requirement of being witness-based is not very restrictive. For any proportionality notion $f$, a potential witness function is the trivial assignment $w(\mathcal{I},W)=\{N\}$ if $W\notin f(\mathcal{I})$ and $w(\mathcal{I},W)=\emptyset$ otherwise.
Under mild assumptions (i.e., anonymity, neutrality, independence of unapproved candidates, and independence of losers), this assignment is a witness function of $f$ (see \Cref{app:wb}). 
Nonetheless, this framework is very helpful as we can impose additional desirable axioms on the witness functions. We refer to \Cref{tab:big-fingerprint} in \Cref{sec:compliance} for a list of all introduced axioms.%

Note that, so far, each proportionality notion was satisfied if a condition held for all subsets $N'\subseteq N$. Conversely, if a committee is not proportional, there must exist some subset $N'\subseteq N$ not meeting this condition. We illustrate this via the example of EJR+.
Let $\mathcal{I} = (N, A, C, k)$ be an instance and $W$ be a committee. For this committee, we say that $N' \subseteq N$ is a \emph{natural witness} for an EJR+ violation of $W$ in $\mathcal I$, if there exist $\ell\in[k]$ and $c\in C\setminus W$ such that $N'\subseteq N_c$,
$|N'|\ge \ell\frac{n}{k}$, and $|A_i\cap W|<\ell$ for all $i\in N'$.  We call the function $\natWitnessEjrp$ selecting \[\natWitnessEjrp (\mathcal{I}, W) \coloneqq \{N' \subseteq N \colon N' \text{ is a natural witness for an EJR+ violation of } W \text{ in } \mathcal{I}\}\] the \emph{natural witness function} of EJR+. Similarly, natural witness functions can be defined for all previously introduced proportionality notions.
We formally show that the function $\natWitnessEjrp$ is indeed a witness function for EJR+.
\begin{observation}
    $\ejrp$ is a witness-based proportionality notion. In particular, $\natWitnessEjrp$ is a witness function of $\ejrp$ according to \Cref{def:witness_Based}. 
\end{observation}
\begin{proof}
    Firstly, $W$ satisfying EJR+ is indeed equivalent to $\natWitnessEjrp(\mathcal{I}, W) = \emptyset$, as $\natWitnessEjrp(\mathcal{I}, W)$ by definition contains precisely the subsets of voters that cause EJR+ violations. 

    Secondly, let $N' \in \natWitnessEjrp(\mathcal{I}, W)$ and consider another instance $\hat{\mathcal{I}} = (\hat N, \hat A, \hat C, k)$ and $(\hat{N'}, \hat{W})$ such that $\lvert\hat N\rvert = \lvert N\rvert$ and
    $(N', W)$ can be locally embedded into $(\hat{N'}, \hat{W})$  via $\Phi_{N'}$ and $\Phi_{C}$. Our goal is to show $\hat{N}'\in \natWitnessEjrp(\hat{\mathcal{I}},\hat W)$.
 As $N'$ is a natural witness for $W$ violating EJR+, there exist $\ell \in [k]$ and $c \in \bigcap_{i \in N'} (A_i \setminus W)$ such that $\lvert N'\rvert \ge \ell \frac{n}{k}$ and $\lvert A_i \cap W\rvert < \ell$ for all $i \in N'$. Via our local embedding, we also obtain that $\Phi_{C}(c) \in \bigcap_{i \in \hat{N'}} \left(\hat{A}_i \setminus \Phi_{C}(W)\right)$ and since $\Phi_{N'}$ is a bijection we know that $\lvert \hat{N'}\rvert = \lvert N'\rvert \ge \ell \frac{n}{k}$. Further, since $\Phi_C$ bijectively maps $W$ onto $\hat{W}$ and preserves the approvals of voters in $\hat{N'}$ for any $i \in \hat{N'}$ we get that $\lvert \hat{A}_i \cap \hat{W}\rvert = \lvert A_{\Phi_{N'}^{-1}(i)} \cap W\rvert < \ell $. In other words, $\hat{N'}$ 
 is a natural witness for an EJR+ violation of $\hat{W}$, i.e., $\hat{N}' \in \natWitnessEjrp(\hat{\mathcal{I}}, \hat{W})$. 
\end{proof}

For our first axiom we impose on witness functions, we exploit a commonality between %
our previously introduced notions: cohesiveness. To motivate 
cohesiveness, we ask the following question: when can a group of voters rightfully claim to be underrepresented?
Let us first answer this question for the well-structured case of party-list instances. Here, voters naturally form groups based on the party they support. They are underrepresented whenever their party does not obtain sufficiently many seats. Crucially, in this case, the supporters of the party can point to the unelected party members. These candidates function as a claim to representation, but also as a solution to the underrepresentation caused by the current committee.
Outside of the party-list domain, parties generally cease to exist. However, the same idea can be generalized: a group of voters should only be able to claim underrepresentation, if they can point to concrete candidates that represent the entire group of voters, but have not been elected into the committee. Indeed, this approach ensures that the underrepresented group does not merely consist of individually dissatisfied voters, but that the dissatisfaction has a common root.

\begin{definition}[Cohesiveness-Based Proportionality Notions]
    Let $f$ be a proportionality notion with witness function $w$. We say that $(f,w)$ is \emph{cohesiveness-based} if for every instance $\mathcal{I}$ and committee $W$ such that $W \notin f(\mathcal{I})$, it holds that $\bigcap_{i \in N'} (A_i \setminus W) \neq \emptyset$ for every $N' \in w(\mathcal{I}, W)$. 
\end{definition}
Indeed, all the previous notions (except for the core) together with their natural witness functions are cohesiveness-based. 

For the second axiom that we impose on witness functions, we argue the following: a voter who is content with the committee has no incentive to participate in a collective complaint. Therefore, if a group of voters witnesses a violation according to a proportionality notion, each of its voters should be dissatisfied with their representation. If we instead claimed that a set of voters $N'$ is a witness for a proportionality violation but includes voters satisfied with the committee, then we unnecessarily enlarged the witness and should have provided a smaller one instead. To formalize this consider a single dissatisfied voter within a witness. If all other voters were to adapt their preferences, then each of these voters would remain individually discontent. Additionally, the group of voters would now be completely unified in their complaint, and therefore have an even stronger collective claim to underrepresentation than before. Consequently, this group of voters should remain a witness for the proportionality violation.

\begin{definition}[Individual discontentment]
Let $f$ be a witness-based proportionality notion with witness function $w$.
We say that $(f,w)$ satisfies \emph{individual discontentment} if the following holds: for any instance $\mathcal I=(N,A,C,k)$, any committee $W$, and any witness $N'\in w(\mathcal I,W)$:
for every $j\in N'$ let $\mathcal I^{(j)}=(N,A^{(j)},C,k)$ be the instance where
$A^{(j)}_i = A_i$ for all $i\in N\setminus N'$ and $A^{(j)}_i = A_j$ for all $i\in N'$.
Then $N'\in w(\mathcal I^{(j)},W)$.
\end{definition}

Indeed, all of our previously studied notions %
satisfy individual discontentment, except for the EJR+ refinement NPR (see \Cref{app:NPR}). 
We consider EJR+ to demonstrate the axiom.
\begin{proposition}
    EJR+ together with its natural witness function is cohesiveness-based and satisfies individual discontentment.
\end{proposition}
\begin{proof}
    $(\ejrp,\natWitnessEjrp)$ being cohesiveness-based follows from the definition.
    To prove satisfaction of individual discontentment, let %
    $N'$ be a natural witness for some $W$ violating EJR+. Hence, there exist $\ell \in [k]$ and $c \notin W$ such that $\lvert N'\rvert \ge \ell\frac{n}{k}$, $c\in A_i$ and $\lvert A_i \cap W \rvert < \ell$ for all $i \in N'$. Now for the modified instance $\mathcal{I}' = (N, A', C, k)$, with $A'_i = A_j $ for an arbitrary $j \in N'$ and for all $i \in N'$ (while the approvals outside of $N'$ stay the same), it still holds that $c\in A'_i$ and $\lvert A'_i \cap W \rvert < \ell$ for all $i \in N'$. Moreover, since it still holds that $\lvert N'\rvert \ge \ell \frac{n}{k}$, $N'$ remains a natural witness for an EJR+ violation.
\end{proof}

Before we come to our characterizations of PJR+ and EJR+, we briefly highlight that two of the previously studied axioms (independence of losers and \ifsv) are satisfied by all witness-based proportionality notions whose witnesses are cohesiveness-based. Hence, they will not be necessary for the later characterizations.
\begin{restatable}{lemma}{wbiol}
Let $f$ be a proportionality notion with witness function $w$. Then $f$ satisfies independence of losers. If $(f,w)$ is cohesiveness-based, $f$ also satisfies \ifsv.
\label{lem:witness_iol_ifsv}
\label{general_lem}
\end{restatable}
\subsection{Characterizing EJR+}
The fact that NPR (a stronger notion than EJR+) does not satisfy individual discontentment is not a coincidence. We show that EJR+ is the strongest (i.e., most demanding or \emph{finest}) lower quota extension among all witness-based notions
whose witnesses are cohesiveness-based and satisfy individual discontentment.
Equivalently, every such notion $f$ is a coarsening of EJR+. 
\begin{theorem}\label{thm:ejrp_coarsening}
    Let $f$ be a lower quota extension and let $w$ be a witness function of 
    $f$ such that $(f,w)$ is cohesiveness-based and satisfies individual discontentment.
    Then $\ejrp \subseteq f$. 
\end{theorem}
\begin{proof}
    We prove this statement by contraposition.
    Let $(f,w)$ be cohesiveness-based and satisfy individual discontentment. Further, let $\mathcal{I}$ be an instance with $W \in \ejrp(\mathcal{I}) \setminus f(\mathcal{I})$. Our goal is to provide a (party-list) instance on which $f$ is not a lower quota extension.

    As $W \notin f(\mathcal{I})$ there exists a witness $N' \in w(\mathcal{I} ,W)$. As $(f,w)$ is cohesiveness-based, there further exists a candidate $c \in \bigcap_{i \in N'} (A_i \setminus W)$ (and $N'$ cannot be empty). Let $\ell \coloneqq \left\lfloor k\frac{|N'|}{n}\right\rfloor$ (possibly $\ell=0$).
Then $|N'| \ge \ell \frac{n}{k}$. As $W$ satisfies EJR+ we know that there is a voter $i \in N'$ with $t \coloneqq \lvert A_i \cap W \rvert \ge \ell$. 

    Now first consider the instance $\mathcal{I}_2 = (N, A', C, k)$ with $A'_j = A_i$ for each $j \in N'$ and $A'_j = A_j$ for each $j \notin N'$. As $(f,w)$  satisfies individual discontentment, we get that $N' \in w(\mathcal{I}_2, W)$. 
    
Let $C^\star \coloneqq W \cup A_i$ and define the instance $\mathcal I_3=(N,A'',C^\star,k)$ by setting 
$A''_j \coloneqq A_i$ for any $j\in N'$, $A''_j \coloneqq W \setminus A_i$ for any $j \in N \setminus N'$ if $W \setminus A_i \neq \emptyset$ and $A''_j \coloneqq A_i$ for any $j \in N \setminus N'$ if $W \setminus A_i = \emptyset$.
 We claim that $\mathcal{I}_3$ is indeed a party-list instance and that $W$ satisfies lower quota on $\mathcal{I}_3$. If $W\setminus A_i = \emptyset$, $W$ satisfies lower quota for party-list by definition, as the instance just contains a single party from which all $k$ candidate are selected. Otherwise, the instance has the parties $A_i$ and $W \setminus A_i$. Lower quota for party-lists now requires that at least $\min\left(\lvert A_i\rvert, \left\lfloor k \frac{\lvert N'\rvert}{n}\right\rfloor\right) = \min(\lvert A_i\rvert,\ell) =  \ell$ candidates are selected from $A_i$. As $\lvert W \cap A_i\rvert = t  \ge \ell$ this is given. Moreover, the voters of the second party are fully satisfied and, thus, $W$ satisfies lower quota for party-lists in this instance. This proves the claim.
 However, we notice that $(N', W)$ in $\mathcal{I}_2$ can be locally embedded  into $(N', W)$ in $\mathcal{I}_3$ (as the ballot of no $j\in N'$ changes). By definition of witness functions, $N'\in w(\mathcal I_2, W)$ implies $N'\in w(\mathcal I_3, W)$. In other words, $W\notin f(\mathcal I_3)$, but $W$ satisfies lower quota, hence $f$ is not a lower quota extension. \qedhere
\end{proof}
Hence, among the widely studied family of cohesiveness-based ``justified representation'' notions, EJR+ is the most demanding one. 
As a consequence, we further obtain an \emph{axiomatic characterization of EJR+} among the cohesiveness- and witness-based proportionality notions. For this purpose, we only need to add monotonicity and independence of approval swaps to \Cref{thm:ejrp_coarsening}.
\begin{corollary}\label{cor:ejrp_Char}
    Let $f$ be a lower quota extension satisfying monotonicity and independence of approval swaps and let $w$ be a witness function of $f$ such that $(f,w)$ is cohesiveness-based and satisfies individual discontentment. Then $f = \ejrp$.
    \label{cor:ejrpluschar}
\end{corollary}
\begin{proof}
    By \Cref{thm:ejrp_coarsening} we know that $\ejrp\subseteq  f$. From \Cref{lem:witness_iol_ifsv} we get that $f$ also satisfies independence of losers and \ifsv. Since $f$ is a lower quota extension, it satisfies lower quota for party-lists. Thus, we can apply \Cref{ejrp-ref-char} and also obtain $f \subseteq \ejrp $. As a consequence, we get that $f = \ejrp$.
\end{proof}
In \Cref{app:robust} we show that all axioms used for this characterization are necessary. If any one of them is removed, a proportionality notion different from EJR+ satisfies the remaining ones. 
\subsection{Characterizing PJR+}
To distinguish PJR+ from EJR+, we consider the approach to representation that focuses on the satisfaction of the group as a whole. From this perspective, we must view all candidates that the group supports and the collective representation that the group obtains. This is in accord with the original motivation of \citet{SFF+17a} for PJR, who favor this viewpoint over the focus on individual voters.
When voters within the group exchange and adapt each other's approvals, the group as a whole still supports the same set of candidates. Therefore, the group should remain eligible for more representation.
To formulate this in the language of witness-based proportionality notions, consider voters $i,j$ in a witness set $N'$ for a proportionality violation. If voter $i$ adopts some of $j$'s approvals, i.e., $A_i\subseteq A_i' \subseteq A_i\cup A_j$, then the collectively allocated set of candidates does not change, and hence $N'$ should stay a witness. According to EJR+, it could be that $N'$ witnesses a violation in $A$, but that in $A'$ no violation can be detected anymore. However, the same is not true for PJR+.

\begin{definition}[Merge-proofness]\label{def:merge_proof}
Let $f$ be a witness-based proportionality notion with witness function $w$.
We say that $(f,w)$ satisfies \emph{merge-proofness} if for every instance $\mathcal{I}=(N,A,C,k)$,
every committee $W$, and every witness $N'\in w(\mathcal{I},W)$, the following holds.
For every profile $A'$ such that $A'_i=A_i$ for all $i\in N\setminus N'$ and
\(
A_i \subseteq A'_i \subseteq \bigcup_{j\in N'} A_j 
\)\text{ for all } $i\in N'$,
we have $N'\in w(\mathcal{I}',W)$ where $\mathcal{I}'\coloneqq (N,A',C,k)$.
\end{definition}

\begin{restatable}{proposition}{pjrpmergeproof}\label{prop:pjrp_merge_proof}
PJR+ with its natural witness function is cohesiveness-based and satisfies merge-proofness.
\end{restatable}

This additional axiom is sufficient to characterize coarsenings of PJR+. 

\begin{restatable}{theorem}{pjrpstr}\label{thm:pjrp_merge_lq}
Let $f$ be a lower quota extension and let $w$ be a witness function of $f$ such that $(f,w)$ is cohesiveness-based and satisfies merge-proofness. Then $\pjrp \subseteq f$.
\end{restatable}

By simply adding monotonicity to \Cref{thm:pjrp_merge_lq}, we obtain a characterization of PJR+.
\begin{corollary}\label{cor:pjrp_Char}
Let $f$ be a lower quota extension satisfying monotonicity and let $w$ be a witness function of $f$ such that $(f,w)$ is cohesiveness-based and satisfies merge-proofness. Then $f = \pjrp$.
    \label{pjrpluscharwitness}
\end{corollary}
In \Cref{app:robust} we show that all axioms used to obtain this characterization are indeed needed.

\section{Discussion, Limitations, and Future Work}
We initiated the study of axiomatically characterizing proportionality notions in approval-based multiwinner voting. In particular, we axiomatically studied the notions of EJR+ and PJR+ and obtained axiomatic characterizations of them. We further managed to establish monotonicity as a key distinction between the classical PJR and EJR and the newer PJR+ and EJR+ notions. In \Cref{app:robust} we additionally survey the axiomatic compliance of most remaining proportionality notions, for instance, core stability, FJR, and priceability.
Our work poses several fundamental limitations; overcoming these may yield new insights into proportionality.
Firstly, with our axioms imposed on witness-based notions, we only capture the justified representation strand of proportionality notions.  While %
these capture basic ideas of proportionality, %
they may fail to rule out intuitively unfair outcomes. For example, following the basic intuition of \citet{PeSk20a}, consider the instance with $k = 6$ and committee $\{c_1\}\cup \{c_6,\dots, c_{10}\}  $ depicted in \Cref{fig:laminar}. 
\begin{wrapstuff}[r,type=figure,width=5cm]
    \centering
    \begin{tikzpicture}
    [yscale=0.6,xscale=0.8,
    voter/.style={anchor=south}]
    
        \foreach \i in {1,...,4}
    		\node[voter] at (\i-0.5, -1) {$\i$};
        
        \draw[fill=teal!40, ultra thick] (0, 0) rectangle (2, 1);
        \draw[fill=teal!\clrstr] (0, 1) rectangle (1, 2);
        \draw[fill=teal!\clrstr] (1, 1) rectangle (2, 2);
        \draw[fill=teal!\clrstr] (0,  2) rectangle (1, 3);
        \draw[fill=teal!\clrstr] (1, 2) rectangle (2, 3);
        \draw[fill=magenta!40, ultra thick] (2, 0) rectangle (4, 1);
        \draw[fill=magenta!40,ultra thick ] (2, 1) rectangle (4, 2);
        \draw[fill=magenta!40, ultra thick] (2, 2) rectangle (4, 3);
        \draw[fill=magenta!40, ultra thick] (2, 3) rectangle (4, 4);
        \draw[fill=magenta!40, ultra thick] (2, 4) rectangle (4, 5);
        \node at ( 1, 0.5) {$c_{1}$};
        \node at ( 0.5, 1.5) {$c_{2}$};
        \node at ( 1.5, 1.5) {$c_{3}$};
        \node at ( 0.5, 2.5) {$c_{4}$};
        \node at ( 1.5, 2.5) {$c_{5}$};
        \node at ( 3, 0.5) {$c_{6}$};
        \node at ( 3, 1.5) {$c_{7}$};
        \node at ( 3, 2.5) {$c_{8}$};
        \node at ( 3, 3.5) {$c_{9}$};
        \node at ( 3, 4.5) {$c_{10}$};
    \end{tikzpicture}
    \caption{Example instance for laminar proportionality.}
    \label{fig:laminar}
\end{wrapstuff}
This committee satisfies EJR+, but is not ``intuitively proportional.'' Voters $1$ and $2$ are only represented by a single candidate, while they would deserve $3$. On the other hand, voters $3$ and $4$ are severely overrepresented. This can be interpreted as a failure of the laminar proportionality property, but also as a failure of priceability or related proportionality notions \citep{PeSk20a}. 
Such notions regard both under- and overrepresentation, and therefore require axioms outside of the witness-based framework.
Reconciling this seems to us like an important possibility for future work that could open the door to new, stronger (or maybe just different) proportionality notions, for instance by building on top of laminar proportionality instead of lower quota for party-lists. 
\wrapstuffclear
A second major limitation (even within the family of justified representation notions) stems from the fact that we only considered cohesiveness-based proportionality notions. With this, we are not able to capture several important notions, for instance FJR \citep{PPS21a} or FPJR \citep{KKL25a} and our witness-based characterizations only apply to a smaller (yet still natural) set of notions. One promising candidate for a characterization beyond cohesiveness-based proportionality notions is the sub-core of \citet{MSW22a} which we state here in an equivalent formulation (see \Cref{app:sub-core} for a formal proof of this equivalence). 
\begin{definition}
    A committee $W$ is in the sub-core if for every subset $C' \subseteq C \setminus W$ of candidates, $\ell \in [k]$, and $N' \subseteq \{i \in N \colon A_i \cap C' \neq \emptyset\}$ of size at least $\ell \frac{n}{k}$ it holds that 
    \(
    \left\lvert\bigcup_{i \in N'} (A_i \cap W) \right\rvert > \ell - \lvert C'\rvert.
    \) 
\end{definition}
It is easy to see that the sub-core indeed generalizes PJR+ (indeed PJR+ corresponds to the sub-core when restricted to $\lvert C'\rvert = 1$). It further satisfies all axioms satisfied by PJR+ except for being cohesiveness-based while (just like PJR+) still being implied by priceability. This poses a problem for finding a direct characterization of PJR+ without resorting to the notion being cohesiveness-based. Similarly, we find it an interesting question whether there is any characterization of the core (despite its existence being an open question). 
For instance, is the core the strongest proportionality notion under some reasonable axioms? Note that at least one of these ``reasonable'' axioms must be violated by EJR+ as the core is not a refinement of EJR+.

Thirdly, in our proofs we assume the lower quota for party-lists axiom to be given. While this might seem like a reasonable restriction, it would be interesting to find out if there are any normative justifications for this. Furthermore, lower quota for party-lists is not the only fairness notion in apportionment settings. For instance, instead of using the Hare quota to define lower quota, one could use the stronger Droop quota $\frac{n}{k+1}$ leading to a different yet still achievable family of proportionality notions \citep{CaEl25a}. We would suspect our characterizations to also work when replacing the Hare quota by the Droop quota. Besides this, we would be interested to see if our axiomatic approach could be applied to aid the design of \emph{overrepresentation notions}, a topic severely underrepresented in the approval-based multiwinner voting literature so far (see the workshop paper of \citet{LaNa25a} for some recent progress in the development of overrepresentation notions inspired by the concept of upper quota in apportionment instances). For overrepresentation, it might be necessary to move on from witness-based notions. 

Going beyond approval-based multiwinner voting, we see several possibilities for applying the axiomatic method to analyze fairness or proportionality notions in other areas of computational social choice. Firstly, there are several settings which are very closely related to approval-based multiwinner voting. For instance, in the approval-based apportionment setting of \citet{BGP+24a} EJR and EJR+ coincide and the core is known to be always satisfiable. Perhaps, in this setting a characterization of the core is possible. Other interesting possible avenues for future work include the ordinal preference setting (in which \citet{BrPe23a} highlighted a similar robustness tradeoff between their new ``rank-based'' notions and the classical notion of PSC \citep{Dumm84a}) or the fair mixing/randomized committee selection settings. Particularly, following a recent paper by \citet{SuVo24a} providing an axiomatic difference between their GRP notion and the GFS notion of \citet{ALS+23a} seems like an interesting open question.  

Finally, we would be interested to see if one can apply the axiomatic method to the setting of fair division. In recent work, \citet{GaSh25a} explored the relationship between \emph{22} different fairness notions in the fair division setting. This paints an even more convoluted picture than in the ABC voting literature. Finding axiomatic differences between fairness notions in fair division could contribute to disentangling this picture.

\bibliographystyle{ACM-Reference-Format}
\bibliography{abb, algo_fork,randomcitations}
\appendix
\section{Missing Proofs}
\label{app:missing}
\lqpleq*
\begin{proof}
    It is already known that $\pjr(\mathcal{I}) \subseteq \LQ(\mathcal{I})$ \citep[Proposition~6]{BLS18a}. Together with the known inclusions $\Core(\mathcal{I}) \subseteq \ejr(\mathcal{I}) \subseteq \pjr(\mathcal{I})$ \citep{LaSk22a} and $\ejrp(\mathcal{I}) \subseteq \pjrp(\mathcal{I}) \subseteq \pjr(\mathcal{I})$ \citep{BrPe23a} it remains to show that $\LQ(\mathcal{I}) \subseteq \ejrp(\mathcal{I})$ and $\LQ(\mathcal{I}) \subseteq \Core(\mathcal{I})$. 

    For the first inclusion, consider a party-list instance $\mathcal{I}$ and let $W$ be a committee not satisfying EJR+ in $\mathcal{I}$. Thus, there is a candidate $c \notin W$, $\ell \in [k]$, and an $\ell$-large set of voters $N'\subseteq N_c$ with $\lvert A_i \cap W \rvert < \ell$ for all $i \in N'$. As $\mathcal{I}$ is a party-list instance, there must exist a party $C_x \subseteq C$ with $c \in C_x$ and $A_i = C_x$ for all $i \in N'$. Thus, we have $N'\subseteq N_x$. As $N'$ is $\ell$-large, we get that $\left\lfloor \frac{k \lvert N'\rvert}{n}\right\rfloor \ge \left\lfloor \frac{k \frac{\ell n}{k}}{n}\right\rfloor = \ell$.  As less than $\ell$ candidates are selected from $C_x$ the committee $W$ also does not satisfy lower quota for party-lists in $\mathcal{I}$.

    On the other hand, assume that $W$ does not satisfy core stability. Thus, there exist a set $N' \subseteq N$ deviating together with a set of candidates $C' \subseteq C$ such that $\lvert N'\rvert \ge \frac{\lvert C'\rvert n}{k}$ and $\lvert A_i \cap W \rvert < \lvert A_i \cap C'\rvert$ for all $i \in N'$. As $\mathcal{I}$ is a party-list instance, we can partition $N'$ and $C'$ into sets $N'_1, \dots, N'_t$ and $D_1, \dots, D_t$ such that for any voter $i \in N'_j$ it holds that $A_i \cap C' = D_j$ and hence, by re-enumeration of the parties, $N'_j\subseteq N_j$. By the mediant-inequality, we obtain that 
    \[
    \frac{n}{k} \le \frac{\lvert N'\rvert}{\lvert C'\rvert} = \frac{\sum_{j = 1}^t \lvert N'_j\rvert }{\sum_{j = 1}^t \lvert D_j\rvert} \le \max_{j \in [t]} \frac{\lvert N'_j\rvert}{\lvert D_j\rvert}. 
    \]
    Thus, for the index $j^*$ maximizing the right hand side, we get that $\lvert N'_{j^*}\rvert \ge \lvert D_{j^*}\rvert \frac{n}{k}$ and consequently $\lvert W \cap C_{j^*}\rvert %
	< \lvert D_{j^*}\rvert \le \lfloor k \frac{\lvert N'_{j^*}\rvert}{n}\rfloor \le \lfloor k \frac{n_{j^*}}{n}\rfloor$. This also witnesses a lower quota violation.
\end{proof}
\pjrptopjr*
\begin{proof}We prove this claim by contraposition.
    Let $f$ be a proportionality notion satisfying monotonicity.
    Let there be an instance $\mathcal{I} = (N, A, C, k)$ and committee $W \in f(\mathcal{I})$ not satisfying PJR+. Thus, there exist a candidate $c \notin W$, $\ell \in [k]$, and $\ell$-large set $N'\subseteq N_c$ with $\lvert \bigcup_{i \in N'} (A_i \cap W) \rvert < \ell$. Our goal is to provide an instance on which $f$ does not refine PJR.
    
    For this purpose, consider the approval profile $A'$ with $A'_i = A_i$ for every $i \in N \setminus N'$ and $A'_i = (A_i \setminus W) \cup \bigcup_{j \in N'} (A_j \cap W)$. Intuitively speaking, each $i\in N'$ adds approvals to their ballot until they satisfy $A'_i\cap W =\bigcup_{j\in N'} A_j\cap W$. Hence, we arrive at $A'$ from $A$ by only adding modifications that are valid according to monotonicity. Since $f$ satisfies monotonicity and $W \in f(N, A, C, k)$, we know that $W \in f(N, A', C, k)$. However, now we observe that $N'$ forms an $\ell$-large and $(\lvert \bigcup_{j \in N'} (A_j \cap W)\rvert + 1)$-cohesive group (with $(A_j \cap W) \cup \{c\}$ being the set of candidates they are cohesive over), while less than $\lvert \bigcup_{j \in N'} (A_j \cap W) \rvert + 1$ candidates are selected from commonly approved candidates. This means that $W$ does not satisfy PJR  
    on $\mathcal I' = (N, A', C, k)$. In other words, $f$ does not refine $\pjr$. 
\end{proof}
\ejrprefc*
\begin{proof}
\textbf{Part 2:} 
    To prove that $\ejrp$ satisfies all axioms consider an instance $\mathcal{I} = (N, A, C,k)$ and let $W\in \ejrp(\mathcal I)$. As discussed in \Cref{prop:lq} EJR+ implies lower quota for party-lists.

    Just as in \Cref{thm:PJRP_Refinement}, for all other four axioms, the proof that EJR+ satisfies them follows the same structure: we start with an instance $\mathcal{I} = (N, A, C, k)$ and $W \in \ejrp(\mathcal I)$. Then, we modify $\mathcal{I}$ to an instance $\mathcal{I}' = (N, A', C', k)$ and need to show that $W \in \ejrp(\mathcal{I}')$. To show this, we assume that we are given a $c \in C' \setminus W$, $\ell \in [k]$ and $\ell$-cohesive group $N' \subseteq N_c$ in the instance $\mathcal{I}'$. We now need to prove that $\left\lvert A'_i \cap W \right\rvert \ge \ell$.

    For independence of losers, we again have $ A_i\cap W = A'_i\cap W $, and thus (since $N'$ also satisfies the EJR+ requirements in $\mathcal{I}$) there must exist an $i \in N'$ with $\lvert A'_i \cap W\rvert \ge \ell$. This shows that $W$ still satisfies EJR+.

    For monotonicity, since the voters only added approvals of candidates from $W$, it also holds that $c\in \bigcap_{i\in N'} (A_i \setminus W)$. Thus, since $W\in \ejrp(\mathcal I)$, there exists some $i \in N'$ with $\lvert A_i \cap W\rvert \ge \ell$.  Since $ \lvert A'_i \cap W\rvert \ge  \lvert A_i \cap W\rvert$ it especially holds that $\lvert A'_i \cap W\rvert \ge \ell$. This proves that there are no witnesses of an EJR+ violation in $\mathcal {I'}$. Therefore, $W \in \ejrp(\mathcal {I'}) $, showing that EJR+ satisfies monotonicity.  
    
    For \ifsv, for all voters $i\in N'$ we have $A'_i\not\subseteq W$, which especially implies $A'_i\not\subseteq A_i\cap W$ and hence $A'_i = A_i$.
    Since $W$ satisfies EJR+ in $\mathcal I$, this implies that $\lvert A'_i \cap W\rvert = \lvert A_i \cap W\rvert \ge \ell $ for some $i \in N'$ and hence EJR+ satisfies robustness to fully satisfied voters.

    For independence of approval swaps, we again get that $N'$ also satisfies the EJR+ requirements in $\mathcal{I}$. Hence, we know that $\ell \le \lvert A_i \cap W\rvert = \lvert A'_i \cap W\rvert$ for some $i \in N'$ and $W$ still satisfies EJR+. As a consequence EJR+ satisfies independence of approval swaps.
\end{proof}
\ejrppjrpcon*
\begin{proof}
    By \Cref{prop:pjr_monotonicity}, $f$ refines PJR+. By \Cref{cor:pjrp_ioas_ejrp}, $f$ therefore refines EJR+.

\end{proof}
\wbiol*
\begin{proof}
    To prove independence of losers, we argue by contraposition: 
    Let $\mathcal I =(N,A,C,k)$,
    $\mathcal I'=\mathcal{I}_{\mid C\setminus \{c\}}$ be given instances, such that $W\subseteq C$, $c\notin W$.
    Then, we have to show that $W\notin f(\mathcal I')$ implies $W\notin f(\mathcal I)$.
    As $f$ is witness-based there exists a set $N' \in w(W, \mathcal{I}')$.
    It is easy to see that $(N', W)$ in $\mathcal{I}'$ is locally embedded into $(N', W)$ in $\mathcal{I}$. By definition of witness functions, hence $N' \in w(W, \mathcal{I})$ and therefore $W \notin f(\mathcal{I})$, as desired. 

    For \ifsv we again consider $N' \in w(W, \mathcal{I})$ and assume that there is a voter $j \in N$ with $A_j \subseteq W$. As $f$ is cohesiveness-based, we know that $\bigcap_{i \in N'} A_i \setminus W \neq \emptyset$ and hence $j \notin N'$. Now, consider a second instance $\mathcal{I}'' = (N, A, C'', k)$ with $A''_i = A_i$ for all $i \in N \setminus \{j\}$. As the approvals for every voter in $N'$ stay the same $(N', W)$ in $\mathcal{I}$ is again locally embedded into $(N', W)$ in $\mathcal{I}''$ and thus remains a witness. It again follows that $W \notin f(\mathcal{I}'')$ and we get that $f$ satisfies \ifsv by the contrapositive.
\end{proof}
\pjrpmergeproof*
\begin{proof}
    Clearly, $(\pjrp, \natWitnessPjrp)$, is cohesiveness-based by definition.

    Let $W$ violate PJR+ on some instance $\mathcal{I}=(N,A,C,k)$ and let $N'\subseteq N_c$ be a natural witness for that with $c\notin W$.
    Then, it holds that $\lvert \bigcup_{i\in N'} A_i\cap W \rvert \frac nk < \lvert N'\rvert$. 

    Consider any profile $A'$ as described in \Cref{def:merge_proof}. Our goal is to show that $N'$ is still a natural witness for a PJR+ violation. Clearly, since $A_i\subseteq A_i'$ for all $i \in N'$, we still have $N'\subseteq N_c$ (and $c\notin W$).
    Secondly, since every $A_i'$ is a subset of $\bigcup_{j\in N'} A_j$ for $i\in N'$, it holds that $$\lvert \bigcup_{i\in N'} A'_i\cap W \rvert \le \lvert \bigcup_{i\in N'} (\bigcup_{j\in N'} A_j) \cap W \rvert = \lvert \bigcup_{i\in N'} A_i\cap W \rvert.$$
    This implies  $\lvert \bigcup_{i\in N'} A_i\cap W \rvert \frac nk \le \lvert \bigcup_{i\in N'} A_i\cap W \rvert \frac nk < \lvert N'\rvert$, as desired.
\end{proof}
\pjrpstr*
\begin{proof}
We prove the statement by contraposition. Let $(f,w)$ be cohesiveness-based and satisfy merge-proofness. Further, let $\mathcal I=(N,A,C,k)$ be given and a committee
$W \in \pjrp(\mathcal I)\setminus f(\mathcal I)$. Our goal is to provide a (party-list) instance on which $f$ is not a lower quota extension.

Since $W\notin f(\mathcal I)$, we have $w(\mathcal I,W)\neq\emptyset$ and can pick a witness $N'\in w(\mathcal I,W)$.

Since $(f,w)$ is cohesiveness-based, there exists a candidate
\(
c \in \bigcap_{i\in N'} \left(A_i\right)\setminus W.
\) Just as in the proof for EJR+ we let
$\ell \coloneqq \left\lfloor k\cdot \frac{|N'|}{n}\right\rfloor.$
Additionally, we set $P \coloneqq \bigcup_{i\in N'} A_i$.
Note that $N'\subseteq N_c$ and $|N'|\ge \ell\cdot \frac{n}{k}$ by definition of $\ell$.
 Since $W\in \pjrp(\mathcal I)$, we get that $|P \cap W|\ge \ell$.

Now, we can construct an instance $\mathcal I_2=(N,A',C,k)$ by setting
\[
A'_i \coloneqq
\begin{cases}
P, & i\in N',\\
A_i, & i\in N\setminus N'.
\end{cases}
\]
Since $A_i \subseteq P = \bigcup_{j\in N'}A_j$ for all $i\in N'$, merge-proofness implies that
$N'\in w(\mathcal I_2,W)$.

Next, we build a party-list instance $\mathcal I_3=(N,A'',C^\star,k)$ as follows.
Let
\[
A''_i \coloneqq
\begin{cases}
P= A_i', & i\in N',\\
W\setminus P , & i\in N\setminus N' \text{ and } W\setminus P \neq\emptyset,\\
P, & i\in N\setminus N' \text{ and } W\setminus P=\emptyset.
\end{cases}
\]
Then, with analogous arguments as for the proof of \Cref{thm:ejrp_coarsening}, $\mathcal I_3$ is a party-list instance with parties given by the non-empty sets among $\{C_1,C_2\}$ and $W$ satisfies lower quota on $I_3$.

However, note that $(N',W)$ in $\mathcal I_2$ can be locally embedded into $(N',W)$ in $\mathcal I_3$ (as no voter $i\in N'$ changes her approval ballot). 
Therefore, witness-rationalizability implies $N'\in w(\mathcal I_3,W)$. In other words, $W\notin f(\mathcal I_3)$, but $W$ satisfies lower quota, hence $f$ is not a lower quota extension. \qedhere
\end{proof}
\section{Single Candidate Justifications}\label{sec:strongIoL}
So far, we considered independence of losers, which states that a candidate outside of a committee dropping out from the election should not negatively affect this committee's proportionality.

Vice versa, if a committee is considered to be proportional within two contexts $C'$ and $C''$, it seems reasonable to demand that it should still be proportional when considering all candidates in $C'\cup C''$.

\begin{definition}
    A proportionality notion $f$ satisfies \emph{strong independence of losers}, if it satisfies independence of losers and additionally
    $W \in f(N, A, C', k)\cap f(N, A, C'', k)$ implies that $W \in f(N, A, C'\cup C'', k)$.
\end{definition}

While intuitive, this notion turns out to be much more restrictive than independence of losers. In fact, it characterizes proportionality notions which only rely on a single candidate outside of $W$ to create a violation.

\begin{theorem}
    A proportionality notion $f$ satisfies strong independence of losers, if and only if the following holds.
    $W\in f(N, A,C,k)$ if and only if there exists no $c\in C \setminus W$ such that $W\notin f(N, A,W\cup \{c\},k)$.   
\end{theorem}
\begin{proof}
    If $W\in f(N, A,C,k)$, then by independence of losers, we obtain $W\in f(N, A,W\cup\{c\},k)$ for all $c\in C\setminus W$. 
    Vice versa, if $W\in f(N, A,W\cup\{c\},k)$ for all $c\in C\setminus W$, then we can repeatedly apply strong independence of losers to obtain $W\in f(N, A,C,k)$.
\end{proof}

This notion is related to the idea of rationalizability in choice theory, which concerns and relates choices made on different feasible sets. From this point of view, strong independence of losers can be decomposed into two parts. Contrary to classical choice theory, the \emph{expansion} part of strong indpendence of losers, which propagates the choice of committees from smaller to larger candidate sets, is much more restrictive than the \emph{contraction} part, which simply corresponds to the regular independence of losers.

\section{Analyzing Near Perfect Representation}\label{app:NPR}
To provide a non-trivial refinement of EJR+ satifying all axioms  from \Cref{ejrp-ref-char}, consider the concept of $g$-representativeness,\footnote{Originally called $f$-representativeness. We rename it here to not overload the variable $f$.} introduced by \citet{BrPe23a} (and generalizing the proportionality degree \citep{Skow21a}).
\begin{definition}[$g$-representativeness \citep{BrPe23a}]
Given an instance $\mathcal I=(N,A,C,k)$ and a function $g\colon \mathbb N^+ \to \mathbb R_{\ge 0}$,
a committee $W$ is \emph{$g$-representative (in $\mathcal I$)} if for every $\ell\in[k]$,
every $c\in C\setminus W$, and every group $N'\subseteq N_c$ with $|N'|\ge \ell\frac{n}{k}$, we have
\(
\frac{1}{|N'|}\sum_{i\in N'} |A_i\cap W| \;\ge\; g(\ell).
\)
\end{definition}

Intuitively, $g$-representativeness gives a lower bound on the average utility of any group of agents approving an unselected candidate. For $\varepsilon > 0$ we denote by $g_{\varepsilon}$ the function $ell \mapsto (\ell - 1 + \varepsilon)$.
For a given instance $\mathcal{I}$ we denote by $\prl(\mathcal{I})$ the set of all $g_\varepsilon$-representative committees for some $\varepsilon > 0$ (with NPR standing for \emph{near perfect representativeness}). If a group of voters of size at least $\ell \frac{n}{k}$ approves a common unselected candidate, this guarantees that their average utility is strictly larger than $\ell - 1$.

It is known that the Proportional Approval Voting (PAV) rule always selects committees from $\prl(\mathcal{I})$ \citep{BrPe23a} and hence this notion is always satisfiable. Further, it is easy to see that any committee $ W \in \prl(\mathcal{I})$ is also in $\ejrp(\mathcal{I})$ for any instance $\mathcal{I}$.\footnote{Since for any $\ell$-large group approving a common unselected candidate, the average utility is larger than $\ell$-1 and since utilities are integer valued, there must exist a voter in this group approving at least $\ell$ candidates of the committee.} We can also show that it satisfies the same set of axioms we used in \Cref{ejrp-ref-char} to weakly characterize EJR+ refinements.
\begin{restatable}{observation}{nprobs}
    NPR satisfies independence of losers, monotonicity, \ifsv, independence of approval swaps, and lower quota for party-lists. Further, $\prl \subsetneq \ejrp$.
    \label{obs:npr}
\end{restatable}
\begin{proof}
    The axiom satisfaction of NPR follows mostly with arguments analogously to those for EJR+. If a committee satisfies NPR then it also does so after any unchosen candidate is deleted from the instance. Hence, NPR satisfies independence of losers. Adding approvals to candidates inside the committee can only increase the average utility of groups, thus showing that $\prl$ satisfies monotonicity. If a voter changes their original ballot $A_i$ to be entirely contained inside $A_i\cap W$, this voter cannot take part in any deviating coalition, proving that \ifsv is satisfied. Similarly, it is easy to see that if a voter changes her ballot to approve the same number of candidates inside $W$ this does not change whether $W$ satisfies NPR.
\begin{wrapstuff}[r,type=figure,width=5cm]
    \centering
    \begin{tikzpicture}
    [yscale=0.6,xscale=0.8,
    voter/.style={anchor=south}]
    
        \foreach \i in {1,...,4}
    		\node[voter] at (\i-0.5, -1) {$\i$};
        
        \draw[fill=magenta!\clrstr] (0, 0) rectangle (4, 1);
        \draw[fill=teal!\strstr, ultra thick] (0, 1) rectangle (4, 2);
        \draw[fill=teal!\strstr, ultra thick] (0, 2) rectangle (3, 3);
        \draw[fill=teal!\strstr, ultra thick] (0, 3) rectangle (2, 4);
        \draw[fill=teal!\strstr, ultra thick] (0, 4) rectangle (1, 5);

        \node at ( 2, 0.5) {$c_{1}$};
        \node at ( 2, 1.5) {$c_{2}$};
        \node at ( 1.5, 2.5) {$c_{3}$};
        \node at ( 1, 3.5) {$c_{4}$};
        \node at ( 0.5, 4.5) {$c_{5}$};
    \end{tikzpicture}
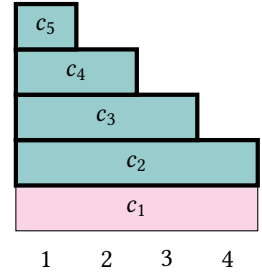
    \caption{Example for EJR+ and NPR.}
    \label{fig:npr}
\end{wrapstuff}
As an average utility of strictly larger than $\ell - 1$ implies that some voter in the group must approve at least $\ell$ candidates in $W$, any committee satisfying $\prl$ also satisfies EJR+ and thus also lower quota for party-lists.
    Finally, to see that $\prl$ is strictly stronger than EJR+ consider the instance depicted in \Cref{fig:npr} with $k = 4$ and committee $\{c_2, \dots, c_5\}$. This committee indeed satisfies EJR+. For any subset of $i$ voters, at least one voter approves $i$ candidates in the outcome, and hence this group could not witness an EJR+ violation. It, however, does not satisfy $\prl$ as for candidate $c_1$ the group $4$-large group $\{1, \dots, 4\}$ on average approves $\frac{4+3+2+1}{4} = 2.5 \le 4-1$ candidates. \qedhere
    \wrapstuffclear  
\end{proof}

Individual discontentment is nontrivial among witness-based notions; we illustrate this by showing that NPR does not satisfy it.
\begin{restatable}{proposition}{nprwit}
    NPR together with its natural witness function is cohesiveness-based. There does not exist any witness function $w$ of NPR such that $(\prl,w)$ satisfies individual discontentment.
    \label{prop:npr_discon}
\end{restatable}
\begin{proof}
    That NPR together with its natural witness function is cohesiveness-based follows from the definition. To see that any witness for NPR violates individual discontentment, let $w$ be any witness for NPR. We will again consider the instance $\mathcal{I}$ depicted in \Cref{fig:npr} for $k = 4$ and with the committee $\{c_2, \dots, c_5\}$. As argued earlier, this committee does not satisfy $\prl$. Let $N' \in w(\mathcal{I}, W)$ be the witness for this.  Let $s\coloneqq \lvert N'\rvert$ and choose $j\in N'$ maximizing $\lvert A_j\cap W\rvert$
    (in the profile of \Cref{fig:npr}, this is the voter with smallest index in $N'$).
    Denote $u\coloneqq |A_j\cap W|$. In \Cref{fig:npr} we have $u=4-j+1$, and since $N'$ contains $s$ distinct voters
    among $\{1,2,3,4\}$, it follows that $u\ge s$.
    Now define a new instance $\hat{\mathcal I}=(N,\hat A,C,k)$ as follows:
    for all $i\in N'$ set $\hat A_i := A'_i$ (so every voter in $N'$ approves exactly $A_j$),
    and for all $i\in N\setminus N'$ set $\hat A_i = W$.
    In particular, in $\hat{\mathcal I}$ the only candidate outside $W$ is still $c_1$, and its supporters are exactly $N'$
    (since voters outside $N'$ approve only candidates in $W$).
    We claim that $W\in \prl(\hat{\mathcal I})$. To see this, consider the only unselected candidate $c_1\notin W$ and any
    $\ell\in[k]$ and any group $S\subseteq N_{c_1}=N'$ with $|S|\ge \ell\cdot \frac{n}{k}=\ell$ (here $n=k=4$).
    Every voter in $N'$ has utility $|A_j\cap W|=u$ in $\hat{\mathcal I}$, hence
    \[
      \frac{1}{|S|}\sum_{i\in S} |\hat A_i\cap W| \;=\; u \;\ge\; s \;\ge\; \ell \;>\; \ell-1,
    \]
    so the near-perfect representativeness condition holds for $c_1$ and all feasible $\ell$ and $S$.
    Therefore $W\in \prl(\hat{\mathcal I})$, and by item~(i) of the witness-function definition we must have
    $w(\hat{\mathcal I},W)=\emptyset$.

    On the other hand, by individual discontentment we have $N'\in w(\mathcal I',W)$.
    Moreover, $(N',W)$ in $\mathcal I'$ can be locally embedded into $(N',W)$ in $\hat{\mathcal I}$ via the identity maps
    (since all ballots of voters in $N'$ and all candidates in $W\cup \bigcup_{i\in N'}A'_i$ are unchanged).
    Hence, by item~(ii) of the witness-function definition, it follows that $N'\in w(\hat{\mathcal I},W)$,
    contradicting $w(\hat{\mathcal I},W)=\emptyset$.
\wrapstuffclear
\end{proof}

\section{Robustness of the Characterizations}
\label{app:robust}
\subsection{Necessity of the Axioms for EJR+}
We again show that both characterizations obtained for EJR+ are robust to leaving out any of the axioms used to show them. 
\subsubsection{Necessity of the Axioms in \Cref{cor:ejrpluschar}}
We begin by showing that the axioms used in \Cref{cor:ejrpluschar} to characterize EJR+ are all necessary. Removing any one of them leads to a notion different from EJR+ that still satisfies all of them. 
\paragraph{Leaving out Cohesiveness-Basedness}
First, without being cohesiveness-based, we consider an notion requiring that committees satisfy EJR+ and a variant of Pareto optimality at the same time. We say that for a given instance $\mathcal{I} = (N, A, C, k)$ a committee $W$ is \emph{externally Pareto optimal} if there is no other committee $W'$ with $W \cap W' = \emptyset$ satisfying $\lvert A_i \cap W\rvert < \lvert A_i \cap W'\rvert$ for every voter $i \in N$. We let $\mathrm{exPareto}(\mathcal{I})$ be the set of externally Pareto optimal committees and consider the ``proportionality'' notion $\ejrppar (\mathcal{I}) \coloneqq \ejrp(\mathcal{I}) \cap \mathrm{exPareto}(\mathcal{I})$. We note that $\ejrppar$ is always satisfiable as the PAV voting rule satisfies both Pareto optimality (and thus also external Pareto optimality) and EJR+ \citep{LaSk22a, BrPe23a}. We define the natural witness of $\ejrppar$ to be any EJR+ witness together with $N$ if the given committee does is not externally Pareto optimal. Indeed this is enough to satisfy all axioms used to characterize EJR+ except for being cohesiveness-based.

\begin{proposition}
    The proportionality notion of $\ejrppar$ with its natural witness is witness-based, satisfies individual discontentment, lower quota extension, monotonicity, and independence of approval swaps. There does not exist a cohesiveness-based witness function for $\ejrppar$. 
\end{proposition}
\begin{proof}
We note that the natural witness is by definition a witness for $\ejrppar$. To see that it satisfies individual discontentment, consider any instance $\mathcal{I}$, committee $W$, and $N' \in w(W, \mathcal{I})$. If $N'$ is a witness due to an $\ejrp$ violation, it stays a witness, even after exchanging the approval sets of voters in $N'$. If $N'$ was due to $W$ not being externally Pareto optimal, there must exist a disjoint committee $W'$ with $\lvert A_i \cap W\rvert < \lvert A_i \cap W'\rvert$ for all $i \in N'$. This of course still holds after replacing all approval ballots by that of a single voter. Hence, $\ejrppar$ satisfies individual discontentment.

For the next axioms let $\mathcal{I} = (N, A, C, k)$ be any instance and $W$ be a committee satisfying both $\ejrp$ and external Pareto optimality. 

For monotonicity, consider an approval profile $A' \supseteq A$ with $A'_i \setminus A_i \subseteq W$ for all $i \in N$ and the instance $\mathcal{I}_1 = (N, A', C, k)$. First, since $W$ satisfied EJR+ and since EJR+ is monotone $W$ also satisfies EJR+ in $\mathcal{I}_1$. To see that it is also externally Pareto optimal in $\mathcal{I}_1$, consider any other committee $W'$ with $W \cap W' = \emptyset$. Since $W$ is externally Pareto optimal in $\mathcal{I}$ we get that 
\[
\lvert A'_i \cap W'\rvert = \lvert A_i \cap W'\rvert \le \lvert A_i \cap W\rvert \le \lvert A'_i \cap W\rvert
\] for every $i \in N$ and therefore $W$ is also externally Pareto optimal in $\mathcal{I}_1$. We thus get that $\ejrppar$ is monotone.

For independence of approval swaps\footnote{We note that this is where we require external Pareto optimality and not ``normal'' Pareto optimality, as approval swaps could lead to a Pareto optimal committee becoming non-Pareto optimal.} consider an approval profile $A''$ with $A''_i \setminus W = A_i \setminus W$ and $\lvert A''_i \cap W\rvert = \lvert A_i \cap W\rvert$ for all $i \in N$. Let $\mathcal{I}_2 = (N, A'', C, k)$. Again, since EJR+ is independent of approval swaps, we also know that $W$ satisfies EJR+ in $\mathcal{I}_2$. For external Pareto optimality, let $W'$ be a committee with $W' \cap W = \emptyset$. Since $W$ is externally Pareto optimal in $\mathcal{I}$ we get that 
\[
\lvert A''_i \cap W'\rvert = \lvert A_i \cap W'\rvert \le \lvert A_i \cap W\rvert = \lvert A''_i \cap W\rvert
\] and hence, $W$ also satisfies external Pareto optimality in $\mathcal{I}_2$.

To see that $\ejrppar$ is a lower quota extension, let $W$ be any committee satisfying lower quota on a party-list instance. First, we note that we required every party in a party-list to have at least one voter. Hence, if we have a different committee $W'$ in which at least one voter gets a higher approval score than in $W$ there must exist a different voter getting a lower approval score (as every party is approved by at least one voter). Thus, there can be no other committee Pareto dominating $W$ and thus, every committee satisfying lower quota is also Pareto optimal. As a consequence, $\ejrppar$ is a lower quota extension. 
\begin{wrapstuff}[r,type=figure,width=5cm]
    \centering
    \begin{tikzpicture}
    [yscale=0.6,xscale=0.8,
    voter/.style={anchor=south}]
    
        \foreach \i in {1,...,4}
    		\node[voter] at (\i-0.5, -1) {$\i$};
        
        \draw[fill=teal!\clrstr] (0, 0) rectangle (3, 1);
        \draw[fill=magenta!\strstr, ultra thick] (1, 2) rectangle (2, 3);
        \draw[fill=teal!\clrstr] (1, 1) rectangle (4, 2);
        \draw[fill=magenta!\strstr, ultra thick] (2, 2) rectangle (3, 3);
        
        \node at ( 1.5, 0.5) {$c_{1}$};
        \node at ( 2.5, 1.5) {$c_{2}$};
        \node at ( 2.5, 2.5) {$c_{4}$};
        \node at ( 1.5, 2.5) {$c_{3}$};
    \end{tikzpicture}
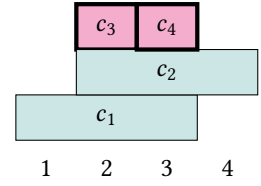
    \caption{Example showing that $\ejrppar$ does not admit a cohesiveness-based witness.}
    \label{fig:ejrppar}
\end{wrapstuff}
Finally, to show that $\ejrppar$ does not admit a cohesiveness-based witness, consider the instance depicted in \Cref{fig:ejrppar}. Here, for $k = 2$ the committee $\{c_3, c_4\}$ satisfies EJR+, but not external Pareto optimality. There are (up to symmetry) four different possible cohesiveness-based witnesses: $\{1\}, \{2\}, \{1,2\}, \{1,2,3\}$. Now, consider the instance, in which voter $4$ approves only candidate $c_4$ (instead of $c_2$). In this instance the committee $\{c_3, c_4\}$ is externally Pareto optimal, as the voter $4$ cannot improve. However, independent of which witness out of $\{1\}, \{2\}, \{1,2\}, \{1,2,3\}$ we would have chosen, this witness could have been locally embedded into this new instance, showing that it would not have been a valid witness. Thus, $\ejrppar$ does not admit a cohesiveness-based witness. \qedhere
\wrapstuffclear
\end{proof}
\paragraph{Leaving out Individual Discontentment}
For individual discontentment, we have already shown in \Cref{prop:npr_discon} that $\prl$ does not satisfy it for \emph{any} possible witness functions. However, akin to EJR+, all other axioms are satisfied by $\prl$.
\begin{proposition}\sloppy
    The proportionality notion of $\prl$ with its natural witness is witness-based, cohesiveness-based, satisfies lower quota extension, monotonicity, and independence of approval swaps. There does not exist a witness for $\prl$ satisfying individual discontentment. 
\end{proposition}
\begin{proof}
    The positive results follow from \Cref{obs:npr} and are in general analogous to EJR+. \Cref{prop:npr_discon} shows that $\prl$ has no witness function satisfying individual discontentment.
\end{proof}
\paragraph{Leaving out Lower Quota Extension}
One trivial proportionality notion satisfying all axioms except for being a lower quota extension is the empty notion (together with the witness function selecting all subsets of $N$). A non-trivial proportionality notion satisfying all the axioms is Droop-EJR+ \citep{CaEl25a}, i.e., the EJR+ variant using the strictly stronger Droop quota. 
\begin{definition}[Droop-EJR+]
    A committee $W$ satisfies \emph{Droop-extended justified representation plus (Droop-EJR+)} if for every $\ell \in [k]$ unselected candidate $c \notin W$, and group $N' \subseteq N_c$ of voters of size $\lvert N'\rvert > \ell \frac{n}{k+1}$ there exists a voter $i \in N'$ with $\lvert A_i \cap W\rvert \ge \ell$.
    
\end{definition}
\begin{proposition}
    Droop-EJR+ together with its natural witness is cohesiveness-based, and satisfies individual discontentment, monotonicity, and independence of approval swaps. It is not a lower quota extension.
\end{proposition}
\begin{proof}
    The positive results follow analogously to EJR+(with the Hare quota). 
To see that Droop-EJR+ is not a lower quota extension, consider a party-list instance with two parties $C_1$ and $C_2$ with the first party receiving $3$ votes and the second part $2$ votes. Here, for $k = 2$, the Droop quota is $\frac{5}{3} < 2$ while the Hare quota is $\frac{5}{2} > 2$. Hence, allocating $2$ seats to the first party satisfies lower quota, but violates Droop-EJR+.
\end{proof}
We note that this also shows that replacing $f$ being a lower quota extension by the weaker property of lower-quota for party-lists does not lead to a characterization of EJR+.
\paragraph{Leaving out Monotonicity}
For monotonicity consider the following variant of EJR+. 
\begin{definition}
     A committee $W$ satisfies equal-EJR+ if for every $\ell \in [k]$, unselected candidate $c \notin W$, and $\ell$-large group $N' \subseteq N_c$ of voters with $\lvert A_i \cap W\rvert = \lvert A_j \cap W\rvert$ for all $i, j\in N'$ it holds that $\lvert A_i \cap W\rvert \ge \ell$ for all $i \in N'$.
\end{definition}
Intuitively, for equal-EJR+ adding approvals to a given committee can lead to new witnesses and hence, it does not satisfy monotonicity.
\begin{proposition}
    Equal-EJR+ together with its natural witness is cohesiveness-based, and satisfies individual discontentment, independence of approval swaps, and is a lower quota extension. It does not satisfy monotonicity.
\end{proposition}
\begin{proof}
    Throughout the proof let $\mathcal{I} = (N,A, C, k)$ be any instance.
    Let $w$ be the natural witness of equal-EJR+. First, $w$ being cohesiveness-based follows from its definition. To see that $w$ satisfies individual discontentment let $N' \in w(\mathcal I, W)$. By the definition of equal-EJR+ this set of voters stays a witness, even after replacing all of them by a copy of a single one.

    For independence of approval swaps, we observe that any witness for equal-EJR+ stays a witness after any approval swap, as the number of approved candidates on the committee stays the same. Further, on party-list instances, any set of voters approving a given party must by definition have the same number of approved candidates on the committee, and thus on party-list instances equal-EJR+ and EJR+ are equivalent (and thus, equal-EJR+ is also a lower quota extension). \begin{wrapstuff}[r,type=figure,width=5cm]
    \centering
    \begin{tikzpicture}
    [yscale=0.6,xscale=0.8,
    voter/.style={anchor=south}]
    
        \foreach \i in {1,...,3}
    		\node[voter] at (\i-0.5, -1) {$\i$};
        
        \draw[fill=magenta!\strstr, ultra thick] (0, 0) rectangle (1, 1);
        \draw[fill=magenta!\strstr, ultra thick] (1, 0) rectangle (2, 1);
        \draw[fill=teal!\clrstr] (0, 1) rectangle (3, 2);
        
        \node at ( 0.5, 0.5) {$c_{1}$};
        \node at ( 1.5, 0.5) {$c_{2}$};
        \node at ( 1.5, 1.5) {$c_{3}$};
    \end{tikzpicture}
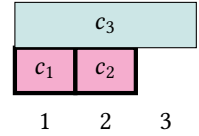
    \caption{Example showing that equal-EJR+ does not satisfy monotonicity.}
    \label{fig:equalejrp}
\end{wrapstuff}
To see that equal-EJR+ fails monotonicity consider the instance displayed in \Cref{fig:equalejrp}. Here, for $k = 2$ the committee $\{c_1, c_2\}$ satisfies equal-EJR+. However, if voter $3$ adds an approval to candidate $c_2$ the committee no longer satisfies equal-EJR+ as now voters $1,2$, and $3$ could deviate with candidate $c_3$. This same instance also witnesses that equal-EJR+ is a weaker notion than EJR+ as the committee $\{c_1, c_2\}$ does not satisfy EJR+.\qedhere

\wrapstuffclear
    
\end{proof}
\paragraph{Leaving out Independence of Approval Swaps}
Finally, for independence of approval-swaps, we note that PJR+ satisfies all remaining axioms.

\subsubsection{Necessity of the Axioms in \Cref{ejrp-ref-char}}
In the previous part, we have already shown that removing either monotonicity or independence of approval-swaps leads to a strictly weaker axiom than EJR+ (equal-EJR+ or PJR+ respectively). We continue with the remaining axioms used to show \Cref{ejrp-ref-char}.
\paragraph{Leaving out \IFSV.}

Without \ifsv, we can define a weakening of EJR+ which satisfies the remaining three axioms. Instead of considering all large enough subsets $N'\subseteq N_c$ as potential violation witnesses, it only considers the supports $N_c$ for $c\in C\setminus W$.
\begin{definition}[Weak-EJR+]
    A committee $W$ satisfies weak-EJR+ if there is no candidate $c \in C\setminus W$ and $\ell \in [k]$ such that $\lvert N_c\rvert \ge \ell \cdot\frac{n}{k}$ and 
    \[
    \left\lvert A_i \cap W \right\rvert < \ell \quad \text{ for all } i \in N_c.
    \]
\end{definition}
Indeed, weak-EJR+ satisfies monotonicity, independence of losers, and lower quota, but not \ifsv. 
\begin{proposition}
    The proportionality notion of weak-EJR+ satisfies monotonicity, independence of losers, independence of approval swaps, and lower quota. It does not satisfy \ifsv. There exists an instance in which weak-EJR+ selects strictly more committees than EJR+.
\end{proposition}
\begin{proof}
    Independence of losers follows in the same fashion as for EJR+: deleting candidates outside the committee can only create less deviating coalitions. For independence of approval swaps, we observe that the only relevant part about the committee $W$ is how many candidates each voter approves in it. Similarly, adding approvals to candidates inside the committee cannot create any new deviating coalitions, and thus, weak-EJR+ also satisfies monotonicity.
    
    For party-list instances, a natural witness of $\ejrp$, say $N'\subseteq N_c$ for some $c\in C\setminus W$, can always be extended to include the entire set of voters $N_c$ voting for the party containing $c$. Thus, weak-EJR+ and EJR+ are equivalent for party-list instances (and therefore weak-EJR+ also satisfies lower quota for party-lists). 
\begin{wrapstuff}[r,type=figure,width=5cm]
    \centering
    \begin{tikzpicture}
    [yscale=0.6,xscale=0.8,
    voter/.style={anchor=south}]
    
        \foreach \i in {1,...,4}
    		\node[voter] at (\i-0.5, -1) {$\i$};
        
        \draw[fill=teal!\clrstr] (0, 0) rectangle (3, 1);
        \draw[fill=magenta!\strstr, ultra thick] (2, 1) rectangle (4, 2);
        \draw[fill=magenta!\strstr, ultra thick] (2, 2) rectangle (4, 3);
        \draw[fill=magenta!\strstr, ultra thick] (2, 3) rectangle (4, 4);
        \draw[fill=magenta!\strstr, ultra thick] (2, 4) rectangle (4, 5);
        
        \node at ( 1.5, 0.5) {$c_{1}$};
        
        \node at ( 3, 1.5) {$c_{2}$};
        \node at ( 3, 2.5) {$c_{3}$};
        \node at ( 3, 3.5) {$c_{4}$};
        \node at ( 3, 4.5) {$c_{5}$};
    \end{tikzpicture}
    \caption{Weak-EJR+ violates \ifsv.}
    \label{fig:weakEJRP}
\end{wrapstuff}
    To see that weak-EJR+ does not satisfy \ifsv consider the instance and committee depicted in \Cref{fig:weakEJRP} for $k = 4$. The committee $\{c_2, \dots, c_5\}$ satisfies weak-EJR+, as the voters approving candidate $c_1$ cannot deviate due to voter $3$. If, however, voter $3$ would change their preferences to only approve $\{c_2, \dots, c_5\}$ this committee would no longer satisfy weak-EJR+, as now voters $1$ and $2$ together could deviate with candidate $c_1$. This example, additionally also shows that weak-EJR+ is strictly weaker than EJR+ as the committee does not satisfy EJR+ due to the voters $1$ and $2$. \qedhere
    \wrapstuffclear
\end{proof}
\paragraph{Leaving out Independence of Losers} Secondly, we consider independence of losers. Here, we define a weakening of EJR+, we term difference-EJR+. Difference-EJR+ requires that the candidates \emph{outside} the committee must be the same.
\begin{definition}[difference-EJR+]
    A committee $W$ satisfies difference-EJR+ if there is no candidate $c \notin W$, $\ell \in [k]$, and group $N' \subseteq N_c$ such that $\lvert N'\rvert \ge \ell\frac{n}{k}$, $A_i \setminus W = A_j \setminus W$ for all $i,j \in N'$ and 
    \[
    \left\lvert \bigcup_{i \in N'} A_i \cap W \right\rvert < \ell.
    \]
\end{definition}
Indeed, difference-EJR+ satisfies monotonicity, \ifsv, independence of approval swaps, and lower quota, but not independence of losers. 
\begin{proposition}
    The proportionality notion of difference-EJR+ satisfies monotonicity, \ifsv,independence of approval swaps, and lower quota. It does not satisfy independence of losers. There exists an instance in which difference-EJR+ selects strictly more committees than EJR+.
\end{proposition}
\begin{proof}
    Monotonicity, independence of approval swaps, and \ifsv again follow the same way as for EJR+. Adding approvals and letting voters switch their approvals can only create less deviating coalitions. 
            \begin{wrapstuff}[r,type=figure,width=5cm]
    \centering
    \begin{tikzpicture}
    [yscale=0.6,xscale=0.8,
    voter/.style={anchor=south}]
    
        \foreach \i in {1,...,4}
    		\node[voter] at (\i-0.5, -1) {$\i$};
        
        \draw[fill=teal!\clrstr] (0, 0) rectangle (2, 1);
        \draw[fill=teal!\clrstr] (0, 1) rectangle (1, 2);
        \draw[fill=magenta!\strstr, ultra thick] (2, 0) rectangle (4, 1);
        \draw[fill=magenta!\strstr, ultra thick] (2, 1) rectangle (4, 2);
        
        \node at ( 1, 0.5) {$c_{1}$};
        \node at ( 0.5, 1.5) {$c_{2}$};
        \node at ( 3, 0.5) {$c_{3}$};
        \node at ( 3, 1.5) {$c_{4}$};
    \end{tikzpicture}
    \caption{Example for EJR+ and difference-EJR+.}
    \label{fig:differenceEJR}
\end{wrapstuff}
    Similarly, as for earlier axioms, for party-list instances, difference-EJR+ and EJR+ are equivalent as the requirement $A_i \setminus W = A_j \setminus W$ is automatically satisfied for voters voting for the same party. 
To see that difference-EJR+ indeed does not satisfy independence of losers consider the instance depicted in \Cref{fig:differenceEJR} for $k = 2$. 
Here, the committee $\{c_3, c_4\}$ satisfies difference-EJR+, as voters $1$ and $2$ together cannot deviate with $c_1$. However, if candidate $c_2$ leaves the instance, voters $1$ and $2$ can deviate, thus leading to a violation of difference-EJR+ (and therefore showing that difference-EJR+ does not satisfy independence of losers). This example also shows that EJR+ is strictly stronger than difference-EJR+, as the committee $\{c_3, c_4\}$ does not satisfy EJR+ (due to the voters $1$ and $2$).\qedhere

\wrapstuffclear
\end{proof}
\paragraph{Leaving out Lower Quota for Party-Lists}

Finally, without lower quota, even the universal axiom $\mathrm{univ}$ satisfies all three remaining axioms, but is a weaker notion than EJR+.

\begin{observation}
    The universal axiom $\mathrm{univ}$ satisfies monotonicity, \ifsv, independence of approval swaps, and independence of losers, however it does not satisfy lower quota for party-lists.
\end{observation}
\begin{proof}
    Monotonicity, \ifsv, and independence of losers simply follow from the fact that all committees are in $\mathrm{univ}$. Since there exist committees in party-list instances not satisfying lower quota for party-lists, this also shows that $\mathrm{univ}$ does not satisfy lower quota for party-lists.
\end{proof}

\subsection{Necessity of the Axioms for PJR+} 
For PJR+ we can use nearly the same proportionality notions as for EJR+ (just changing the ``deviation-type'' from EJR-like to PJR-like).
\label{app:nec:pjr}
\subsection{Necessity of the Axioms in \Cref{cor:pjrp_Char}}
\paragraph{Leaving out Cohesiveness-Basedness}
As we show in \Cref{app:sub-core} the Sub-Core satisfies all axioms except for having a cohesiveness-based witness function.
\paragraph{Leaving out Merge-Proofness}
As shown previously EJR+ satisfies all remaining axioms.
\paragraph{Leaving our Lower Quota Extension}
Similar to EJR+ the empty notion satisfies everything except being a lower quota extension. An alternative option would be Droop-PJR+.
\paragraph{Leaving out Monotonicity.} Without monotonicity, all three remaining axioms are indeed satisfied by PJR (see also \Cref{sec:compliance}). However, they do not characterize PJR, as we can define a weakening of PJR satisfying independence of losers, lower quota for party-lists, and \ifsv. 

\begin{definition}[Overlap-PJR]
    A committee $W$ satisfies overlap-PJR if there is no $\ell \in [k]$ and $\ell$-cohesive group of voters $N'$ of size $\lvert N'\rvert \ge \ell \frac{n}{k}$ such that $A_i \cap W = A_j \cap W$ for all $i,j \in N'$ and $\lvert A_i \cap W\rvert < \ell$ for any $i \in N'$.
    \label{def:overlap}
\end{definition}

Indeed, overlap-PJR is strictly weaker than PJR (and therefore also strictly weaker than PJR+) and does not satisfy monotonicity, but all other axioms.

\begin{proposition}
    Overlap-PJR is a lower quota extension. Together with its natural witness function overlap-PJR is cohesiveness-based and satisfies merge-proofness.
\end{proposition}
\begin{proof}
 Merge-proofness follows in the same way as for PJR+. Further, together with its natural witness overlap-PJR is cohesiveness-based by definition.
 
 To see that overlap-PJR is a lower quota extension, we observe that the $A_i \cap W = A_j \cap W$ condition is satisfied for all voters voting for the same party in a party-list election, and hence overlap-PJR is equivalent to PJR for these instances. 
To see that overlap-PJR is indeed strictly weaker than PJR (and thus also strictly weaker than PJR+) and does not satisfy monotonicity consider the instance depicted in \Cref{fig:overlap} with $k = 4$. 
\begin{wrapstuff}[r,type=figure,width=5cm]
    \centering
    \begin{tikzpicture}
    [yscale=0.6,xscale=0.8,
    voter/.style={anchor=south}]
    
        \foreach \i in {1,...,4}
    		\node[voter] at (\i-0.5, -1) {$\i$};
        
        \draw[fill=teal!\clrstr] (0, 0) rectangle (2, 1);
        \draw[fill=teal!\clrstr] (0, 1) rectangle (2, 2);
        \draw[fill=teal!\strstr, ultra thick] (0, 2) rectangle (1, 3);
        \draw[fill=magenta!\strstr, ultra thick] (2, 0) rectangle (4, 1);
        \draw[fill=magenta!\strstr, ultra thick] (2, 1) rectangle (4, 2);
        \draw[fill=magenta!\strstr, ultra thick] (2, 2) rectangle (4, 3);
        
        \node at ( 1, 0.5) {$c_{1}$};
        \node at ( 1, 1.5) {$c_{2}$};
        \node at ( 0.5, 2.5) {$c_{3}$};
        \node at ( 3, 0.5) {$c_{4}$};
        \node at ( 3, 1.5) {$c_{5}$};
        \node at ( 3, 2.5) {$c_{6}$};
    \end{tikzpicture}
    \caption{Example for PJR $\neq$ overlap-PJR violating monotonicity.}
    \label{fig:overlap}
\end{wrapstuff}
The depicted committee $\{c_3, c_4, c_5, c_6\}$ does not satisfy PJR as the first two voters constitute a $2$-large and $2$-cohesive group who only received a single candidate in the committee. It does, however, still satisfy overlap-PJR, as $A_1 \cap \{c_3, c_4, c_5, c_6\} = \{c_3\}  \neq \emptyset = A_2 \cap \{c_3, c_4, c_5, c_6\}.$ If the second voter were to approve candidate $c_3$, however, this would even lead to an overlap-PJR violation. Thus, overlap-PJR does not satisfy monotonicity. \qedhere
\wrapstuffclear
\end{proof}

\subsection{Necessity of the Axioms in \Cref{thm:PJRP_Refinement}}
\paragraph{Leaving out \IFSV.}

Without \ifsv, we can define a weakening of PJR+ which satisfies the remaining three axioms. Instead of considering all large enough subsets $N'\subseteq N_c$ as potential violation witnesses, it only considers the supports $N_c$ for $c\in C\setminus W$.
\begin{definition}[Weak-PJR+]
    A committee $W$ satisfies weak-PJR+ if there is no candidate $c \in C\setminus W$ and $\ell \in [k]$ such that $\lvert N_c\rvert \ge \ell \cdot\frac{n}{k}$ and 
    \[
    \left\lvert \bigcup_{i \in N_c} A_i \cap W \right\rvert < \ell.
    \]
\end{definition}
Indeed, weak-PJR+ satisfies monotonicity, independence of losers, and lower quota, but not \ifsv. 
\begin{proposition}
    The proportionality notion of weak-PJR+ satisfies monotonicity, independence of losers, and lower quota. It does not satisfy \ifsv. There exists an instance in which weak-PJR+ selects strictly more committees than PJR+.
\end{proposition}
\begin{proof}
    Independence of losers again follows in the same fashion as for PJR+: deleting candidates outside the committee can only create less deviating coalitions. Similarly, adding approvals to candidates inside the committee cannot create any new deviating coalitions, and thus, weak-PJR+ also satisfies monotonicity.
    
    For party-list instances, deviating coalitions can always be extended to include the entire set of voters voting for a party, and thus, weak-PJR+ and PJR+ are equivalent for party-list instances (and therefore weak-PJR+ also satisfies lower quota for party-lists). 
\begin{wrapstuff}[r,type=figure,width=5cm]
    \centering
    \begin{tikzpicture}
    [yscale=0.6,xscale=0.8,
    voter/.style={anchor=south}]
    
        \foreach \i in {1,...,4}
    		\node[voter] at (\i-0.5, -1) {$\i$};
        
        \draw[fill=teal!\clrstr] (0, 0) rectangle (3, 1);
        \draw[fill=magenta!\strstr, ultra thick] (2, 1) rectangle (4, 2);
        \draw[fill=magenta!\strstr, ultra thick] (2, 2) rectangle (4, 3);
        \draw[fill=magenta!\strstr, ultra thick] (2, 3) rectangle (4, 4);
        \draw[fill=magenta!\strstr, ultra thick] (2, 4) rectangle (4, 5);
        
        \node at ( 1.5, 0.5) {$c_{1}$};
        
        \node at ( 3, 1.5) {$c_{2}$};
        \node at ( 3, 2.5) {$c_{3}$};
        \node at ( 3, 3.5) {$c_{4}$};
        \node at ( 3, 4.5) {$c_{5}$};
    \end{tikzpicture}
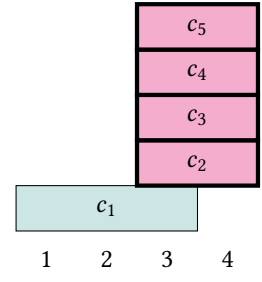
    \caption{Example for weak-PJR+.}
    \label{fig:weakPJRP}
\end{wrapstuff}
    To see that weak-PJR+ does not satisfy \ifsv consider the instance and committee depicted in \Cref{fig:weakPJRP} for $k = 4$. The committee $\{c_2, \dots, c_5\}$ satisfies weak-PJR+, as the voters approving candidate $c_1$ cannot deviate due to voter $3$. If, however, voter $3$ would change their preferences to only approve $\{c_2, \dots, c_5\}$ this committee would no longer satisfy weak-PJR+, as now voters $1$ and $2$ together could deviate with candidate $c_1$. This example, additionally also shows that weak-PJR+ is strictly weaker than PJR+ as the committee does not satisfy PJR+ due to the voters $1$ and $2$. \qedhere
    \wrapstuffclear
\end{proof}
\paragraph{Leaving out Independence of Losers} Thirdly, we consider independence of losers. Here, we define a weakening of PJR+, we term difference-PJR+. While overlap-PJR (\Cref{def:overlap}) required that the candidates approved inside the committee must be the same for a deviating coalition, difference-PJR+ requires that the candidates \emph{outside} the committee must be the same.
\begin{definition}[difference-PJR+]
    A committee $W$ satisfies difference-PJR+ if there is no candidate $c \notin W$, $\ell \in [k]$, and group $N' \subseteq N_c$ such that $\lvert N'\rvert \ge \ell\frac{n}{k}$, $A_i \setminus W = A_j \setminus W$ for all $i,j \in N'$ and 
    \[
    \left\lvert \bigcup_{i \in N'} A_i \cap W \right\rvert < \ell.
    \]
\end{definition}
Indeed, difference-PJR+ satisfies monotonicity, \ifsv, and lower quota, but not independence of losers. 
\begin{proposition}
    The proportionality notion of difference-PJR+ satisfies monotonicity, \ifsv, and lower quota. It does not satisfy independence of losers. There exists an instance in which difference-PJR+ selects strictly more committees than PJR+.
\end{proposition}
\begin{proof}
    \begin{wrapstuff}[r,type=figure,width=5cm]
    \centering
    \begin{tikzpicture}
    [yscale=0.6,xscale=0.8,
    voter/.style={anchor=south}]
    
        \foreach \i in {1,...,4}
    		\node[voter] at (\i-0.5, -1) {$\i$};
        
        \draw[fill=teal!\clrstr] (0, 0) rectangle (2, 1);
        \draw[fill=teal!\clrstr] (0, 1) rectangle (1, 2);
        \draw[fill=magenta!\strstr, ultra thick] (2, 0) rectangle (4, 1);
        \draw[fill=magenta!\strstr, ultra thick] (2, 1) rectangle (4, 2);
        
        \node at ( 1, 0.5) {$c_{1}$};
        \node at ( 0.5, 1.5) {$c_{2}$};
        \node at ( 3, 0.5) {$c_{3}$};
        \node at ( 3, 1.5) {$c_{4}$};
    \end{tikzpicture}
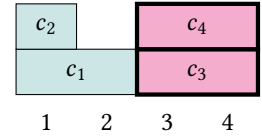
    \caption{Example for PJR+ and difference-PJR+.}
    \label{fig:difference}
\end{wrapstuff}
Monotonicity and \ifsv again follow the same way as for PJR+. Adding approvals and letting voters switch their approvals can only create less deviating coalitions. Similarly, as for earlier axioms, for party-list instances, difference-PJR+ and PJR+ are equivalent as the requirement $A_i \setminus W = A_j \setminus W$ is automatically satisfied for voters voting for the same party. 
To see that difference-PJR+ indeed does not satisfy independence of losers consider the instance depicted in \Cref{fig:difference} for $k = 2$. Here, the committee $\{c_3, c_4\}$ satisfies difference-PJR+, as voters $1$ and $2$ together cannot deviate with $c_1$. However, if candidate $c_2$ leaves the instance, voters $1$ and $2$ can deviate, thus leading to a violation of difference-PJR+ (and therefore showing that difference-PJR+ does not satisfy independence of losers). This example also shows that PJR+ is strictly stronger than difference-PJR+, as the committee $\{c_3, c_4\}$ does not satisfy PJR+ (due to the voters $1$ and $2$).
\wrapstuffclear
\end{proof}
\paragraph{Leaving out Lower Quota for Party-Lists}

Finally, without lower quota, even the universal axiom $\mathrm{univ}$ satisfies all three remaining axioms, but is a weaker axiom than PJR+.

\begin{observation}
    The universal axiom $\mathrm{univ}$ satisfies monotonicity, \ifsv, and independence of losers, however it does not satisfy lower quota for party-lists.
\end{observation}
\begin{proof}
    Monotonicity, \ifsv, and independence of losers simply follow from the fact that all committees are in $\mathrm{univ}$. Since there exist committees in party-list instances not satisfying lower quota for party-lists, this also shows that $\mathrm{univ}$ does not satisfy lower quota for party-lists.
\end{proof}
\section{Witness-Basedness, Neutrality, and Anonymity}
\label{app:wb}

We introduce the necessary axioms to characterize witness-based proportionality notions.

\begin{definition}[Anonymity]
A proportionality notion $f$ satisfies \emph{anonymity} if for every instance
$\mathcal I=(N,A,C,k)$, every bijection $\pi\colon N\to N$, and the relabeled profile
$A^\pi$ defined by $A^\pi_{\pi(i)} := A_i$ for all $i\in N$, it holds that
\[
W\in f(N,A,C,k)\quad\text{if and only if}\quad W\in f(N,A^\pi,C,k)
\]
for every committee $W\subseteq C$.
\end{definition}

\begin{definition}[Neutrality]
A proportionality notion $f$ satisfies \emph{neutrality} if for every instance
$\mathcal I=(N,A,C,k)$, every bijection $\sigma\colon C\to \hat C$, and the relabeled profile
$\sigma(A)$ defined by $\sigma(A)_i := \{\sigma(c)\mid c\in A_i\}$ for all $i\in N$, it holds that
\[
W\in f(N,A,C,k)\quad\text{if and only if}\quad \sigma(W)\in f\bigl(N,\sigma(A),\hat C,k\bigr)
\]
for every committee $W\subseteq C$.
\end{definition}

The following axiom complements independence of losers.
\begin{definition}[Independence of Unapproved Candidates]
A proportionality notion $f$ satisfies \emph{independence of unapproved candidates} if for every instance
$\mathcal I=(N,A,C,k)$, every candidate $c\in C\setminus W$ such that $c\notin A_i$ for all $i\in N$, and the reduced instance $\mathcal I'= \mathcal I_{\mid C\setminus \{c\}}$, it holds that
\[
W\notin f(\mathcal I)\quad\text{implies}\quad W\notin f(\mathcal I')
\]
for every committee $W\subseteq C$.
\end{definition}

Note that all conditions are very mild. Nonetheless, they suffice to characterize the concept of witness-basedness.
\begin{proposition}
    A proportionality notion is witness-based if and only if it is anonymous, neutral, and satisfies independence of unapproved candidates and independence of losers.
    \label{prop:wb_char}
\end{proposition}
\begin{proof}
We prove both directions.

\smallskip
$(\Rightarrow)$ Let $f$ be witness-based and fix a witness function $w$ for~$f$.

\paragraph{Anonymity.}
Let $\pi\colon N\to N$ be a permutation of voters and let $\mathcal I^\pi$ denote the relabeled instance
(with $A^\pi_{\pi(i)}:=A_i$ for all $i\in N$).
Since $\pi$ is arbitrary, it suffices to show that $f(\mathcal I) \subseteq f(\mathcal I^\pi)$, because the other set inclusion follows from applying this claim with $\pi^-1$ to $\mathcal I^\pi$.

To prove the set inclusion, let a committee $W\notin f(\mathcal I^\pi)$ be given. Our goal is to show $W\notin f(\mathcal I)$
By definition of witness functions, $w(\mathcal I^\pi,W)\neq\emptyset$, so pick some
$N'\in w(\mathcal I^\pi,W)$.
The pair $(N',W)$ in $\mathcal I^\pi$ locally embeds into $(\pi^{-1}(N'),W)$ in $\mathcal I$ via
$\Phi_{N'}=\pi^{-1}$ and $\Phi_C=\mathrm{id}_C$.
Hence, by property~(ii) of witness functions, $\pi^{-1}(N')\in w(\mathcal I,W)\neq \emptyset$. Since $w$ is a witness function of $f$, this implies $W\notin f(\mathcal I)$, as desired.

\paragraph{Neutrality.}
Let $\sigma\colon C\to \hat C$ be a bijection and write $\sigma(\mathcal I)$ for the relabeled instance.
For any $W\subseteq C$ and $N'\subseteq N$, the pair $(N',W)$ in $\mathcal I$ locally embeds into $(N',\sigma(W))$ in
$\sigma(\mathcal I)$ via $\Phi_{N'}=\mathrm{id}_{N'}$ and $\Phi_C=\sigma$, and conversely via $\sigma^{-1}$.
The same argument as above yields
$W\in f(\mathcal I)\iff \sigma(W)\in f(\sigma(\mathcal I))$.

\paragraph{Independence of unapproved candidates.}
Fix $\mathcal I=(N,A,C,k)$, a committee $W\subseteq C$, and a candidate $c\in C\setminus W$ such that $c\notin A_i$ for all $i\in N$.
Let $\mathcal I'=(N,A_{\mid C\setminus\{c\}},C\setminus\{c\},k)$.
Let $W\notin f(\mathcal I)$. Our goal is to show $W\notin f(\mathcal I')$.
By definition of witness functions, $w(\mathcal I,W)\neq\emptyset$; pick $N'\in w(\mathcal I,W)$.
The pair $(N',W)$ in $\mathcal I'$ locally embeds into $(N',W)$ in $\mathcal I$ via identity maps on voters
and candidates, so property~(ii) implies $N'\in w(\mathcal I',W)$. Since $w$ is a witness function of $f$, this implies $W\notin f(\mathcal I')$, as desired.

\paragraph{Independence of losers.}
Fix $\mathcal I=(N,A,C,k)$, a committee $W\subseteq C$, and a candidate $c\in C\setminus W$.
Let $\mathcal I'=(N,A_{\mid C\setminus\{c\}},C\setminus\{c\},k)$.
We prove independence of losers by contraposition, i.e., let $W\notin f(\mathcal I')$. Our goal is to show $W\notin f(\mathcal I)$.
By definition of witness functions, $w(\mathcal I',W)\neq\emptyset$; pick $N'\in w(\mathcal I',W)$.
The pair $(N',W)$ in $\mathcal I'$ locally embeds into $(N',W)$ in $\mathcal I$ via identity maps on voters
and candidates, so property~(ii) implies $N'\in w(\mathcal I,W)$. Since $w$ is a witness function of $f$, this implies $W\notin f(\mathcal I)$, as desired.

\smallskip
$(\Leftarrow)$ Assume now that $f$ is anonymous, neutral, and satisfies independence of unapproved candidates and independence of losers.
Define $w$ by
\[
w(\mathcal I,W)\ :=\
\begin{cases}
\emptyset & \text{if } W\in f(\mathcal I),\\
\{N\} & \text{if } W\notin f(\mathcal I).
\end{cases}
\]
We claim that $w$ is a witness function of $f$. Property~(i) holds by construction.

For property~(ii), fix an instance $\mathcal I=(N,A,C,k)$, $N'\subseteq N$, and $W\subseteq C$ with $N'\in w(\mathcal I, W)$.
By definition of $w$, we know $w(\mathcal I,W)=\{N\}$ and hence $N'=N$.
Let $\hat{\mathcal I}=(\hat N,\hat A,\hat C,k)$ and $\hat W\subseteq \hat C$ be such that $\lvert \hat N \rvert = \lvert N\rvert$, $(N,W)$ locally embeds into $(\hat N,\hat W)$ via $(\Phi_N,\Phi_C)$. 
Our goal is to show $\hat N \in w(\hat {\mathcal I}, \hat W)$.

Note that from the definition of $w$, we can deduce $W\notin f(\mathcal I)$.
Let
\[
D := W \,\cup\, \bigcup_{i\in N} A_i
\qquad\text{and}\qquad
\hat D := \Phi_C(D).
\]
By independence of unapproved candidates, deleting candidates in $C\setminus D$ does not affect whether $W$ is chosen. Since $W\subseteq D$ and $\Phi_C$ maps $W$ onto $\hat W$, we also have $\hat W \subseteq \hat D$. Therefore, if $\hat W$ were chosen on $\hat{\mathcal I}$, deleting candidates in $\hat C\setminus \hat D$ would not change that. Formally, let
\[
\mathcal I_D := (N, A_{\mid D}, D, k)
\quad\text{and}\quad
\hat{\mathcal I}_{\hat D} := (\hat N, \hat A_{\mid \hat D}, \hat D, k).
\]
Then, since $f$ satisfies independence of unapproved candidates, we have $W\notin f(\mathcal I_D)$ and since $f$ satisfies independence of losers, we have the logical implication [$\hat W\in f(\hat{\mathcal I}) \Rightarrow \hat W\in f(\hat{\mathcal I}_{\hat D})$].

By definition of local embedding, $\Phi_N$ is a bijection $N\to \hat N$ and $\Phi_C$ restricts to a bijection
$D\to \hat D$ that preserves approvals on~$D$. Thus $\hat{\mathcal I}_{\hat D}$ is obtained from
$\mathcal I_D$ by renaming voters and candidates according to $\Phi_N$ and $\Phi_C$.
By anonymity and neutrality, we obtain
\[
W\notin f(\mathcal I_D) \ \Longrightarrow\ \hat W\notin f(\hat{\mathcal I}_{\hat D}),
\]
which implies $W\notin f(\hat{\mathcal I})$ as reasoned above.
$w(\hat{\mathcal I},\hat W)=\{\hat N\}$, and in particular $\Phi_N(N)=\hat N'\in w(\hat{\mathcal I},\hat W)$.

Hence $w$ satisfies property~(ii), so it is a witness function for~$f$. This shows that $f$ is witness-based.
\end{proof}

\section{Axiomatic Properties of Other Proportionality Notions}
\label{sec:compliance}
\newcommand{\qmark}{\textcolor{black!55}{?}}
\begin{table*}[t]
  \small
  \centering
  \setlength{\tabcolsep}{3.2pt}
  \renewcommand{\arraystretch}{1}
  \resizebox{0.8\textwidth}{!}{%
  \begin{tabular}{llllllllllll}
    \toprule
      & \textbf{JR} & \textbf{PJR} & \textbf{EJR} & \textbf{PJR+} & \textbf{EJR+} & 
      \textbf{Sub-Core} &
      \textbf{FPJR} &
      \textbf{FJR} &
      \textbf{Core} & \textbf{Price} & \textbf{NPR} \\
    \midrule
    Lower quota (party-lists)
      & \xmark & \cmark & \cmark & \cmark$^{\star}$ & \cmark$^{\star}$ & \cmark & \cmark & \cmark & \cmark & \cmark & \cmark \\
    Lower quota extension
      & \xmark & \cmark & \cmark & \cmark$^{\star}$ & \cmark$^{\star}$ &\cmark &\cmark &\cmark & \cmark & \xmark & \cmark \\
    \addlinespace
    Monotonicity
      & \cmark & \xmark & \xmark & \cmark$^{\star}$ & \cmark$^{\star}$ &\cmark &\xmark &\xmark & \cmark & \cmark & \cmark \\
    \makecell[l]{Robustness to\\fully satisfied voters}
      & \cmark & \cmark & \cmark & \cmark$^{\star}$ & \cmark$^{\star}$ & \cmark &\cmark &\cmark & \cmark & \xmark & \cmark \\
    \makecell[l]{Independence of\\approval swaps}
      & \cmark & \xmark & \xmark & \xmark & \cmark$^{\star}$ & \xmark &\xmark &\xmark & \xmark & \xmark & \cmark \\
    Independence of losers
      & \cmark & \cmark & \cmark & \cmark$^{\star}$ & \cmark$^{\star}$ &\cmark &\cmark &\cmark & \cmark & \cmark & \cmark \\
    \addlinespace
    \hdashline
    \addlinespace
    Witness-based
      & \cmark & \cmark & \cmark & \cmark$^{\star}$ & \cmark$^{\star}$ & \cmark &\cmark &\cmark & \cmark & \cmark & \cmark \\
    Cohesiveness-based
      & \cmark & \cmark & \cmark & \cmark$^{\star}$ & \cmark$^{\star}$ &\xmark &\xmark &\xmark & \xmark & \xmark & \cmark \\
    Individual discontentment
      & \cmark & \cmark & \cmark & \cmark & \cmark$^{\star}$ &\cmark &\cmark &\cmark & \cmark & \xmark & \xmark \\
    Merge-proofness
      & \cmark & \cmark & \xmark & \cmark$^{\star}$ & \xmark & \cmark &\cmark &\xmark & \xmark & \xmark & \xmark \\
    \bottomrule
  \end{tabular}}
  \caption{Axiomatic satisfactions of the proportionality notions.
  \cmark/\xmark indicate satisfaction/violation.
  Entries marked $^{\star}$ are used as part of our characterizations.}
  \label{tab:big-fingerprint}
\end{table*}

\subsection{JR}
\begin{proposition}
JR is a \happygreen{witness-based proportionality notion satisfying independence of losers, monotonicity, \ifsv, and independence of approval swaps}. Together with its natural witness function, JR is \happygreen{cohesiveness-based and satisfies individual discontentment and merge-proofness}. JR does not satisfy \textcolor{red!40!gray}{lower quota for party-lists}.    
\end{proposition}
\begin{proof}
    JR satisfying \textcolor{green!30!gray}{\ifsv} and \textcolor{green!30!gray}{independence of losers} follows from \Cref{general_lem}.
    
    \noindent\textbf{\textcolor{green!30!gray}{Monotonicity}:} Consider an instance $\mathcal{I} = (N, A, C, k)$, committee $W \in JR(\mathcal{I})$ and approval profile $A' \supseteq A$ with $A'_i \setminus A_i \subseteq W$ for all $i \in N$. Let $c \notin W$ be any candidate unselected by $W$ and consider any group $N' \subseteq \{i \in N \colon c \in A'_i\}$ of voters. As $c \notin W$ we also know that $N' \subseteq \{i \in N \colon c \in A_i\}$. i.e., they also approve of $c$ in the original instance. As $W$ satisfies JR in $\mathcal{I}$, we know that either $\lvert N'\rvert < \frac{n}{k}$ or there must exist a $c' \in W$ and $i \in N'$ with $c' \in A_i$. However, then as we only added approvals to candidates in $W$ it also holds that $c' \in A'_i$. Hence, $W$ also satisfies JR in the instance $(N, A', C, k)$.

    \noindent\textbf{\textcolor{green!30!gray}{Merge-proofness and Individual Discontentment}:}
Let $w$ be the natural witness function of JR, i.e., $w(\mathcal{I},W)$ contains every set $N'\subseteq N$ such that $\lvert N'\rvert \ge \frac{n}{k}$, $\bigcap_{i\in N'} A_i \neq \emptyset$, and $A_i \cap W = \emptyset$ for all $i\in N'$. For any $N' \in w(\mathcal{I},W)$, any $c\in\bigcap_{i\in N'} A_i$ cannot be in $W$ (otherwise each voter in $N'$ would be represented), hence $\bigcap_{i\in N'} A_i \setminus W \neq \emptyset$ and JR together with $w$ is cohesiveness-based.

To see that JR satisfies individual discontentment, let $N'\in w(\mathcal{I},W)$ and  $\mathcal{I}^{\mathrm{ID}}=(N,A^{\mathrm{ID}},C,k)$ be the instance with $A^{\mathrm{ID}}_i=A_i$ for all $i\notin N'$ and $A^{\mathrm{ID}}_i=A_j$ for all $i\in N'$ and some $j\in N'$. As $A_j\cap W=\emptyset$, no voter in $N'$ is represented by $W$ in $\mathcal{I}^{\mathrm{ID}}$. Moreover, since $N'$ is $1$-cohesive in $\mathcal{I}$, we have $\bigcap_{i\in N'} A_i \subseteq A_j = \bigcap_{i\in N'} A^{\mathrm{ID}}_i$, so $N'$ remains $1$-cohesive and $1$-large. Thus, $N' \in w(\mathcal{I}^{\mathrm{ID}},W)$.

For merge-proofness, let $\mathcal{I}^{\mathrm{MP}}=(N,A^{\mathrm{MP}},C,k)$ be any instance with $A^{\mathrm{MP}}_i=A_i$ for all $i\notin N'$ and $A_i \subseteq A^{\mathrm{MP}}_i \subseteq \bigcup_{j\in N'} A_j$ for all $i\in N'$. Since $A_i \subseteq A^{\mathrm{MP}}_i$, the set $N'$ stays $1$-cohesive and $1$-large. Further, $\bigcup_{j\in N'} A_j$ is disjoint from $W$ (as each $A_j$ is), so $A^{\mathrm{MP}}_i\cap W=\emptyset$ for all $i\in N'$. Hence $N' \in w(\mathcal{I}^{\mathrm{MP}},W)$.

   \begin{wrapstuff}[r,type=figure,width=5cm]
    \centering
    \begin{tikzpicture}
    [yscale=0.6,xscale=0.8,
    voter/.style={anchor=south}]
    
        \foreach \i in {1,...,4}
    		\node[voter] at (\i-0.5, -1) {$\i$};
        
        \draw[fill=teal!\strstr, ultra thick] (0, 0) rectangle (2, 1);
        \draw[fill=teal!\strstr,ultra thick] (0, 1) rectangle (2, 2);
        \draw[fill=teal!\strstr, ultra thick] (0, 2) rectangle (2, 3);
        \draw[fill=magenta!\strstr, ultra thick] (2, 0) rectangle (4, 1);
        \draw[fill=magenta!\clrstr] (2, 1) rectangle (4, 2);
        \draw[fill=magenta!\clrstr] (2, 2) rectangle (4, 3);
        \node at ( 1, 0.5) {$c_{1}$};
        \node at ( 1, 1.5) {$c_{2}$};
        \node at ( 1, 2.5) {$c_{3}$};
        \node at ( 3, 0.5) {$c_{4}$};
        \node at ( 3, 1.5) {$c_{5}$};
        \node at ( 3, 2.5) {$c_{6}$};
    \end{tikzpicture}
    \caption{Example for JR and lower quota for party-lists.}
    \label{fig:party}
\end{wrapstuff}
\noindent\textbf{\textcolor{red!40!gray}{Lower Quota for party-lists}:} To see that JR does not satisfy lower quota for party-lists consider the instance depticted in \Cref{fig:party} for $k = 4$. One possible committee satisfying JR is $\{c_1, c_2, c_3, c_4\}$. This however, is a lower quota for party-lists violation, as the two voters $3$ and $4$ together form a party deserving $4\cdot \frac{2}{4} = 2$ seats.

    \noindent\textbf{\textcolor{green!30!gray}{Independence of Approval Swaps}:}
Consider an instance $\mathcal{I}=(N,A,C,k)$, a committee $W\in JR(\mathcal{I})$ and an approval profile $\bar A$ such that $\bar A_i\setminus W = A_i\setminus W$ and $\lvert \bar A_i \cap W\rvert = \lvert A_i \cap W\rvert$ for all $i \in N$. We show that $W\in JR(N,\bar A,C,k)$.
Let $N'\subseteq N$ be any $1$-large and $1$-cohesive group in $(N,\bar A,C,k)$. If there exists $i\in N'$ with $\bar A_i\cap W\neq\emptyset$, then $W$ represents $N'$ and we are done. Otherwise, $\bar A_i\cap W=\emptyset$ for all $i\in N'$, and by $\lvert \bar A_i \cap W\rvert = \lvert A_i \cap W\rvert$ we also have $A_i\cap W=\emptyset$ for all $i\in N'$. Together with $\bar A_i\setminus W = A_i\setminus W$, this implies $\bar A_i=A_i$ for all $i\in N'$. Hence, $N'$ is also $1$-cohesive in $\mathcal{I}$ and, since no voter in $N'$ is represented by $W$, this contradicts $W\in JR(\mathcal{I})$. Therefore $W\in JR(N,\bar A,C,k)$.

\wrapstuffclear
\end{proof}
\subsection{PJR}
\begin{proposition}
PJR is a \happygreen{witness-based proportionality notion satisfying independence of losers, \ifsv, and lower quota extension}. Together with its natural witness PJR is \happygreen{cohesiveness-based and satisfies individual discontentment and merge-proofness}. PJR does not satisfy \sadred{monotonicity or independence of approval swaps}. \label{pjr:prop}   
\end{proposition}
\begin{proof}
PJR satisfying \textcolor{green!30!gray}{\ifsv} and \textcolor{green!30!gray}{independence of losers} follows from \Cref{general_lem}. Moreover, PJR is a \happygreen{lower quota extension} by \Cref{prop:lq}.

\noindent\textbf{\textcolor{green!30!gray}{Cohesiveness-Based, Merge-proofness, and Individual Discontentment}:}
Let $w^{\pjr}$ be the natural witness of PJR. Consider an instance $\mathcal{I}=(N,A,C,k)$, a committee $W$, and a witness $N' \in w^{\pjr}(\mathcal{I},W)$. Then there exists some $\ell \in [k]$ such that $N'$ is $\ell$-large and $\ell$-cohesive, but
\[
\Bigl\lvert W \cap \bigcup_{i\in N'} A_i \Bigr\rvert < \ell.
\]
Since $\lvert \bigcap_{i\in N'} A_i\rvert \ge \ell$ but fewer than $\ell$ candidates approved by (at least one voter in) $N'$ are contained in $W$, we must have $\bigcap_{i\in N'} A_i \setminus W \neq \emptyset$, i.e., $w^{\pjr}$ is \happygreen{cohesiveness-based}.

For \happygreen{individual discontentment}, consider $\mathcal{I}'=(N,A',C,k)$ where $A'_i=A_i$ for all $i\notin N'$ and $A'_i=A_j$ for all $i\in N'$ and some fixed $j\in N'$. Then $\bigcap_{i\in N'} A'_i = A_j \supseteq \bigcap_{i\in N'} A_i$, so $N'$ remains $\ell$-cohesive and $\ell$-large. Also, $\bigcup_{i\in N'} A'_i = A_j \subseteq \bigcup_{i\in N'} A_i$, hence
\[
\Bigl\lvert W \cap \bigcup_{i\in N'} A'_i \Bigr\rvert \le
\Bigl\lvert W \cap \bigcup_{i\in N'} A_i \Bigr\rvert < \ell,
\]
so $N'\in w^{\pjr}(\mathcal{I}',W)$.

For \happygreen{merge-proofness}, consider $\mathcal{I}'=(N,A',C,k)$ where $A'_i=A_i$ for $i\notin N'$ and $A_i\subseteq A'_i \subseteq \bigcup_{j\in N'} A_j$ for all $i\in N'$. Then $\bigcap_{i\in N'} A'_i \supseteq \bigcap_{i\in N'} A_i$, so $N'$ stays $\ell$-cohesive. Moreover,
\[
\bigcup_{i\in N'} A'_i \subseteq \bigcup_{i\in N'} A_i
\quad\Rightarrow\quad
\Bigl\lvert W \cap \bigcup_{i\in N'} A'_i \Bigr\rvert \le
\Bigl\lvert W \cap \bigcup_{i\in N'} A_i \Bigr\rvert < \ell,
\]
and hence $N' \in w^{\pjr}(\mathcal{I}',W)$.

\noindent\textbf{\sadred{Monotonicity}:}
PJR does not satisfy monotonicity; this is witnessed by the instance and committee in \Cref{fig:pjr_mon} (see also \Cref{exp:pjr_mon}).

\noindent\textbf{\sadred{Independence of Approval Swaps}:}
\begin{wrapstuff}[r,type=figure,width=5cm]
\centering
\begin{tikzpicture}
[yscale=0.6,xscale=0.8,
voter/.style={anchor=south}]

    \foreach \i in {1,...,4}
        \node[voter] at (\i-0.5, -1) {$\i$};

    \draw[fill=teal!\strstr, ultra thick] (0, 0) rectangle (2, 1);
    \draw[fill=teal!\clrstr] (0, 1) rectangle (2, 2);
    \draw[fill=teal!\strstr, ultra thick] (0, 2) rectangle (1, 3);
    \draw[fill=magenta!\strstr, ultra thick] (2, 0) rectangle (4, 1);
    \draw[fill=magenta!\strstr, ultra thick] (2, 1) rectangle (4, 2);
    \draw[fill=magenta!\strstr, ultra thick] (1, 2) rectangle (4, 3);
    \draw[fill=magenta!\strstr, ultra thick] (2, 3) rectangle (4, 4);

    \node at ( 1, 0.5) {$c_{1}$};
    \node at ( 1, 1.5) {$c_{2}$};
    \node at ( 0.5, 2.5) {$c_{3}$};
    \node at ( 3, 0.5) {$c_{4}$};
    \node at ( 3, 1.5) {$c_{5}$};
    \node at ( 2.5, 2.5) {$c_{6}$};
    \node at ( 3, 3.5) {$c_{7}$};
\end{tikzpicture}
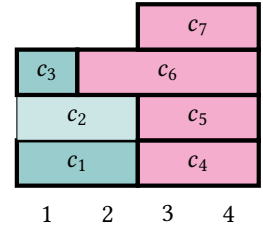
\caption{Example for PJR and independence of approval swaps.}
\label{fig:pjr_ias}
\end{wrapstuff}
To see that PJR does not satisfy independence of approval swaps, consider the instance depicted in \Cref{fig:pjr_ias} with $k=6$ and the (bold) committee $W=\{c_1,c_3,c_4,c_5,c_6,c_7\}$. One checks that $W$ satisfies PJR: in particular, the group $\{1,2\}$ is $2$-large and $2$-cohesive and obtains two approved candidates in $W$ (namely $c_1$ and $c_3$ or $c_6$). Now perform an approval swap for voter $2$ inside $W$, replacing $c_6$ by $c_3$. This does not change $A_2\setminus W$ and keeps $\lvert A_2\cap W\rvert$ constant. However, after the swap the group $\{1,2\}$ becomes $3$-cohesive (while still $3$-large for $k=6$ and $n=4$), but $W$ still contains only two candidates approved by at least one of them, contradicting PJR. 
\wrapstuffclear
\end{proof}

\subsection{EJR}
\begin{proposition}
EJR is a \happygreen{witness-based proportionality notion satisfying independence of losers, \ifsv, and lower quota extension}. Together with its natural witness EJR is \happygreen{cohesiveness-based and satisfies individual discontentment}. EJR does not satisfy \sadred{merge-proofness, monotonicity, or independence of approval swaps}.    
\end{proposition}
\begin{proof}
EJR satisfying \textcolor{green!30!gray}{\ifsv} and \textcolor{green!30!gray}{independence of losers} follows from \Cref{general_lem}. \happygreen{Lower quota for party-lists} follows from \Cref{prop:lq}.

\noindent\textbf{\textcolor{green!30!gray}{Cohesiveness-Based and Individual Discontentment}:}
Let $w^{\ejr}$ denote the natural witness function of EJR. Consider an instance $\mathcal{I}=(N,A,C,k)$, a committee $W$, and a witness $N' \in w^{\ejr}(\mathcal{I},W)$. Then there exists some $\ell \in [k]$ such that $N'$ is $\ell$-large and $\ell$-cohesive and
\[
\lvert A_i \cap W\rvert < \ell \qquad \text{for all } i \in N'.
\]
Since $\lvert \bigcap_{i\in N'} A_i\rvert \ge \ell$, if it held that $\bigcap_{i\in N'} A_i \subseteq W$, then every voter $i\in N'$ would approve at least $\ell$ members of $W$, contradicting $\lvert A_i \cap W\rvert < \ell$. Hence, $\bigcap_{i\in N'} A_i \setminus W \neq \emptyset$ and $w^{\ejr}$ is \happygreen{cohesiveness-based}.

To see \happygreen{individual discontentment}, pick any $c \in \bigcap_{i\in N'} A_i \setminus W$ and consider the modified instance $\mathcal{I}'=(N,A',C,k)$ with $A'_i=A_i$ for all $i\notin N'$ and $A'_i=A_j$ for all $i\in N'$ and some fixed $j\in N'$. Then $c\in A'_i$ for all $i\in N'$ and, since $\lvert A_j\cap W\rvert < \ell$, we also have $\lvert A'_i \cap W\rvert < \ell$ for all $i\in N'$. Thus $N'$ is still a witness for an EJR violation in $\mathcal{I}'$, i.e., $N' \in w^{\ejr}(\mathcal{I}',W)$.

\noindent\textbf{\sadred{Merge-proofness}:}
\begin{wrapstuff}[r,type=figure,width=5cm]
\centering
\begin{tikzpicture}
[yscale=0.6,xscale=0.8,
voter/.style={anchor=south}]

    \foreach \i in {1,...,6}
        \node[voter] at (\i-0.5, -1) {$\i$};

    \draw[fill=teal!\strstr, ultra thick] (0, 0) rectangle (2, 1);
    \draw[fill=teal!\strstr, ultra thick] (2, 1) rectangle (4, 2);
    \draw[fill=magenta!\strstr, ultra thick] (4, 2) rectangle (6, 3);

    \draw[fill=teal!\clrstr] (0, 3) rectangle (4, 4);
    \draw[fill=teal!\clrstr] (0, 4) rectangle (4, 5);

    \node at ( 1, 0.5) {$c_{1}$};
    \node at ( 3, 1.5) {$c_{2}$};
    \node at ( 5, 2.5) {$c_{3}$};
    \node at ( 2, 3.5) {$c_{4}$};
    \node at ( 2, 4.5) {$c_{5}$};
\end{tikzpicture}
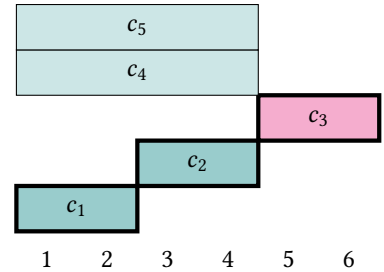
\caption{Example for EJR and merge-proofness.}
\label{fig:ejr_mp}
\end{wrapstuff}
Consider the instance in \Cref{fig:ejr_mp} with $k=3$ and the (bold) committee $W=\{c_1,c_2,c_3\}$. Let $N'=\{1,2,3,4\}$. Then $N'$ is $2$-large (since $\lvert N'\rvert=4 = 2\cdot \frac{n}{k}$) and $2$-cohesive because $\bigcap_{i\in N'} A_i = \{c_4,c_5\}$. Moreover, every voter in $N'$ approves exactly one member of $W$ (either $c_1$ or $c_2$), so $\lvert A_i\cap W\rvert < 2$ for all $i\in N'$. Hence $N'$ is a witness of an EJR violation for $W$.

Now consider the merged profile $A'$ with $A'_i = A_i$ for $i\notin N'$ and $A'_i=\{c_1,c_2,c_4,c_5\}$ for all $i\in N'$. Then $A_i \subseteq A'_i \subseteq \bigcup_{j\in N'} A_j$ for all $i\in N'$, but every voter in $N'$ now approves two members of $W$ (namely $c_1$ and $c_2$), so $N'$ is no longer a witness for $W$. Thus, EJR with its natural witness does not satisfy merge-proofness.
\wrapstuffclear

\noindent\textbf{\sadred{Monotonicity}:}
The instance and committee depicted in \Cref{fig:pjr_mon} (used for PJR) also witness a monotonicity violation for EJR.

\noindent\textbf{\sadred{Independence of Approval Swaps}:}
Similarly, the instance and committee depicted in \Cref{fig:pjr_ias} witness an independence of approval swaps violation for EJR.
\end{proof}

\subsection{Core}
\begin{proposition}
Core stability is a \happygreen{witness-based proportionality notion satisfying independence of losers, lower quota for party-lists, monotonicity, and \ifsv}. Together with its natural witness the core satisfies \happygreen{individual discontentment}. There does not exist any witness function for the core that is \sadred{cohesiveness-based} or satisfies \sadred{merge-proofness}. Core stability does not satisfy \sadred{independence of approval swaps}.    
\end{proposition}
\begin{proof}
The core satisfying \textcolor{green!30!gray}{\ifsv} and \textcolor{green!30!gray}{independence of losers} follows analogously to \Cref{general_lem}. \happygreen{Lower quota for party-lists} follows from \Cref{prop:lq}.

\noindent\textbf{\textcolor{green!30!gray}{Witness-Based and Individual Discontentment}:}
Let $w^{\Core}$ be the natural witness function of the core.
Then $W\in\Core(\mathcal{I})$ if and only if $w^{\Core}(\mathcal{I},W)=\emptyset$, and the standard local embedding argument as in the proof for EJR+ shows that $w^{\Core}$ is indeed a witness function.

Further, $w^{\Core}$ satisfies \happygreen{individual discontentment}: if $N'\in w^{\Core}(\mathcal{I},W)$ is witnessed by some $\ell$ and $C'$, then for any instance $\mathcal{I}'=(N,A',C,k)$ with $A'_i=A_i$ for $i\notin N'$ and $A'_i=A_j$ for all $i\in N'$ and some $j\in N'$, the same $\ell$ and $C'$ witness $N'\in w^{\Core}(\mathcal{I}',W)$ since
\[
|A'_i\cap C'|=|A_j\cap C'|>|A_j\cap W|=|A'_i\cap W|
\qquad \text{for all } i\in N'.
\]

\begin{wrapstuff}[r,type=figure,width=7.2cm]
\centering
\begin{tikzpicture}
[yscale=0.6,xscale=0.75,
voter/.style={anchor=south}]

\node[name] at (3, 6.5) {Instance $\mathcal{I}$:};

    \foreach \i in {1,...,6}
        \node[voter] at (\i-0.5, -1) {$\i$};

    \draw[fill=teal!\strstr, ultra thick] (0, 0) rectangle (2, 1);
    \draw[fill=teal!\strstr, ultra thick] (2, 1) rectangle (4, 2);
    \draw[fill=teal!\strstr, ultra thick] (4, 2) rectangle (6, 3);

    \draw[fill=magenta!\clrstr] (0, 3) rectangle (2, 4);
    \draw[fill=magenta!\clrstr] (4, 3) rectangle (6, 4);
    \draw[fill=magenta!\clrstr] (0, 4) rectangle (4, 5);
    \draw[fill=magenta!\clrstr] (2, 5) rectangle (6, 6);

    \node at (1, 0.5) {$c_{1}$};
    \node at (3, 1.5) {$c_{2}$};
    \node at (5, 2.5) {$c_{3}$};

    \node at (1, 3.5) {$c_{4}$};
    \node at (5, 3.5) {$c_{4}$};
    \node at (2, 4.5) {$c_{5}$};
    \node at (4, 5.5) {$c_{6}$};
\end{tikzpicture}

\begin{tikzpicture}
[yscale=0.6,xscale=0.75,
voter/.style={anchor=south}]
\node[name] at (3, 6.5) {Instance $\mathcal{I}'$:};
    \foreach \i in {1,...,6}
        \node[voter] at (\i-0.5, -1) {$\i$};

    \draw[fill=teal!\strstr, ultra thick] (0, 0) rectangle (2, 1);
    \draw[fill=teal!\strstr, ultra thick] (2, 1) rectangle (4, 2);
    \draw[fill=teal!\strstr, ultra thick] (4, 2) rectangle (6, 3);

    \draw[fill=magenta!\clrstr] (0, 3) rectangle (2, 4);
    \draw[fill=magenta!\clrstr] (4, 3) rectangle (6, 4);
    \draw[fill=magenta!\clrstr] (0, 4) rectangle (4, 5);
    \draw[fill=magenta!\clrstr] (2, 5) rectangle (5, 6);

    \node at (1, 0.5) {$c_{1}$};
    \node at (3, 1.5) {$c_{2}$};
    \node at (5, 2.5) {$c_{3}$};

    \node at (1, 3.5) {$c_{4}$};
    \node at (5, 3.5) {$c_{4}$};
    \node at (2, 4.5) {$c_{5}$};
    \node at (3.5, 5.5) {$c_{6}$};
\end{tikzpicture}
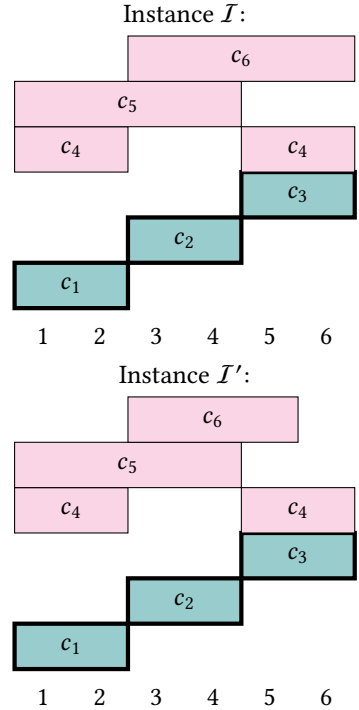
\caption{Instances used to separate the core from cohesiveness-basedness and merge-proofness.}
\label{fig:core_cycle}
\end{wrapstuff}
\noindent\textbf{\sadred{Cohesiveness-Basedness}:}
Consider the instance $\mathcal{I}$ depicted in \Cref{fig:core_cycle} with $k=3$ and the committee $W=\{c_1,c_2,c_3\}$. Then $W\notin\Core(\mathcal{I})$: for $\ell=3$, the coalition $N'=N$ satisfies $|N'|=6=3\cdot \frac{n}{k}$ and can deviate with $C'=\{c_4,c_5,c_6\}$ since every voter approves exactly two candidates in $C'$ but only one candidate in $W$.

Now assume for contradiction that there exists some witness function $w$ for the core such that the pair $(\Core,w)$ is cohesiveness-based. As $W\notin\Core(\mathcal{I})$, we have $w(\mathcal{I},W)\neq \emptyset$ and can pick some $N^\star\in w(\mathcal{I},W)$. Since $(\Core,w)$ is cohesiveness-based, there exists a candidate
\[
c \in \bigcap_{i\in N^\star} A_i \setminus W.
\]
In \Cref{fig:core_cycle}, the only candidates outside of $W$ are $c_4,c_5,c_6$, hence $c\in\{c_4,c_5,c_6\}$ and $N^\star$ is contained in the support of this candidate. Without loss of generality we assume the candidate is $c_5$ and that $N^\star \subseteq \{1,2,3,4\}$.
We construct the modified instance $\mathcal{I}'$ depicted second in \Cref{fig:core_cycle}.
Now, one can verify that indeed $W\in\Core(\mathcal{I}')$. 
However, $(N^\star, W)$ in $\mathcal{I}$ can be locally embedded into $(N^\star, W)$ in $\mathcal{I}'$ and hence, must also have been a witness for this instance, a contradiction.
Therefore, there does not exist any witness function for the core that is cohesiveness-based.
\wrapstuffclear

\noindent\textbf{\sadred{Merge-Proofness}:}
We again use the instance $\mathcal{I}$ and committee $W$ from \Cref{fig:core_cycle}.
First observe that in this instance, the \emph{only} deviation against $W$ is the grand coalition with $\ell=3$:
for $\ell=1$ every voter already approves one member of $W$, so no coalition can deviate with a single candidate;
for $\ell=2$, a deviating coalition would need size at least $2\cdot\frac{n}{k}=4$ and would require two candidates $C'$ such that every voter in the coalition approves both (as they need $|A_i\cap C'|>1$), but no four voters in \Cref{fig:core_cycle} have two commonly approved candidates.

Assume for contradiction that there exists a witness function $w$ for the core satisfying merge-proofness.
Since $W\notin\Core(\mathcal{I})$, pick $N^\star\in w(\mathcal{I},W)$.
We claim that $N^\star$ must contain voters approving at least two different members of $W$:
otherwise $N^\star\subseteq \{1,2\}$ or $N^\star\subseteq \{3,4\}$ or $N^\star\subseteq \{5,6\}$, and by the same construction as above we obtain an instance $\tilde{\mathcal I}$ in which $W$ is core-stable, while $(N^\star,W)$ locally-embeds into this instance and hence $N^\star$ would remain a witness, contradicting property~(i).

Thus, $\bigcup_{j\in N^\star} A_j$ contains at least two members of $W$.
Now define a merged profile $A'$ by setting $A'_i=A_i$ for $i\notin N^\star$ and, for $i\in N^\star$,
\[
A'_i \coloneqq A_i \cup \Bigl(\bigcup_{j\in N^\star} A_j \cap W\Bigr).
\]
Then $A_i \subseteq A'_i \subseteq \bigcup_{j\in N^\star} A_j$ for all $i\in N^\star$, so this is an admissible merge according to merge-proofness.
Moreover, every voter $i\in N^\star$ now approves at least two members of $W$, i.e., $|A'_i\cap W|\ge 2$.
Since we did not add any approvals to candidates outside $W$, every such voter still approves exactly two of $\{c_4,c_5,c_6\}$, hence
\[
|A'_i\cap \{c_4,c_5,c_6\}| = 2 \qquad \text{for all } i\in N^\star,
\]
and therefore the unique deviation of $\mathcal{I}$ (the grand coalition with $C'=\{c_4,c_5,c_6\}$) is destroyed in the new instance $(N,A',C,k)$ (as the voters in $N^\star$ no longer strictly improve).
Since no other deviation existed before and increasing approvals within $W$ cannot create new deviating coalitions, we conclude that $W\in\Core(N,A',C,k)$.

However, by merge-proofness we must have $N^\star\in w((N,A',C,k),W)$, contradicting $W\in\Core(N,A',C,k)$ and property~(i).
Hence, there does not exist any witness function for the core satisfying merge-proofness.

\noindent\textbf{\sadred{Independence of Approval Swaps}:}
The instance and committee depicted in \Cref{fig:pjr_ias} also witness an independence of approval swaps violation for the core.

\noindent\textbf{\textcolor{green!30!gray}{Monotonicity}:}
Consider an instance $\mathcal{I} = (N, A, C, k)$, committee $W \in \Core(\mathcal{I})$ and approval profile $A' \supseteq A$ with $A'_i \setminus A_i \subseteq W$ for all $i \in N$. Consider any group $N' \subseteq N$ of voters and $C' \subseteq C$ of candidates with $\lvert C'\rvert \le \ell$ and $\lvert N'\rvert \ge \ell\frac{n}{k}$. As $W$ satisfies the core in $\mathcal{I}$ we know that there exists some $i \in N'$ with
\[
\lvert A_i \cap W \rvert \ge \lvert A_i \cap C'\rvert.
\]
Thus, since $A'_i \setminus A_i \subseteq W$ we also get that 
\[
\lvert A'_i \cap W\rvert = \lvert A_i \cap W\rvert + \lvert A'_i \setminus A_i\rvert \ge \lvert A_i \cap C'\rvert + \lvert A'_i \setminus A_i\rvert \ge \lvert A'_i \cap C'\rvert.
\]
Hence, $N'$ and $C'$ are not a deviating coalition and thus $W$ also satisfies the core in the instance $(N, A', C, k)$. 
\end{proof}

Next, for completeness, we study the compliance of various other axioms studied throughout the approval-based multiwinner voting literature.
\subsection{FJR}
We begin with the notion of full justified representation (FJR) introduced by \citet*{PPS21a}, as an axiom that is still always satisfiable, yet only slightly weaker than the core. The proportionality notion of FJR is based on a generalization of cohesive groups. Given an instance $\mathcal{I} = (N, A, C, k)$ for a subset $T \subseteq C$ of candidates and $\beta \in [k]$ we say that a group $N' \subseteq N$ of voters is $(\beta, T)$-cohesive if $\lvert A_i \cap T\rvert \ge \beta$ for all $i\in N'$ and $\lvert N'\rvert \ge \frac{\lvert T\rvert n}{k}$.
\begin{definition}[\citep{PPS21a}]
    A committee $W$ satisfies \emph{full justified representation (FJR)} if for every $(\beta, T)$-cohesive group $N' \subseteq N$ of voters there exists some $i \in N'$ with $\lvert A_i \cap W\rvert \ge \beta$.
\end{definition}
\begin{proposition}
FJR is a \happygreen{witness-based proportionality notion} satisfying \happygreen{independence of losers, lower quota for party-lists}, and \happygreen{\ifsv}. Together with its natural witness FJR satisfies \happygreen{individual discontentment}. There does not exist any witness function for the core that is \sadred{cohesiveness-based} or satisfies \sadred{merge-proofness}. FJR does not satisfy \sadred{independence of approval swaps} or \sadred{monotonicity}.    
\end{proposition}
\begin{proof}
FJR satisfying \textcolor{green!30!gray}{\ifsv} and \textcolor{green!30!gray}{independence of losers} follows analogously to \Cref{general_lem}. \happygreen{Lower quota for party-lists} follows from \Cref{prop:lq} as for any instance it holds that $\ejr \supseteq \mathrm{FJR} \supseteq \Core$. FJR satisfying individual discontentment follows in the same way as for the core.

\noindent\textbf{\sadred{Monotonicity}:}
     \begin{wrapstuff}[r,type=figure,width=5cm]
    \centering
    \begin{tikzpicture}
    [yscale=0.6,xscale=0.8,
    voter/.style={anchor=south}]
    
        \foreach \i in {1,...,4}
    		\node[voter] at (\i-0.5, -1) {$\i$};
        
        \draw[fill=teal!\strstr, ultra thick] (0, 0) rectangle (2, 1);
        \draw[fill=teal!\clrstr] (0, 1) rectangle (2, 2);
        \draw[fill=teal!\strstr, ultra thick] (0, 2) rectangle (1, 3);
        \draw[fill=magenta!\strstr, ultra thick] (2, 0) rectangle (4, 1);
        \draw[fill=magenta!\strstr, ultra thick] (2, 1) rectangle (4, 2);
        \draw[fill=magenta!\strstr, ultra thick] (2, 2) rectangle (4, 3);
        \draw[fill=magenta!\strstr, ultra thick] (2, 3) rectangle (4, 4);
        
        \node at ( 1, 0.5) {$c_{1}$};
        \node at ( 1, 1.5) {$c_{2}$};
        \node at ( 0.5, 2.5) {$c_{3}$};
        \node at ( 3, 0.5) {$c_{4}$};
        \node at ( 3, 1.5) {$c_{5}$};
        \node at ( 3, 2.5) {$c_{6}$};
        \node at ( 3, 3.5) {$c_{7}$};
    \end{tikzpicture}
    \captionof*{figure}{Figure \ref*{fig:pjr_mon}: Example for FJR and Monotonicity.}
    \label{fig:fjr_mon}
\end{wrapstuff}
To see that FJR does not satisfy monotonicity again consider the instance depicted in \Cref{fig:pjr_mon}. The depicted committee $\{c_1, c_3, c_4, c_5, c_6, c_7\}$ satisfies FJR for $k = 6$ as the group of voters $\{1,2\}$ is only $(2, \{c_1, c_2, c_3\})$-cohesive. However, if the voter $2$ additionally approved candidate $c_3$ the committee no longer satisfies FJR, as now $\{1,2\}$ are $(3, \{c_1, c_2, c_3\})$-cohesive, but only receive two candidates in the outcome.
\wrapstuffclear    

The instance depicted in \Cref{fig:pjr_ias} also witnesses that FJR violates \sadred{independence of approval swaps}. Further, any deviation in \Cref{fig:core_cycle} is not just a core deviation, but also an FJR deviation, and hence, this instance also witnesses that FJR violates both \sadred{cohesiveness-basedness} and \sadred{merge-proofness}.
\end{proof}

\subsection{FPJR}
\begin{definition}[\citep{KKL25a}]
    A committee $W$ satisfies \emph{full proportional justified representation (FPJR)} if for every $(\beta, T)$-cohesive group $N' \subseteq N$ of voters it holds that $$\left\lvert\bigcup_{i \in N'} A_i \cap W\right\rvert \ge \beta.$$
\end{definition}
\begin{proposition}
FPJR is a \happygreen{witness-based proportionality notion satisfying independence of losers, \ifsv, and is a lower quota extension}. Together with its natural witness FPJR \happygreen{satisfies individual discontentment and merge-proofness}. FPJR does not satisfy \textcolor{red!40!gray}{monotonicity and independence of approval swaps} and there is no cohesiveness-based witness for FPJR.    
\end{proposition}
\begin{proof}
FPJR satisfying \textcolor{green!30!gray}{\ifsv} and \textcolor{green!30!gray}{independence of losers} follows analogously to \Cref{general_lem}. \happygreen{Lower quota for party-lists} follows from \Cref{prop:lq} as for any instance it holds that $\ejr \supseteq \mathrm{FPJR} \supseteq \Core$. FJPR satisfying individual discontentment follows in the same way as for the core.

\noindent\textbf{\sadred{Cohesiveness-Basedness}:}
To see that there is no cohesiveness-based witness for FPJR, consider the instance depicted in \Cref{fig:laminar}. Here without loss of generality the only possible witness for FPJR is $\{1\}$. It is however, easy to see that this is not a valid witness. For instance if voter $2$ would only approve $c_1$ the committee would satisfy FPJR.

\noindent\textbf{\sadred{Monotonicity}:}
     \begin{wrapstuff}[r,type=figure,width=5cm]
    \centering
    \begin{tikzpicture}
    [yscale=0.6,xscale=0.8,
    voter/.style={anchor=south}]
    
        \foreach \i in {1,...,4}
    		\node[voter] at (\i-0.5, -1) {$\i$};
        
        \draw[fill=teal!\strstr, ultra thick] (0, 0) rectangle (2, 1);
        \draw[fill=teal!\clrstr] (0, 1) rectangle (2, 2);
        \draw[fill=teal!\strstr, ultra thick] (0, 2) rectangle (1, 3);
        \draw[fill=magenta!\strstr, ultra thick] (2, 0) rectangle (4, 1);
        \draw[fill=magenta!\strstr, ultra thick] (2, 1) rectangle (4, 2);
        \draw[fill=magenta!\strstr, ultra thick] (2, 2) rectangle (4, 3);
        \draw[fill=magenta!\strstr, ultra thick] (2, 3) rectangle (4, 4);
        
        \node at ( 1, 0.5) {$c_{1}$};
        \node at ( 1, 1.5) {$c_{2}$};
        \node at ( 0.5, 2.5) {$c_{3}$};
        \node at ( 3, 0.5) {$c_{4}$};
        \node at ( 3, 1.5) {$c_{5}$};
        \node at ( 3, 2.5) {$c_{6}$};
        \node at ( 3, 3.5) {$c_{7}$};
    \end{tikzpicture}
    \captionof*{figure}{Figure \ref*{fig:pjr_mon}: Example for FPJR and Monotonicity.}
    \label{fig:fpjr_mon}
\end{wrapstuff}
To see that FPJR does not satisfy monotonicity again consider the instance depicted in \Cref{fig:pjr_mon}. The depicted committee $\{c_1, c_3, c_4, c_5, c_6, c_7\}$ satisfies FJR for $k = 6$ as the group of voters $\{1,2\}$ is only $(2, \{c_1, c_2, c_3\})$-cohesive. However, if the voter $2$ additionally approved candidate $c_3$ the committee no longer satisfies FPJR, as now $\{1,2\}$ are $(3, \{c_1, c_2, c_3\})$-cohesive, but only receive two candidates in the outcome.
\wrapstuffclear    

\noindent\textbf{\sadred{Independence of Approval Swaps}:}
\begin{wrapstuff}[r,type=figure,width=5cm]
\centering
\begin{tikzpicture}
[yscale=0.6,xscale=0.8,
voter/.style={anchor=south}]

    \foreach \i in {1,...,4}
        \node[voter] at (\i-0.5, -1) {$\i$};

    \draw[fill=teal!\strstr, ultra thick] (0, 0) rectangle (2, 1);
    \draw[fill=teal!\clrstr] (0, 1) rectangle (2, 2);
    \draw[fill=teal!\strstr, ultra thick] (0, 2) rectangle (1, 3);
    \draw[fill=magenta!\strstr, ultra thick] (2, 0) rectangle (4, 1);
    \draw[fill=magenta!\strstr, ultra thick] (2, 1) rectangle (4, 2);
    \draw[fill=magenta!\strstr, ultra thick] (1, 2) rectangle (4, 3);
    \draw[fill=magenta!\strstr, ultra thick] (2, 3) rectangle (4, 4);

    \node at ( 1, 0.5) {$c_{1}$};
    \node at ( 1, 1.5) {$c_{2}$};
    \node at ( 0.5, 2.5) {$c_{3}$};
    \node at ( 3, 0.5) {$c_{4}$};
    \node at ( 3, 1.5) {$c_{5}$};
    \node at ( 2.5, 2.5) {$c_{6}$};
    \node at ( 3, 3.5) {$c_{7}$};
\end{tikzpicture}
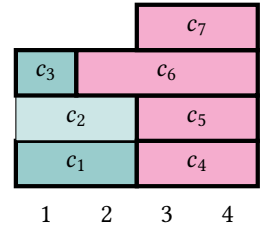
\captionof*{figure}{Figure \ref*{fig:pjr_ias}: Example for FPJR and independence of approval swaps.}
\label{fig:fpjr_ias}
\end{wrapstuff}
For independence of approval swaps, we again consider \Cref{fig:pjr_ias}. Here, the committee $c_1, c_3, c_4, \dots, c_7$ satisfies FPJR for $k = 6$. However, if the voter $c_2$ switches their approval profile to be $\{c_1, c_2, c_3\}$ (and thus still only approving two candidates in the outcome), the group of voters $1$ and $2$ is now $(3, \{c_1, c_2, c_3\})$-cohesive while only receiving two candidates in the outcome. Hence, in this new instance the committee does not satisfy FPJR and thus FPJR does not satisfy independence of approval swaps. 

\noindent\textbf{\happygreen{Merge-Proofness}:}
To see that FPJR together with its natural witness $w$ satisfies merge-proofness, consider any instance $\mathcal I$, committee $W$ and $N' \in w(\mathcal{I}, W)$. We assume that $N'$ is $(\beta, T)$-cohesive and that $\lvert \bigcup A_i \cap W\rvert < \beta$. However, after adding additional candidates from $\bigcup A_i \cap W\rvert$ to the approval sets of voters in $N'$ this group is still $(\beta, T)$-cohesive (as they did not lose any approvals), while $\lvert \bigcup A_i \cap W\rvert < \beta$ continues to hold. Thus, $N'$ remains a witness and FPJR is merge-proof.
\wrapstuffclear
\end{proof}
\subsection{Sub-Core}
\label{app:sub-core}
\begin{definition}[\citep{MSW22a}]
    A committee $W$ is sub-core-stable, if there does not exist any $\ell \in [k]$, group $N' \subseteq N$ of voters of size at least $\lvert N'\rvert \ge \ell \frac{n}{k}$ and set of candidates $C' \subseteq C $ of size at most $\lvert C'\rvert \le \ell$ such that 
    \[
    A_i \cap C' \supset A_i \cap W \text{ for all } i \in N'.
    \]
\end{definition}

\begin{proposition}
    A committee $W$ is in the sub-core if and only if for every subset $C' \subseteq C \setminus W$ of candidates and $N' \subseteq \{i \in N \colon A_i \cap C' \neq \emptyset\}$ of size at least $\ell \frac{n}{k}$ it holds that 
    \[
    \left\lvert\bigcup_{i \in N'} A_i \cap W \right\rvert > \ell - \lvert C'\rvert
    \]
    \label{sub_core_equiv}
\end{proposition}
\begin{proof}
    Firstly, assume that $W$ is in the subcore and consider any subset $C' \subseteq C \setminus W$ of candidates and $N' \subseteq \{i \in N \colon A_i \cap C' \neq \emptyset\}$ of size at least $\ell \frac{n}{k}$. Now consider the set $C'' \coloneqq C' \cup (\bigcup_{i \in N'} A_i \cap W)$. As $A_i \cap C'' \supset A_i \cap W$ for every $i \in N'$ and since $W$ is in the sub-core, we get that $\lvert C' \cup (\bigcup_{i \in N'} A_i \cap W)\rvert = \lvert C''\rvert > \ell$ showing the first implication. 

    Secondly, assume that $W$ is not in the sub-core. Then there exists an $\ell \in [k]$, subset $N' \subseteq N$ of voters of size at least $\ell\frac{n}{k}$ and a set of candidates $C' \subseteq C$ of size at most $\lvert C'\rvert \le \ell$ such that $A_i \cap C' \supset A_i \cap W$ for all $i\in N'$. Consider the set $C'' = C' \setminus \left(\bigcup_{i \in N'} A_i \cap W \right)$. By definition every voter in $N'$ approves at least one candidate of $C''$. Further, from the definition, we get 
    \[
    \lvert C''\rvert = \lvert C'\rvert - \left\lvert\bigcup_{i \in N'} A_i \cap W \right\rvert \le \ell - \left\lvert\bigcup_{i \in N'} A_i \cap W \right\rvert
    \] and hence, this direction of the equivalence is true.

\end{proof}

\begin{proposition}
    Let $W$ be a committee in the sub-core. Then $W$ also satisfies FPJR.
\end{proposition}
\begin{proof}
    Let $N'$ be a $(\beta, T)$-cohesive group and $T' = T \cap W$. If $\lvert T'\rvert < \beta - 1$ we add $\beta - \lvert T'\rvert -  1$ other (arbitrary) candidates from $T$ to $T'$. 
    We let $\lvert T\rvert = \ell$ and hence $\lvert N'\rvert \ge \ell \frac{n}{k}$. Let $C' = T \setminus T'$. By definition, we know that $\lvert C'\rvert \le \ell - \beta + 1$. Further, as every voter in $N'$ approves at least $\beta$ candidates in $T$ they also approve at least one candidate in $C'$. Thus, we know that $N' \subseteq \{i \in N\colon A_i \cap C'\neq \emptyset \}$. As we additionally required that $C' \subseteq T \setminus W$, by \Cref{sub_core_equiv} we get that 
    \[
    \left\lvert\bigcup_{i \in N'} A_i \cap W \right\rvert > \ell - \lvert C'\rvert = \ell - (\ell - \beta + 1) = \beta - 1
    \] and therefore $W$ also satisfies FPJR.
\end{proof}
This as a corollary gives us a simpler proof that priceability implies FPJR.

\begin{proposition}
The sub-core is a \happygreen{witness-based proportionality notion satisfying independence of losers, monotonicity, \ifsv, and is a lower quota extension}. Together with its natural witness the sub-core \happygreen{satisfies individual discontentment and merge-proofness}. The sub-core does not satisfy \sadred{independence of approval swaps} and there is no \sadred{cohesiveness-based} witness for the sub-core.    
\end{proposition}
\begin{proof}
    Using \Cref{sub_core_equiv} the positive results follow almost verbatim to PJR+, simply replacing the single unchosen candidate by the unchosen set of candidates. As the sub-core lies in-between the core and PJR it is also a lower quota extension. 

    To see that the sub-core possesses no cohesiveness-based witness, we can again consider \Cref{fig:laminar} and see that there cannot be any cohesiveness-based witness in this instance.
\end{proof}
\subsection{Priceability}
\begin{definition}[\citep{PeSk20a}]
A committee $W$ is priceable, if there exists a budget $B > 0$ and a function $p\colon N\times C\to[0,1]$ such that 
\begin{itemize}
    \item[C1)] $c \notin A_i \implies p(i,c) = 0$
    \item[C2)] $\sum_{c \in C} p(i,c)\le \frac{B}{n}$ for all $i \in N$
    \item[C3)] $\sum_{i \in N}p(i,c) = 1$ for all $c \in W$
    \item[C4)] $\sum_{i \in N}p(i,c) = 0$ for all $c \notin W$
    \item[C5)]$\sum_{i \in N_c}\left(\frac{B}{n} - \sum_{c \in A_i} p(i,c)\right) \le 1$ for all $c \notin W$. 
\end{itemize}

\end{definition}

\begin{proposition}
Priceability is a \happygreen{witness-based proportionality notion satisfying independence of losers, monotonicity and lower quota for party-lists}. Priceability is not a \sadred{lower quota extension} and does not satisfy \sadred{\ifsv} or \sadred{independence of approval swaps}.  There are no witness functions for priceability satisfying \sadred{individual discontentment}, \sadred{merge-proofness}, or \sadred{cohesiveness-basedness}.    
\end{proposition}
\begin{proof}
    First, it is easy to see that priceability satisfies independence of losers. Deleting candidates outside the committee does not change a ``valid'' price-system. As priceability is also anonymous and neutral it is a witness-based proportionality notion by \Cref{prop:wb_char}.

    To see that priceability is monotone, we also observe that a price-system stays valid after adding additional approvals for the chosen committee. Further, lower quota for party-lists follows from the fact that every committee satisfying priceability also satisfies PJR \citep{PeSk20a}.

    \noindent\textbf{\sadred{Lower Quota Extension}:} To see that priceability is not a lower quota extension, consider an apportionment instance with two parties, the first party receiving 100 votes, and the second party 1 vote. For $k = 2$ giving one seat to each party satisfies lower-quota for party-lists. It, however, does not lead to a valid price-system: for the lone supporter of the second party to have enough budget to purchase a candidate of the second party, it must hold that $B\ge 101$, i.e., $\frac Bn \ge 1$. Therefore, the voters voting for the first party possess at least a budget of $100$. Even after spending a total budget of $1$ to purchase a candidate from the first party, they have a remaining budget of $99>1$ to purchase further candidates from this party, and therefore violate condition C5).

\begin{wrapstuff}[r,type=figure,width=5cm]
\centering
\begin{tikzpicture}
[yscale=0.6,xscale=0.8,
voter/.style={anchor=south}]

    \foreach \i in {1,...,4}
        \node[voter] at (\i-0.5, -1) {$\i$};

    \draw[fill=teal!\strstr, ultra thick] (0, 0) rectangle (2, 1);
    \draw[fill=teal!\strstr, ultra thick] (0, 1) rectangle (1, 2);
    \draw[fill=teal!\strstr, ultra thick] (2, 1) rectangle (3, 2);
    \draw[fill=teal!\strstr, ultra thick] (1, 2) rectangle (3, 3);
    \draw[fill=magenta!\clrstr] (1, 3) rectangle (4, 4);

    \node at ( 1, 0.5) {$c_{1}$};
    \node at ( 0.5, 1.5) {$c_{2}$};
    \node at ( 2.5, 1.5) {$c_{2}$};
    \node at ( 2, 2.5) {$c_{3}$};
    \node at ( 2.5, 3.5) {$c_{4}$};
\end{tikzpicture}
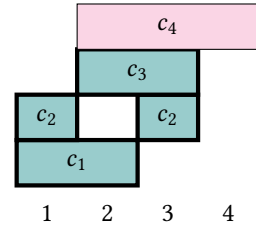
\captionof*{figure}{Figure \ref*{fig:pjr_ias}: Example for priceability and \ifsv.}
\label{fig:price_ifsv}
\end{wrapstuff}

 \noindent\textbf{\sadred{Robustness to Fully Satisfied Voters}:} 
This would follow trivially from an example where a candidate on the committee would not be approved by any voter after changing the preference. We can, however, also give an example where every candidate is approved by at least one voter. For this, consider the instance depicted in \Cref{fig:price_ifsv} with the committee $\{c_1, c_2, c_3\}$. This committee is priceable, for instance with $B = 4$ and voters $1,2,3$ paying for candidates $c_2,c_1,c_3$, respectively. However, if the voters $2$ and $3$ change their approval profiles to be $\{c_3\}$ the committee would no longer be priceable: in that instance, voter $1$ would need to buy both $c_2$ and $c_3$. Hence, the total budget $B$ would need to be at least $8$, and thus voter $4$ together with candidate $c_4$ would violate C5).

Further, by priceability violating \ifsv, we also get from \Cref{lem:witness_iol_ifsv} that there does not exist a cohesiveness-based witness function for it.
\wrapstuffclear
\begin{wrapstuff}[r,type=figure,width=5cm]
\centering
\begin{tikzpicture}
[yscale=0.6,xscale=0.8,
voter/.style={anchor=south}]

    \foreach \i in {1,...,4}
        \node[voter] at (\i-0.5, -1) {$\i$};

    \draw[fill=teal!\strstr, ultra thick] (0, 0) rectangle (2, 1);
    \draw[fill=teal!\strstr, ultra thick] (0, 1) rectangle (1, 2);
    \draw[fill=teal!\strstr, ultra thick] (2, 1) rectangle (3, 2);
    \draw[fill=magenta!\clrstr] (1, 2) rectangle (4, 3);

    \node at ( 1, 0.5) {$c_{1}$};
    \node at ( 0.5, 1.5) {$c_{2}$};
    \node at ( 2.5, 1.5) {$c_{2}$};
    \node at ( 2.5, 2.5) {$c_{3}$};
\end{tikzpicture}
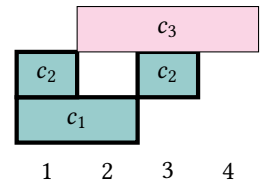
\captionof*{figure}{Figure \ref*{fig:pjr_ias}: Example for priceability and independence of approval swaps.}
\label{fig:price_ias}
\end{wrapstuff}
\noindent\textbf{\sadred{Independence of Approval Swaps}:}
For independence of approval swaps, consider the committee, depicted in \Cref{fig:price_ias}. The committee consisting of $c_1$ and $c_2$ is priceable for a budget of $B = 4$. The second voter could pay for $c_1$ and the third for $c_2$. However, if the third voter switches their approval set to be $c_1$, the committee is no longer priceable. The voters approving $c_3$ have a budget of at least $3$. They can, however, spend at most $1$ and thus $c_3$ is still affordable leading to a violation of condition C5).
\wrapstuffclear
\begin{wrapstuff}[r,type=figure,width=5cm]
\centering
\begin{tikzpicture}
[yscale=0.6,xscale=0.8,
voter/.style={anchor=south}]

    \foreach \i in {1,...,4}
        \node[voter] at (\i-0.5, -1) {$\i$};

    \draw[fill=teal!\strstr, ultra thick] (0, 0) rectangle (4, 1);
    \draw[fill=teal!\strstr, ultra thick] (0, 1) rectangle (1, 2);
    \draw[fill=teal!\strstr, ultra thick] (0, 2) rectangle (1, 3);
    \draw[fill=magenta!\clrstr] (2, 1) rectangle (4, 2);

    \node at ( 2, 0.5) {$c_{1}$};
    \node at ( 0.5, 1.5) {$c_{2}$};
    \node at ( 0.5, 2.5) {$c_{3}$};
    \node at ( 3, 1.5) {$c_{4}$};
\end{tikzpicture}
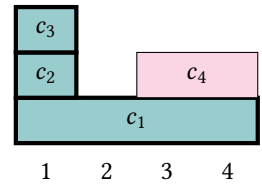
\caption{Example for priceability and witness functions.}
\label{fig:price_wf}
\end{wrapstuff}
For the witness function counter-examples first consider the instance depicted in \Cref{fig:price_wf} and the committee $\{c_1, c_2, c_3\}$ with $k = 3$. We first note that this committee is not priceable. To afford $c_2$ and $c_3$ voter $1$ must have at least a budget of $2$. However, then, $3$ and $4$ must have a budget of at least $3$ left after $c_1$ got bought and hence $c_4$ would violate condition C5).
\noindent\textbf{\sadred{Cohesiveness-Based}:}
First, we note that any subset of $\{3,4\}$ is not a valid witness in this instance. For example, if voter $2$ would additionally approve candidates $c_2$ and $c_3$ the committee $\{c_1, c_2, c_3\}$ would become priceable by giving each voter a budget of $1$ and letting voter $1$ pay for $c_2$, voter $2$ for $c_3$ and voter $3$ for $c_1$. Then voter $4$ is left with a budget of exactly $1$ and hence does not violate condition C5). Thus, as a consequence, there is no cohesiveness-based witness.

\noindent\textbf{\sadred{Individual Discontentment}:}
For individual discontentment, we note that by the cohesiveness-based example, the witness must include voter $1$ or $2$. Further, by local isomorphisms, the witness must contain at least one voter $i\in \{3,4\}$. However, then we could turn this voter into a fully satisfied one by letting her adapt voter $1$ or $2$'s preferences. In the resulting profile, the committee becomes priceable, despite the existence of a witness. This is the desired contradiction.
\wrapstuffclear

\begin{wrapstuff}[r,type=figure,width=5cm]
\centering
\begin{tikzpicture}
[yscale=0.6,xscale=0.8,
voter/.style={anchor=south}]

    \foreach \i in {1,...,4}
        \node[voter] at (\i-0.5, -1) {$\i$};

    \draw[fill=teal!\strstr, ultra thick] (0, 0) rectangle (1, 1);
    \draw[fill=teal!\strstr, ultra thick] (0, 1) rectangle (1, 2);
    \draw[fill=teal!\strstr, ultra thick] (1, 0) rectangle (3, 1);
    \draw[fill=magenta!\clrstr] (3, 0) rectangle (4, 1);

    \node at ( 0.5, 0.5) {$c_{1}$};
    \node at ( 0.5, 1.5) {$c_{2}$};
    \node at ( 2, 0.5) {$c_{3}$};
    \node at ( 3.5, 0.5) {$c_{4}$};
\end{tikzpicture}
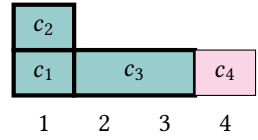
\caption{Example for priceability and merge-proofness.}
\label{fig:price_mp}
\end{wrapstuff}
\noindent\textbf{\sadred{Merge-Proofness}:}
To show that there is no witness for priceability satisfying merge-proofness, consider the instance depicted in \Cref{fig:price_mp} with the committee $\{c_1, c_2, c_3\}$. This committee is not priceable. One can verify that every possible witness, except for $\{1,2,3,4\}$ and $\{2,3,4\}$ can be locally embedded into another instance, committee pair, in which the committee is priceable. (E.g., the witness must contain voter $4$ as otherwise, we can locally embed the witness into an instance where voters only approve candidates from $W$. Further case distinctions concern whether $1$ is contained in the witness and what happens if the witness is only of cardinality $2$.)
For $\{1,2,3,4\}$ and $\{2,3,4\}$ we can let voter $4$ additionally approve $c_3$ by merge-proofness. However, then letting voter $4$ pay for $c_3$ and voter $1$ for $c_1$ and $c_2$ leads to a valid price-system with $B=8$. \qedhere

\wrapstuffclear
\end{proof}

\end{document}